%% file: shapley-aaai20-zhao.tex
\def\year{2020}\relax
\documentclass[letterpaper]{article} % DO NOT CHANGE THIS
\usepackage{aaai20}  % DO NOT CHANGE THIS
\usepackage{times}  % DO NOT CHANGE THIS
\usepackage{helvet} % DO NOT CHANGE THIS
\usepackage{courier}  % DO NOT CHANGE THIS
\usepackage[hyphens]{url}  % DO NOT CHANGE THIS
\usepackage{graphicx} % DO NOT CHANGE THIS
\usepackage{graphicx}  % DO NOT CHANGE THIS
\usepackage{amsfonts}
\usepackage{mathtools}
\usepackage{graphicx}
\usepackage{framed}
\usepackage{amsthm}
\usepackage{xcolor}
\usepackage{pgf}
\usepackage{subcaption}

\newtheorem{definition}{Definition}
\newtheorem{example}{Example}
\newtheorem{question}{Question}
\newtheorem{theorem}{Theorem}

\newtheorem{lemma}{Lemma}

\title{Coalitional Games with Stochastic Characteristic Functions and Private Types}%Incentive Compatible Reward Sharing Mechanisms for Project Supplier Assignment}
\author{Dengji Zhao,\textsuperscript{\rm 1} Yiqing Huang,\textsuperscript{\rm 1} Liat Cohen,\textsuperscript{\rm 2} Tal Grinshpoun\textsuperscript{\rm 2}\\
\textsuperscript{\rm 1}School of Information Science and Technology\\ 
ShanghaiTech University, China\\
\textsuperscript{\rm 2}Department of Computer Science\\ 
Ben-Gurion University of the Negev, Israel\\
\textsuperscript{\rm 3}Department of Industrial Engineering and Management\\ Ariel University, Israel
}
\begin{document}

\maketitle

\begin{abstract}
The research on coalitional games has focused on how to share the reward among a coalition such that players are incentivised to collaborate together. It assumes that the (deterministic or stochastic) characteristic function is known in advance. This paper studies a new setting (a task allocation problem) where the characteristic function is not known and it is controlled by some private information from the players. Hence, the challenge here is twofold: (i) incentivize players to reveal their private information truthfully, (ii) incentivize them to collaborate together. We show that existing reward distribution mechanisms or auctions cannot solve the challenge. Hence, we propose the very first mechanism for the problem from the perspective of both mechanism design and coalitional games. %We see this as a promising bridge connecting the two different fields for many interesting task allocation settings.
\end{abstract}

\section{Introduction}
Cooperative games have been used to model competitions between groups/coalitions of players~\cite{Davis1995kernel}. A cooperative game is defined by specifying a value for each coalition and it studies which coalitions will form and how they share the payoffs to enforce the collaboration. The value for each coalition represents the reward they can achieve together, e.g., finishing a task or building a social good. This value (the characteristic function) is often predefined and public, which can be either deterministic or stochastic. 
However, in real-world applications, what a coalition can achieve may depend on players' capabilities, which is not necessarily known to everyone in advance. %the value of a coalition depends on their capabilities which may contain some uncertainty that is only known to each player itself.

In this paper, we study a cooperative game where the characteristic function is controlled by some private information owned by the coalition. Specifically, we study a task allocation problem, where a group of players collaborate to accomplish a sequence of tasks in order and there is a deadline to finish all the tasks. 
%{\color{red}This models many applications such as constructing houses and manufacturers' product lines, where a task cannot be started until all the early tasks are finished.} 
Because of the deadline, we need to find the best set of players to do the tasks. Each player's capability is modelled by how much time she needs to finish a task. However, the finishing time is not fixed and it is a random variable following some distribution, which is only known to the player. Because of the uncertainty, the objective is to maximize the probability to meet the deadline~\cite{frank1969shortest}. In order to find the best set of players to meet the objective, we need to first know their private distributions (which then defines the probability to meet the deadline for each coalition). That is, the characteristic function is defined by the private distributions of the players in a coalition.

We model the problem as a coalitional game, but the value for each coalition (the characteristic function) is controlled by the players' private information, which has never been investigated in the literature. On one hand, we want each player to tell their true private information to make the best decision, and on the other hand, we want the reward to be fairly distributed among the players. %It contains both the private information revelation challenge and the payoff distribution callenge. 
Cooperative games are often used to take care of the reward/cost distribution to enforce collaboration, while mechanism design is good at private information elicitation in a competitive environment~\cite{nisan2007algorithmic}. We will show that our challenge cannot be solved by just using techniques from one side. %Can we merge them together to solve both challenges in one go?  %since the focus in coalitional games is usually on reward distribution, without relating to incentive compatibility. For instance, computation of the \emph{Shapley value}~\cite{shapley1953value} requires knowledge of the agents' capabilities in order to compute their marginal contributions.

Thus, our goal is to design new reward sharing mechanisms such that players are incentivized to report their private distributions truthfully and the reward are distributed fairly among all players. %We see this problem as an interesting bridge connecting cooperative game theory and mechanism design. 
To combat this problem, we propose a novel mechanism to solve both challenges using the techniques from both cooperative game theory and mechanism design. Our solution is based on a modified Shapley value which distributes the reward according to the players' capabilities which incentivizes all players to reveal their true capabilities. The cost to achieve the goal is that the total reward distributed is more than what the coalition can generate. For comparison, we also studied a solution from a non-cooperative perspective by using Vickrey-Clarke-Groves (VCG) mechanism~\cite{vickrey1961counterspeculation,clarke1971multipart,groves1973incentives}. VCG is very good at handling private information revelation, but the reward distribution is not as fair as the Shapley value. 

\subsection{Related Work}
As closely related work, coalitional games with stochastic characteristic functions have been a rich line of research initiated by Charnes and Granot~\shortcite{Charnes1976,Charnes1977}. They assumed that the value for each coalition is a random variable following some known/public distributions, so the challenge is about how to promise a payoff for the grand coalition, such that the coalition is still enforceable and it does not over pay if the random outcome is not good. Their setting does not formally model where the randomness comes from. In our model, since each player has a random task completion time, together, we can say that the total completion time for a coalition is a random variable controlled by the players. We focus on how to incentivize the players to reveal their private random variable. Bachrach \textit{et al.}~\shortcite{Bachrach2012} studied the network reliability problem, where an edge has a probability to fail for connecting two points in a network, but the probability is public, which is essentially a special case of \cite{Charnes1976,Charnes1977}. Their focus is to compute the Shapley value efficiently. 

There are also many other studies on coalitional games with imcomplete information from different perspectives. Chalkiadakis and Boutilier~\shortcite{Chalkiadakis2004,Chalkiadakis2012,Chalkiadakis2007} took a learning/bargaining approach to learn the types of the other players in a repeated game form, while we directly ask them to report their types. Li and Conitzer~\shortcite{Li2015} studied the least core concept in a setting where the characteristic function has some well-defined noise (random effect), which is not controlled by any player.  Ieong ad Shoham~\shortcite{Ieong2008} considered the game where each player is uncertain about which game they are playing as a Bayesian game with information partition. On top of the stochastic characteristic function modelled by \cite{Charnes1976},  Suijs \textit{et al.}~\shortcite{Suijs1999,Suijs1999GEB} extended the setting to the case where each coalition has some actions to choose and each action leads to a different random payoff. Myeson~\shortcite{MYERSON2007} further considered that each player has multiple types with some distribution and each player also knows the type distributions of the others, then the paper investigated mechanisms with side payments to incentivize players to reveal their true types. However, when the players' types are given, the payoff they can generate is not stochastic. Our model requires each player to report her private type (which is a distribution) and given their types, the value they can generate together is also stochastic.

Our solution has also applied type verification for their task execution time. \citeauthor{Nisan2001}~\shortcite{Nisan2001} studied a task scheduling problem where each agent first declares a time she needs to finish a task and later when the task is allocated and executed, the mechanism is able to verify how much time the agent actually took. By doing so, the mechanism can verify what the agent reported and pay the agent according to the execution, which is a direct verification of the agent's report. In our setting, we cannot directly verify a player's time distribution, but we are still able to see how much time the execution takes, which is a partial verification of the distribution. This partial verification has been applied first by \citeauthor{porter2008fault}~\shortcite{porter2008fault} in a single task allocation where each worker has a (private) probability to fail the task. Inspired by their solution, different extensions have been investigated~\cite{Ramchurn2009Trust,stein2011algorithms,conitzer2014mechanism,zhao2016fault}.
%~\cite{porter2008fault,Ramchurn2009Trust,conitzer2014mechanism,stein2011algorithms,zhao2016fault}. 
%For instance, \citeauthor{porter2008fault}~\shortcite{porter2008fault} first looked at the uncertainty to successfully finish a single task (so the outcome is either success or failure), which has been extended to more complex settings~\cite{Ramchurn2009Trust,zhao2016fault}. Stein \textit{et al.}~\shortcite{stein2011algorithms} considered the uncertain duration to finish a single task before a deadline by allocating multiple workers for one task to maximize the chance. Conizter and Vidali~\shortcite{conitzer2014mechanism} also studied uncertain duration time to finish a task without deadline under a monotone hazard rate condition. 
All these settings are studied from a non-cooperative perspective and mostly for a single task, while our setting has multiple interdependent tasks and looks at the cooperative perspective.

%Alternatively, the problem may be modeled as an auction in order to ensure truthful reports by the agents through mechanism design. However, auction mechanisms, such as VCG~\cite{} {\color{red}Dengji -- what is the common reference to VCG?} do not deal with issues of fair reward distribution, which is a core element in our problem.

%ignores several important real-life considerations, such as 

\section{The Model}
We investigate a coalitional game where the characteristic function is not given in advance. It is defined by the capabilities of the players, and each player's capability is a private information of the player. This setting exists in many real-world task allocation problems. 

%In this paper, we consider a task allocation problem, where there is a sequence of $m$ different tasks $T = (\tau_1, \cdots,  \tau_m)$ to be finished in order, i.e. $\tau_i$ cannot be started until all tasks before $\tau_i$ have been finished. There is a deadline $d$ to finish all the tasks. Finishing all tasks before the deadline generates a value $V$, and finishing the tasks after the deadline does not generate any value. 

We consider herein a task allocation problem. %that is based on the model of \cite{cohen2019assigning}. 
The problem is of a project that consists of a sequence of $m$ different tasks $T = (\tau_1, \cdots,  \tau_m)$ to be finished in order, i.e. $\tau_i$ cannot be started until all tasks before $\tau_i$ have been finished. There is a deadline $d$ to finish the entire project. Finishing all tasks before the deadline generates a value $V$, and finishing the tasks after the deadline does not generate any value.

There are $n$ agents (players) denoted by $N= \{1, \cdots, n\}$ who can perform the tasks with different capabilities. Without loss of generality, we assume that each player is only capable of doing one of the tasks. Let $N_{\tau_i} \subseteq N$ be the set of players who can handle task $\tau_i$. We have $N_{\tau_i} \neq \emptyset$ for all $\tau_i\in T$, $N_{\tau_i} \cap N_{\tau_j} = \emptyset$ for all $\tau_i \neq \tau_j\in T$, and $\cup_{\tau_i \in T} N_{\tau_i} = N$.

For each player $i\in N_{\tau_j}$, her capability to handle task $\tau_j$ is measured by the execution time she needs to accomplish $\tau_j$. We also consider that there is some uncertainty for player $i$ to accomplish the task. Therefore, player $i$ does not know the exact time she will need to finish $\tau_j$, but she does know a duration distribution. 

Studies on various task allocation problems assume private information of the players, regarding either probability of success~\cite{porter2008fault,zhao2016fault} or task duration~\cite{stein2011algorithms,conitzer2014mechanism}. Similarly, in our coalitional game the execution time distribution is $i$'s private information to determine the characteristic function. Let the discrete random variable $E_i$ denote the execution time of player $i$ on task $\tau_j$. We use $e_i \in \{1, 2, \cdots\}$ to denote a realization of $E_i$. Let $f_i$ be the probability mass function of $E_i$, i.e., $Pr[E_i = e_i] = f_i(e_i)$. There might be a cost $c_i$ for $i$ to execute task $\tau_j$. We assume the cost is public and it can be ignored for the current analysis.

The goal of the above coalitional game is to find the optimal task allocation (one task can only be allocated to at most one player) such that the tasks $T$ can be finished before deadline $d$ with the highest probability. This would generate the highest expected value/reward for the players. That is, the characteristic function $v: 2^N \rightarrow \mathbb{R}$ can be defined by 
\begin{quote}
$v(S)$ equals the highest probability that a group of players $S\subseteq N$ can finish all the tasks $T$ before $d$.
\end{quote}
It is evident that the definition of $v$ satisfies $v(\emptyset) = 0$ and the \emph{monotonicity} property, i.e., $v(S) \leq v(T) \leq 1$ for all $S \subseteq T \subseteq N$. $v(S)\times V$ is the expected value the coalition $S$ can cooperate to achieve.

Since duration distributions are private and $v(S)$ is not publicly known, we cannot easily achieve the goal and share the reward among the players with the standard techniques for coalitional games. The challenge here is that players can manipulate the game by misreporting their time distributions, which is not available in classical coalitional games. 

%{\color{red}Dengji -- I moved this paragraph to the end of the introduction. Leaving it as is will also look bad (because it already appeared earlier. I propose removing it (the model section is anyway not the perfect place to set the paper's objectives -- this is why I moved it to the introduction in the first place). This would also save some space that we need... What do you think?}
\emph{The goal of this paper is to design new reward sharing mechanisms for the above game such that players are incentivized to report their time distributions truthfully. %We see this problem as a nice bridge between coalitional game and mechanism design.
}

The reward sharing mechanism requires each player to report her execution time distribution, but the player may not necessarily report her true distribution. For each player $i\in N$, let $f_i$ be the density function of her true distribution, and $f_i^\prime$ be her report. Let $f=(f_1, \cdots, f_n)$ be the true density function profile of all players and $f^\prime=(f_1^\prime, \cdots, f_n^\prime)$ be their report profile. We also denote $f$ by $(f_i, f_{-i})$ and  $f^\prime$ by $(f_i^\prime, f_{-i}^\prime)$. Let $\mathcal{F}_i$ be the density function space of $f_i$ and $\mathcal{F}=(\mathcal{F}_1, \cdots, \mathcal{F}_n)$ be the space of density function profile $f$.

\begin{definition}
A \textbf{reward sharing mechanism} is defined by $x = (x_i)_{i\in N}$, where $x_i: \mathcal{F} \rightarrow \mathbb{R}$ defines the reward player $i$ receives given all players' report profile. %under all the players' time distribution report profile.
\end{definition}

To incentivize players to report their private information truthfully, a concept called \textbf{incentive compatible} has been defined in the literature of mechanism design. We apply the same concept here to define the incentives to report their time distributions truthfully. We say a reward sharing mechanism is incentive compatible if for each player reporting her time distribution truthfully is a dominant strategy.

\begin{definition}
A reward sharing mechanism $x$ is \textbf{incentive compatible} if for all players $i\in N$, for all $f\in \mathcal{F}$, for all $f^\prime\in \mathcal{F}_i$, we have $x_i(f) \geq x_i(f_i^\prime, f_{-i})$.
\end{definition}

In the rest of the paper, we study incentive compatible reward sharing mechanisms.

\section{The Failure of the Shapley Value}
Shapley value is a well-known solution concept in cooperative game theory~\cite{shapley1953value}. It divides the reward among the players in a coalition according to their marginal contributions. It has many desirable properties such as \emph{efficiency} (the reward is fully distributed to the players), \emph{symmetry} (equal players receive equal rewards) and \emph{null player} (dummy players receive no reward).

If we simply apply the Shapley value in our setting, the reward for each player is defined as:
\begin{equation}
\scalebox{0.85}{
$x_i^{sha}(f^\prime) =\displaystyle \sum_{S\subseteq N\setminus \{i\}} \dfrac{|S|!(n-|S|-1)!}{n!}(v(S\cup \{i\}) - v(S))
$}
\end{equation}
where $v(S\cup \{i\})$ and $v(S)$ are defined under the players' report profile $f^\prime$.

The Shapley value of a player $i$ is the average marginal contribution among all permutations of the players. In each permutation, player $i$'s marginal contribution is $v(S\cup \{i\}) - v(S)$, where $S$ is the set of all players before $i$ in the permutation.

{\color{black}In this paper, we assume that in the grand coalition, for each task group, the player who is assigned the task has the highest Shapley value among the same task group. This induces some conditions on the time distributions, which is not clear what the exact conditions are. Intuitively, it implies if a player is better than the others in the same task group in the grand coalition, then it should also be better than them in sub-coalitions.}

Let us consider a simple example:
\begin{quote}
\begin{example}
\label{ex1}
There are two tasks $T=(\tau_1, \tau_2)$ to be finished before deadline $2$ and four players $N = \{1,2,3,4\}$ with $N_{\tau_1}=\{1,2\}$ and $N_{\tau_2}=\{3,4\}$. Their execution time density functions are:
\begin{equation}
f_1(e) = f_3(e) = 
\begin{cases}
    3/4      & \quad \text{if } e = 1\\
    1/4    & \quad \text{if } e = 3\\
    0       & \quad \text{otherwise}
  \end{cases}
\end{equation}
\begin{equation}
f_2(e) = f_4(e) = 
\begin{cases}
    1/4      & \quad \text{if } e = 1\\
    3/4    & \quad \text{if } e = 2\\
    0       & \quad \text{otherwise}
  \end{cases}
\end{equation}

If all players report their density functions truthfully, then the allocation to maximize the probability to finish both $\tau_1$ and $\tau_2$ is to assign $\tau_1$ to player $1$ and $\tau_2$ to player $3$. The probability is $9/16$, i.e., $v(N) = 9/16$.
\end{example}
\end{quote}

The characteristic function in Example~\ref{ex1} is defined as: $v(\{1,3\}) = v(\{1,2,3\}) = v(\{1,3,4\}) = v(\{1,2,3,4\}) = 9/16$, $v(\{1,4\})=v(\{1,2,4\})=3/16$, $v(\{2,3\})=v(\{2,3,4\})=3/16$, $v(\{2,4\})=1/16$, and the value for the rest coalitions is zero.

To compute the Shapley value for player $1$, we list out all the permutations of the players and check player $1$'s marginal contribution in each permutation in the following table. The average marginal contribution of player $1$ among all permutations is $47/192 \approx 0.2448$.

\begin{center}
  \begin{tabular}{ c | c | c | c }
    \hline
    order & \# permut. & S & $v(S\cup \{1\}) - v(S)$ \\ \hline
	(1,\{2,3,4\}) & 6 & \{\} &  0 \\ \hline
    (2,1,\{3,4\}) & 2 & \{2\} & 0 \\ \hline
    (3,1,\{2,4\}) & 2 & \{3\} & 9/16 \\ \hline
    (4,1,\{2,3\}) & 2 & \{4\} &  3/16 \\ \hline
    (\{3,4\},1,2) & 2 & \{3,4\} &  9/16 \\ \hline
    (\{2,4\},1,3) & 2 & \{2,4\} &  2/16 \\ \hline
    (\{2,3\},1,4) & 2 & \{2,3\} &  6/16 \\ \hline
    (\{2,3,4\},1) & 6 & \{2,3,4\} & 6/16 \\ \hline
  \end{tabular}
\end{center}

Similarly, we get the Shapley value for all players:
\begin{center}
  \begin{tabular}{ c | r  }
    \hline
    player & $x_i^{sha}(f)$ \\ \hline
	1 & 47/192 \\ \hline
	2 & 7/192 \\ \hline
	3 & 47/192 \\ \hline
	4 & 7/192 \\ \hline
  \end{tabular}
\end{center}

\begin{question}
\label{q1}
In Example~\ref{ex1}, only players $1$ and $3$ actually perform the tasks, so can we reward only players $1$ and $3$ according to their Shapley value?
\end{question}
The answer is no. For example, player $2$ could misreport a density function $f_2^\prime$ such that $v(\{2,3\}) > v(\{1,3\})$, in order to be selected for $\tau_1$ and receive a non-zero Shapley value. Therefore, we cannot exclude rewards to players who have not been assigned a task in the reward sharing mechanism.

\begin{question}
\label{q2}
In Example~\ref{ex1}, can any of the players misreport to gain a higher Shapley value?
\end{question}
The answer is yes. For instance, player $1$ can misreport $f_1^\prime$ such that
\begin{equation}
f_1^\prime(e) = 
\begin{cases}
    1      & \quad \text{if } e = 1\\
    0       & \quad \text{otherwise}
  \end{cases}
\end{equation}
If the other players still report truthfully, player $1$'s Shapley value under this misreport is $67/192$, which is larger than the value $47/192$ under report $f_1$. {\color{black}This kind of misreport applies to all players.}
The reason is that their Shapley value only depends on what they have reported, not what they can actually do.

In next section, we show how to modify the Shapley value to link it to the players' true time distributions.
%\begin{question}
%\label{q3}
%Can we modify the Shapley value such that their reward indeed depends on their true execution time distributions and they are not incentivized to misreport?
%\end{question}
%The answer is yes. This is the focus of the next section.

\section{Truthful Shapley Value}
As evident from Question~\ref{q2}, players can report a more promising execution time distribution to receive a higher Shapley value. This is because the Shapley value mechanism never verifies their reports. In reality, we could actually observe how much time a player has spent to accomplish her task. Therefore, we can pay them according to their execution outcomes. A similar approach has been applied in other task allocation settings by using auctions, especially VCG mechanism~\cite{porter2008fault,Ramchurn2009Trust,conitzer2014mechanism,stein2011algorithms,zhao2016fault}. 

\begin{definition}
Given all players' execution time distribution report profile $f^\prime = (f_1^\prime, \cdots, f_n^\prime)$, for each coalition $S\subseteq N$, $\pi_S^{f^\prime}: T \rightarrow N \cup \{\perp\}$ is the \textbf{task assignment} to define $v(S)$. $\pi_S^{f^\prime}(\tau_j) = \perp$ means that $\tau_j$ has not been assigned to any player under coalition $S$ with reports $f^\prime$. 
\end{definition}

Next are some specific notations for the new mechanism. 
\begin{itemize}
\item Let $v(S, f^\prime)$ be the highest probability to finish all the tasks before the deadline under the report profile $f^\prime$.
\item Let $v(\pi_S^{f^\prime}, f^{\prime\prime})$ 
%{\color{red} Dengji -- Why do you use $f^{\prime\prime}$ when simple $f$ denotes the true density function? (according to the Model section)} {\color{blue} not necessarily true distributions, I changed the wording} 
be the probability to finish all the tasks before the deadline given that the optimal task assignment is defined by $\pi_S^{f^\prime}$ but the actual probability to finish the tasks is calculated by $f^{\prime\prime}$. That is, $v(\pi_S^{f^\prime}, f^{\prime\prime})$ uses $f^\prime$ to determine the optimal task assignment under coalition $S$, but uses $f^{\prime\prime}$ to recalculate the probability without changing the task assignment. In the mechanism, $f^\prime$ represents their reports and $f^{\prime\prime}$ represents what we observed.
\end{itemize}

\begin{framed}
	\noindent\textbf{Shapley Value with Execution Verification (SEV)}\\
	\rule{\textwidth}{0.5pt}
Given all players' report profile $f^\prime$, 

\begin{itemize}
\item for each player $i$ who has been assigned a task under $\pi_N^{f^\prime}$, if her realised execution time is $e_i$, then her Shapley value is updated as:
  
\begin{equation}
\label{eq_sev1}
\scalebox{0.69}{
$x_i^{sev}(f^\prime, e_i)=\displaystyle \sum_{S\subseteq N\setminus \{i\}} \dfrac{|S|!(n-|S|-1)!}{n!}(v(\pi_{S\cup \{i\}}^{f^\prime}, (f_i^{e_i},f_{-i}^\prime)) - v(S, f^\prime))
$}
\end{equation}
where $f_i^{e_i}$ represents the realization $e_i$ and is defined as:
\begin{equation*}
f_i^{e_i}(e) = 
\begin{cases}
    1      & \quad \text{if } e = e_i\\
    0       & \quad \text{otherwise}
  \end{cases}
\end{equation*}

\item for each player $j$ who has not been assigned any task in the assignment $\pi_N^{f^\prime}$, her Shapley value stays the same (as we cannot observe $j$'s execution time):
\begin{equation}
\label{eq_sev2}
\scalebox{0.69}{
$x_j^{sev}(f^\prime) = \displaystyle \sum_{S\subseteq N\setminus \{j\}} \dfrac{|S|!(n-|S|-1)!}{n!}(v(S\cup \{j\}, f^\prime) - v(S, f^\prime))
$}
\end{equation}
\end{itemize}
\end{framed}

\begin{theorem}
\label{thm_sevic}
The SEV mechanism is incentive compatible for all players who are assigned a task, but it is not incentive compatible for players who are not assigned a task.
\end{theorem}
\begin{proof}
For each player $i$ who has been assigned a task, $i$'s reward varies according her execution outcomes, but her expected reward is $E[x_i^{sev}(f^\prime, e_i)] = \sum_{e_i\in E_i} f_i(e_i)x_i^{sev}(f^\prime, e_i)$. If $i$ reports her true distribution, i.e., $f_i^\prime= f_i$, $E[x_i^{sev}(f^\prime, e_i)]$ equals her Shapley value $x_i^{sha}(f^\prime)$. Now if $f_i^\prime \neq f_i$, can $i$ receive a larger $E[x_i^{sev}(f^\prime, e_i)]$?

No matter what $i$ reports, $i$ is either assigned or not assigned the task. If $i$ is assigned the task when $i$ reports $f_i$, then we show that her expected reward is maximized when $i$ reports $f_i$. For the expected reward, we could look the expected reward (marginal contribution) in each player permutation. For each permutation, assume $S\subset N$ is the set of players before $i$:
\begin{itemize}
\item if $i$ has a non-zero marginal contribution, i.e., $v(S\cup \{i\}) - v(S) > 0$, it means that $i$ is assigned the task in the coalition $S\cup \{i\}$. Then $i$'s expected marginal contribution would be $\sum_{e_i\in E_i} f_i(e_i)(v(S\cup \{i\}, e_i) - v(S))$, where $v(S\cup \{i\}, e_i)$ is the probability to finish all the tasks when $i$'s execution time is fixed to $e_i$ without changing the task assignment defined by $\pi_{S\cup \{i\}}^{f^\prime}$. If $i$ reports $f_i$, then $\sum_{e_i\in E_i} f_i(e_i)v(S\cup \{i\}, e_i) = v(S\cup \{i\})$. If $i$ reports $f_i^\prime$ which dominates $f_i$, then $i$ is still assigned the task, but $\sum_{e_i\in E_i} f_i(e_i)v(S\cup \{i\}, e_i)$ stays the same (although in this case $v(S\cup \{i\})$ might be increased). If $i$ reports $f_i^\prime$ such that $i$ is not assigned the task under the coalition $S\cup \{i\}$, then $i$'s marginal contribution becomes zero. Therefore, reporting $f_i$ truthfully maximizes $i$'s expected marginal contribution in this permutation.
\item if $i$ has a zero marginal contribution, i.e., $v(S\cup \{i\}) = v(S)$, it means that $i$ is not assigned the task in the coalition $S\cup \{i\}$. Thus, $i$'s expected marginal contribution in this case is zero. If $i$ reports $f_i^\prime$ such that $i$ is assigned the task under the coalition $S\cup \{i\}$, then the expected probability $\sum_{e_i\in E_i} f_i(e_i)v(S\cup \{i\}, e_i)$ would be less than $v(S\cup \{i\})$. This is because $i$ takes a task from another player and makes the allocation non-optimal. Therefore, $i$'s expected marginal contribution in this case is negative. Hence, reporting $f_i$ also maximizes $i$'s expected marginal contribution in this case.
\end{itemize}
Since the expected marginal contribution in each permutation is maximized when $i$ reports truthfully, then her total expected reward is also maximized in this case.

For each player $j$ who is not assigned a task, we can easily find a counter example where $j$ can misreport to gain a higher reward. We can always find a setting and a permutation where $j$'s marginal contribution is zero. For this permutation, we can change $j$'s distribution such that $j$'s marginal contribution is greater than zero, but $j$ is still not assigned a task in the grand coalition $N$. By doing so, the mechanism cannot observe $j$'s execution outcome and simply pays her the standard Shapley value which is greater than what she could have when she reports $f_j$ truthfully. 
\end{proof}

As seen in Theorem~\ref{thm_sevic}, players who were not assigned a task can misreport to gain a higher reward under the SEV mechanism, because there is a lack of verification on their reports. It is certainly not ideal to assign each task to all  its players to try, which is also not practical. Instead, we do the manipulation on behalf of the players to maximize their rewards they could gain. For each player who is not assigned the task, we treat this player as good as the player who is assigned the task to calculate her new Shapley value as her reward. This reward is the best the player could get by misreporting, and therefore, there is not incentive for misreporting anymore. The updated mechanism is defined as follows.

\begin{framed}
	\noindent\textbf{Shapley Value with Execution Verification and Bonus (SEVB)}\\
	\rule{\textwidth}{0.5pt}
Given all players' report profile $f^\prime$, 

\begin{itemize}
\item for each player $i$ who is assigned a task under $\pi_N^{f^\prime}$, if her realised execution time is $e_i$, then her Shapley value is defined as the same as in SEV (Equation~\eqref{eq_sev1}), i.e., $x_i^{sevb}(f^\prime, e_i) = x_i^{sev}(f^\prime, e_i)$.

\item for each player $j$ (assume $j\in N_{\tau_i}$) who is not assigned the task $\tau_i$ in $\pi_N^{f^\prime}$, her Shapley value is upgraded as $x_j^{sevb}(f^\prime) =$

\begin{equation}
\label{eq_sevb}
\scalebox{0.66}{
$\displaystyle \sum_{S\subseteq N\setminus \{j\}} \dfrac{|S|!(n-|S|-1)!}{n!}(v(S\cup \{j\}, (f_j^*, f_{-j}^\prime)) - v(S, (f_j^*, f_{-j}^\prime)))
$}
\end{equation}
where $f_j^* = f_{i^*}^\prime$ and $i^* = \pi_N^{f^\prime}(\tau_i)$, i.e., $i^*$ is the player who is assigned $\tau_i$.
\end{itemize}
\end{framed}

\begin{theorem}
The SEVB mechanism is incentive compatible.% for all players.
\end{theorem}
\begin{proof}
For each player $i$ who is assigned a task, as proved in Theorem~\ref{thm_sevic}, there is no incentive to misreport when the players who are not assigned a task are paid according to Equation~\ref{eq_sev2}. Now the unassigned players' reward has upgraded to Equation~\ref{eq_sevb}, we need to prove that it is not better for $i$ to misreport such that $i$ is not assigned the task in the grand coalition.

Assume that $i$ is assigned task $\tau_i$, if $i$ misreports $f_i^\prime$ such that $\tau_i$ is assigned to $j$, then $i$'s reward will the upgraded Shapley value when we treat $i$'s report as identical as $j$'s report. Initially, we know that $j$'s report is not better than $i$'s (otherwise, $i$ would not be assigned). Now $i$'s misreport $f_i^\prime$ is not as good as $j$'s. We can prove that $i$'s reward under $f_i^\prime$ is not better than under $f_i$. This can be proved from each single permutation. 
\begin{itemize}
\item If $i$ has a non-zero marginal contribution in a permutation when $i$ reports truthfully (assume $S$ are the players before $i$), when $i$ reports $f_i^\prime$, $i$ is either assigned or not assigned the task. If $i$ is still assigned the task, then her marginal contribution is calculated as when treat $i$ as $j$, which is clearly not better than $i$ reports truthfully. If $i$ is not assigned the task, then her marginal contribution becomes zero, which is again worse than reporting truthfully.  
\item If $i$ has zero marginal contribution in a permutation when $i$ reports truthfully, then the contribution stays the same if she reports $f_i^\prime$.
\end{itemize}
In summary, we get that player $i$'s reward is not increased when she misreports $f_i^\prime$.

For player $j$ who is not assigned a task, assume $j$ belongs to task group $\tau_i$ and $i$ is assigned $\tau_i$ when $j$ reports $f_j$:
\begin{itemize}
\item if $j$ misreports $f_j^\prime$ such that $f_j^\prime$ is not dominating $i$'s report $f_i^\prime$, then $j$ is still not assigned task $\tau_i$ and still receives the same reward. 
\item if $f_j^\prime$ dominates $f_i^\prime$, then $j$ will be assigned $\tau_i$. We will prove that $j$'s expected reward under this case is not better than reporting $f_j$. 
\end{itemize}

When $f_j^\prime$ dominates $f_i^\prime$, $j$ is assigned $\tau_j$ in the grand coalition and will execute task $\tau_i$. For each permutation where $j$'s marginal contribution is non-zero under $f_j^\prime$, then under $f_j$, $j$'s marginal contribution can be either zero or non-zero. If it is zero, then her expected reward in this permutation will be negative, because $j$ is forced to change the optimal allocation, but could not increase the value. If it is non-zero, then her expected reward stays the same, as the reward is calculated according to her execution outcomes (her true $f_j$).

Together, we proved that $j$ cannot receive a better reward by misreporting.
\end{proof}

Since the SEVB mechanism may pay more than their actual Shapley values for players who did not receive a task, the total payment together might be greater than the total reward the grand coalition can gain. Theorem~\ref{thm_bound2} shows that the extra payment is bounded.

\begin{theorem}
\label{thm_bound2}
Given any execution time distribution report profile $f$, the total reward distributed under the SEVB mechanism is bounded by $(1+\sum_{i=1}^{k}\dfrac{i!(m-1)!}{(m+i)!}(n-m-k+1))v(N, f)$, where integer $k\geq 1$ is a parameter to maximize the bound given $n$ and $m$.
\end{theorem}

{\color{black}To prove Theorem~\ref{thm_bound2}, we first need to show some key properties of the Shapley value in our setting.}

\begin{figure}[t]
  \centering
      \includegraphics[width=0.2\textwidth]{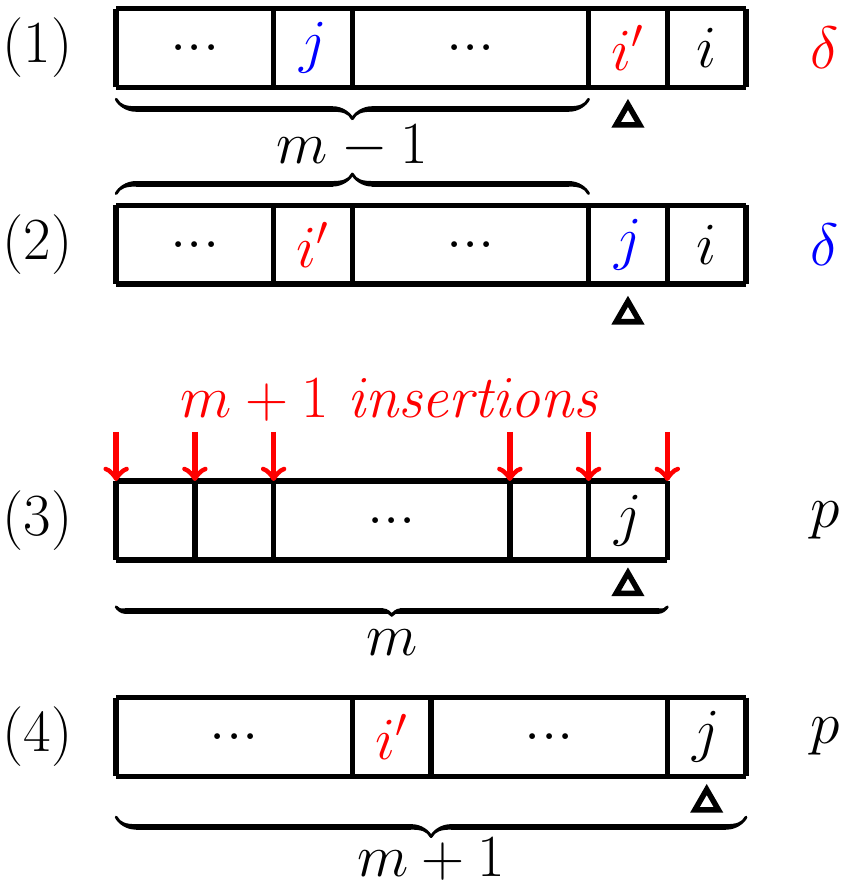}
  \caption{For proof of Lemma~\ref{lem_bound2}\label{fig_0}}
\end{figure}

%tal
\begin{lemma}
\label{lem_bound2}
Given that there are $m$ players each of whom can do one of the $m$ tasks and together can finish the tasks before the deadline with a probability $p$, assume that player $i$ is assigned $\tau_i$. Also consider that we add one more player $i^\prime$ for task $\tau_i$, such that $i^\prime$ will not be assigned task $\tau_i$ when the $m+1$ players collaborate together. If player $i^\prime$'s Shapley value in the grand coalition is $\epsilon$, then the Shapley value for $i$ is $\dfrac{p}{m}-m\epsilon$ and the Shapley value for the other players is $\dfrac{p}{m}+\epsilon$. Also, the total reward distributed under the SEVB mechanism for the $m+1$ players is maximized when $i^\prime$ is a dummy player, i.e., $\epsilon = 0$.
\end{lemma}
\begin{proof}
Under the computation of the Shapley value, for each permutation of the $m+1$ players, if player $i^\prime$'s marginal contribution is $\delta >0$, then player $i$ must be the last player and $i^\prime$ is the second last player in the permutation (otherwise, $i^\prime$'s marginal contribution is zero, see Figure~\ref{fig_0} permutation $(1)$). If we simply switch player $i^\prime$ with any player $j$ before $i^\prime$, under this new permutation, $j$'s marginal contribution is also $\delta$ (Figure~\ref{fig_0} permutation $(2)$). Clearly, if we change the order of all players before $i^\prime$, we get $(m-1)!$ different permutations and in each permutation $i^\prime$'s marginal contribution is also $\delta$. Similarly for player $j$, we also get $(m-1)!$ different permutations where $j$'s marginal contribution is $\delta$.

Before the addition of player $i^\prime$, for each permutation of the $m$ players, only the last player has a marginal contribution $p$. Therefore, for each player $j$, we have $(m-1)!$ different permutations where $j$ is the last player with marginal contribution $p$. We have in total $m!$ permutations, hence $j$'s Shapley value is $\dfrac{p}{m}$.

After adding player $i^\prime$, for each original permutation of the $m$ players, we can insert $i^\prime$ in-between the $m$ players. We have $m+1$ possible insertions for $i^\prime$, where each such insertion creates a new permutation, in which the marginal contribution of the last player $j\neq i$ (last according the original permutation) remains $p$ among all the $m+1$ new permutations (Figure~\ref{fig_0} permutations $(3)$ and $(4)$). The addition of a player increases the total number of permutations from $m!$ to $(m+1)!$, however, the number of permutations where the marginal contribution of each player $j\neq i$ is $p$ also increases from $(m-1)!$ to $(m+1)(m-1)!$. Therefore, the Shapley values of all players except $i$ (with $m$ players) are transferred to the setting of $m+1$ players. In addition, as depicted in Figure~\ref{fig_0} permutations $(1)$ and $(2)$,
player $j\neq i$ also receives whatever player $i^\prime$ has in terms of Shapley value. Thus, if player $i^\prime$'s Shapley value is $\epsilon$, then the same value is also added to player $j$'s Shapley value. Hence, the new Shapley value for player $j\neq i$ is $\dfrac{p}{m} + \epsilon$. Since the sum of all the players' Shapley value is $p$, we get that player $i$'s Shapley value is $\dfrac{p}{m} - m\epsilon$.

Under the SEVB mechanism, all players except $i^\prime$ receive a reward equal to their Shapley value. For player $i^\prime$, her reward under SEVB is $\dfrac{p}{(m+1)m}$. This is computed by assuming that $i^\prime$ has the same distribution as $i$. Since $i$ and $i^\prime$ are identical, their Shapley value is the same. Now for each permutation, $i$'s marginal contribution is $p$ if and only if $i^\prime$ is the last and $i$ is the second last in the permutation. Thus, the total number of permutations where $i$'s marginal contribution is $p$ is $(m-1)!$, leading to a Shapley value of $\dfrac{p}{(m+1)m}$ for players $i$ and $i^\prime$. Hence, the total reward distributed under SEVB is 
$$ (\dfrac{p}{m} - m\epsilon) +  (\dfrac{p}{m} + \epsilon)(m-1) + \dfrac{p}{(m+1)m}.$$ To simplify it, we get $(1+\dfrac{1}{(m+1)m})p - \epsilon$, which is maximized when $\epsilon = 0$.
\end{proof}

\begin{lemma}
\label{lem_bound3}
Given $n\geq m+1$ players for $m$ tasks, assume that there are $n-(m-1)$ players for task $\tau_i$ and one player for each of the other $m-1$ tasks and the $n$ players together can finish the $m$ tasks before the deadline with probability $p$ where $\tau_i$ is allocated to player $i$. Then the total reward distributed under the SEVB mechanism is maximized when all players $j\neq i$ for $\tau_i$ are dummy players, i.e., their Shapley value under the grand coalition is zero. 
\end{lemma}
%\begin{proof}
%From Lemma~\ref{lem_bound2}, we know that when $n=m+1$, this lemma holds. Now if $n=m+2$, we need to add one more player for task $\tau_i$. Let us call the three players for $\tau_i$ $i, i^{d1}$ and $i^{d2}$ and $i$ is the player assigned task $\tau_i$ and $i^{d1}$ is a dummy player and $i^{d2}$ is the newly added player. Since $i^{d1}$ is a dummy player is dummy player, it is Shapley value is zero. If the Shapley value for $i^{d2}$ is $\epsilon$, then following the proof of Lemma~\ref{lem_bound2}, we get that the Shapley value for all players $j\not\in \{i, i^{d1}, i^{d2}\}$ is $\dfrac{p}{m}+\epsilon$ and the Shapley value for $i$ is $\dfrac{p}{m} - m\epsilon$, which are also their reward distributed under the SEVB mechanism. The reward for $i^{d2}$ under the SEVB mechanism is $\dfrac{p}{(m+1)m}$, but the reward for $i^{d1}$ is $\dfrac{p}{(m+1)m} - \delta$ (because $i^{d2}$ creates competition for $i^{d1}$, which decreases $i^{d1}$'s reward). Therefore, the total reward distributed under the SEVB mechanism is 
%$$(\dfrac{p}{m} - m\epsilon) + (\dfrac{p}{m}+\epsilon)(m-1) + 2\dfrac{p}{(m+1)m} - \delta$$
%which is $(1+\dfrac{2}{(m+1)m})p - \epsilon - \delta$. This is maximised when $\epsilon = \delta = 0$. Following this, we can prove that for any $n \geq m+1$, the lemma holds.
%\end{proof}

\begin{lemma}
\label{lem_bound4}
Given $n\geq m+1$ players for $m$ tasks, assume that there are $k\geq 1$ identical players for task $\tau_j$, $n-k-(m-2)$ players for task $\tau_i \neq \tau_j$ (where all players except for one of them are dummy players), one player for each of the other $m-2$ tasks and the $n$ players together can finish the $m$ tasks before the deadline with a probability $p$. Then the total reward distributed under the SEVB mechanism is $(1+\sum_{i=1}^{k}\dfrac{i!(m-1)!}{(m+i)!}(n-m-k+1))p$. 
\end{lemma}

%{\color{black}Due to the space limit and the proofs of Lemmas~\ref{lem_bound3} and \ref{lem_bound4} are rather involved, we attached them in the supplement.}
The proofs of Lemma~\ref{lem_bound3} and \ref{lem_bound4} are given in the appendix.

\begin{proof}[Proof of Theorem~\ref{thm_bound2}]
From Lemma~\ref{lem_bound3}, we know that when there is only one player for each task except for one task, then the worst case happens when $n-m$ players are dummy players for one task. The total reward distributed in this case is $(1+\dfrac{n-m}{m(m+1)})v(N, f)$ (set $k=1$ in Lemma~\ref{lem_bound4}). However, this is not necessarily the worst case in general. It is evident that if we simply move a dummy player from one task group to another task group, it will not change the dummy player's reward under the SEVB mechanism (it is always $\dfrac{1}{(m+1)m}v(N,f)$). 

Interestingly, if we move a dummy player to another task group and make it non-dummy, it may increase the total reward distributed under the SEVB mechanism. Assume that initially all dummy players are for task $\tau_i$. Now if we move a dummy $i^*$ from $\tau_i$ to $\tau_j\neq \tau_i$. To make $i^*$ non-dummy, assume that the Shapley value for $i^*$ is $\epsilon >0$ (but it does not dominate the original player for $\tau_j$, otherwise, $v(N, f)$ will be increased and make the two settings not comparable). Hence, the total reward for the non-dummy players has reduced by $\epsilon$ due to the change of $i*$ (property of the Shapley value). The reward for $i^*$ stays the same as $\dfrac{1}{(m+1)m}v(N,f)$. The key difference is the reward of the other dummy players for $\tau_i$. Following the proof of Lemma~\ref{lem_bound4}, the new Shapley value for the other dummy players becomes $\dfrac{1}{(m+1)m}v(N,f) + \dfrac{2}{m+2}\epsilon$, which has been increased by $\dfrac{2}{m+2}\epsilon$. Thus, the total reward increase after moving $i^*$ is 
\begin{equation}
\label{eq_thmb1}
\begin{aligned}
-\epsilon+\dfrac{2}{m+2}\epsilon (n-m-1) =\\ \dfrac{2n-3m-4}{m+2}\epsilon
\end{aligned}
\end{equation}
If $\dfrac{2n-3m-4}{m+2}\epsilon > 0$, then the difference is maximised when $\epsilon$ is maximised, which happens when $i^*$ is identical to the original player for $\tau_j$. Therefore, if moving a dummy player from $\tau_i$ to another task group increases the total reward, then the worst case happens when the dummy player is identical to the best player in the new group.

If we keep moving dummy players from $\tau_i$ to other groups, if this increases the total reward, then we can again prove that the worst case exists when the dummy player is identical to the best player in the new group. Moreover, following the analysis of Lemma~\ref{lem_bound4}, we can also prove that the reward stays the same even if the dummy players moved out from $\tau_i$'s group are not going to the same task group. Thus, we can keep moving dummy players to the same task $\tau_j$ to check if it increases the reward and stop when it does not increase the reward. {\color{black}When this stops, we reach the worst case setting, which is the setting given by Lemma~\ref{lem_bound4}.}
%We need to prove that the setting in Lemma~\ref{lem_bound3} is the worst case. 
%
%From Lemma~\ref{lem_bound2}, when we add more players for the same task, dummy players will maximize the reward under the SEVB mechanism. Now if we move one dummy player to another task, the reward stays the same, however this dummy becomes non-dummy, the reward might be increased. Need to prove that when the non-dummy is as good as the good player, the reward is maximised.
\end{proof}

\section{The VCG with Verification}
Our coalitional game setting can be also treated as an auction setting and apply standard auction mechanisms. This section shows that VCG can be adopted with execution verification to satisfy incentive compatibility in our setting.

\begin{framed}
	\noindent\textbf{VCG with Execution Verification (VCGEV)}\\
	\rule{\textwidth}{0.5pt}
Given all players' report profile $f^\prime$, 
for each player $i$,
\begin{itemize}
\item if $i$ is assigned a task under under $\pi_N^{f^\prime}$, and her realised execution time is $e_i$, $i$'s reward is defined as:
\begin{equation}
\scalebox{0.84}{
$
x_i^{vcgev}(f^\prime, e_i)= v(\pi_N^{f^\prime}, (f_i^{e_i}, f_{-i}^\prime)) - v(N\setminus\{i\}, f_{-i}^\prime)
$}
\end{equation}
where $f_i^{e_i}$ represents $e_i$ for $i$ and is defined as:
\begin{equation*}
f_i^{e_i}(e) = 
\begin{cases}
    1      & \quad \text{if } e = e_i\\
    0       & \quad \text{otherwise}
  \end{cases}
\end{equation*}

\item otherwise, $i$'s reward is $x_i^{vcgev}(f^\prime)= 0$.
\end{itemize}
\end{framed}

VCG is often applied in a non-cooperative setting, so the reward for each player is her marginal contribution given that all the other players are already in the coalition. A player's reward equals the marginal contribution in one permutation where the player is the last one joining the group. Therefore, if in the same task group, two good players are identical, then neither of them will be rewarded. On the other hand, if only one player can handle each task and the others are not capable, then each player who is assigned a task would get paid $v(N)$. Therefore, the reward distribution under VCG does not consider players' contributions in a more thoughtful manner as the Shapley value does. Nonetheless, the VCG based mechanism is incentive compatible. %{\color{red}We will compare their performance in the experiments.}

\begin{theorem}
The VCGEV mechanism is incentive compatible.
\end{theorem}
\begin{proof}
If a player $i$ is assigned a task, then her expected reward is $E[x_i^{vcgev}(f^\prime, e_i)] = \sum_{e_i\in E_i} f_i(e_i)x_i^{vcgev}(f^\prime, e_i)$. $v(N\setminus\{i\}, f_{-i}^\prime)$ is the probability to finish all the tasks without $i$'s participation, which is independent of $i$. The expectation of $v(\pi_N^{f^\prime}, (f_i^{e_i}, f_{-i}^\prime))$ is $\sum_{e_i\in E_i} f_i(e_i) v(\pi_N^{f^\prime}, (f_i^{e_i}, f_{-i}^\prime))$, which equals $v(N, (f_i, f_{-i}^\prime))$ if $i$ reports $f_i$ truthfully. It is evident that $v(N, (f_i, f_{-i}^\prime)) \geq v(N\setminus\{i\}, f_{-i}^\prime)$, $i$'s reward is non-negative when reporting truthfully.

Suppose that when $i$ misreports $f_i^\prime \neq f_i$, we have $\sum_{e_i\in E_i} f_i(e_i) v(\pi_N^{f^\prime}, (f_i^{e_i}, f_{-i}^\prime)) > v(N, (f_i, f_{-i}^\prime))$. Since $\sum_{e_i\in E_i} f_i(e_i) v(\pi_N^{f^\prime}, (f_i^{e_i}, f_{-i}^\prime))$ uses $i$'s true execution outcome to calculate the reward, the only thing that $f_i^\prime$ influences is the task assignment $\pi_N^{f^\prime}$. If indeed assignment $\pi_N^{f^\prime}$ gives a higher expected probability to finish all the tasks, then the definition of $v(N, (f_i, f_{-i}^\prime))$ should have chosen this assignment. However, we know $v(N, (f_i, f_{-i}^\prime))$ is optimal, which is a contradiction.

If a player $j$ is not assigned a task, then her marginal contribution is zero. If she misreports to be assigned a task, her expected reward becomes negative (following the above analysis).
\end{proof}

\begin{theorem}
Given any execution time distribution report profile $f$, the total reward distributed under the VCGEV mechanism can be as small as zero and as large as $m\times v(N, f)$.
\end{theorem}
%{\color{black}The proof is simple and is given in the supplement.}
\begin{proof}
For each player who is assigned a task, if there is another identical player in the same task group, then the player's reward under the VCGEV mechanism is zero. Therefore, in the worst case, all players receive zero reward.

In the other extreme case, where there is only one player for each task, then the reward for each player is $v(N, f)$, which is the maximal reward a player can get in any setting. Hence, the maximal total reward under the mechanism is $m\times v(N, f)$.
\end{proof}

It is worth mentioning that there is another simple mechanism, which is also truthful and the total distributed reward is exactly what they have generated. The mechanism simply shares the reward equally among all players after they have executed the whole project. More specifically, each player receives $V/n$ if they finished the project before the deadline, otherwise receives zero. It is easy to prove that the mechanism is IC. However, the weakness is that it does not pay them according to their marginal contributions/capabilities.

\section{Experiments}
We evaluate the proposed mechanisms on many instances that differ in the task allocation settings, completion-time distributions, and the deadlines. We generate random distributions with support size $s$, which is given as an input, such that the completion times for each player are in the range of [1\dots $s$] with uniform distribution.% ({\color{red} will this work?}). %An example for the distribution where $m=4$ may look like [0.58:0.23, 3.24:0.32, 3.87:0.17, 3.44:0.28]. For the instance itself, we allow different number of agents in each layer.

%Here we present the results for $2$ tasks, $3$ tasks and $4$ tasks settings with different number of players. In theory, we have seen that the total reward bound under the SEVB mechanism is a kind of linear function of the number of players $n$, while the higher bound of the VCGEV mechanism is independent of $n$. In the experiments of $4$ tasks with $4$ players for each task group, Figure~\ref{fig_4t} shows the total distributed reward in proportion of $v(N,f)$ for $20$ random instances. The results show that VCGEV varies a lot between settings, while SEVB is fairly stable. As in theory, the bound of SEVB is increasing as $n$ increases. We tested it on $2$ task settings with the number of players for each task varies from $2$ to $8$. The results (Figure~\ref{fig_2ts}) show that the average proportion over $v(N,f)$ is actually decreasing as $n$ increases for random $10$ instances (because more players for the same task will reduce their Shapley value). A similar trend holds for the VCGEV mechanism (Figure~\ref{fig_2tvcg}). We also compared the setings when the number of players varies between tasks. Figure~\ref{fig_3ts} shows the percentage over $v(N,f)$ on SEVB for $3$ tasks with $[3,3,3]$, $[4,3,2]$ and $[5,3,1]$ player settings on $50$ instances. It indicates that when the players are more imbalanced between tasks, the total reward becomes less.

Here we present the results for settings with $2$, $3$ and $4$ tasks, and with different number of players. In theory, we have seen that the total reward bound under the SEVB mechanism is a kind of linear function of the number of players $n$, whereas the higher bound of the VCGEV mechanism is independent of $n$. 
In the experiments of $4$ tasks with $4$ players for each task group, Figure~\ref{fig_4t} shows the total distributed reward in proportion of $v(N,f)$ for $20$ random instances. The results show that VCGEV varies a lot between settings, while SEVB is fairly stable. As in theory, the bound of SEVB is increasing as $n$ increases. We tested it on $2$ task settings with the number of players for each task varying from $2$ to $8$. The results in Figure~\ref{fig_2ts} show that for $10$ random instances the average proportion over $v(N,f)$ is actually decreasing as $n$ increases (because more players for the same task reduce their Shapley value). A similar trend holds for the VCGEV mechanism (Figure~\ref{fig_2tvcg}). We also examine settings with varying number of players for the different tasks. Figure~\ref{fig_3ts} shows the percentage over $v(N,f)$ on SEVB for $3$ tasks with $[3,3,3]$, $[4,3,2]$ and $[5,3,1]$ player settings on $50$ instances. It indicates that when the players are more imbalanced between tasks, the total reward is reduced.

\begin{figure}[t]
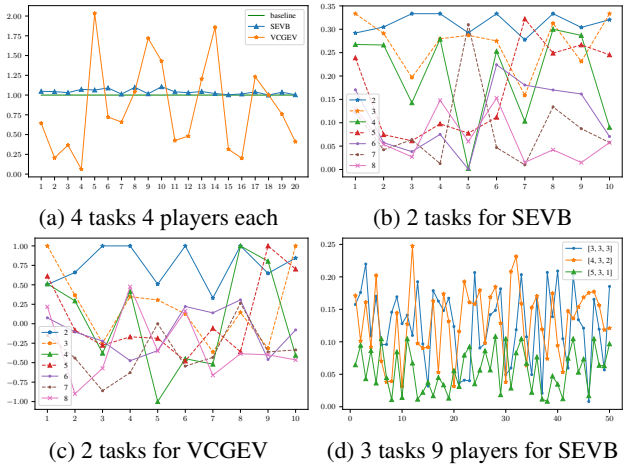

\centering
\begin{subfigure}[b]{0.23\textwidth}
\centering
      \scalebox{0.37}{\input{4t.pgf}}
  \caption{4 tasks 4 players each\label{fig_4t}}
\end{subfigure}
\begin{subfigure}[b]{0.23\textwidth}
  \centering
      \scalebox{0.37}{\input{2t_sValue.pgf}}
  \caption{2 tasks for SEVB\label{fig_2ts}}
\end{subfigure}
\begin{subfigure}[b]{0.23\textwidth}
  \centering
      \scalebox{0.37}{\input{2t_VCG.pgf}}
  \caption{2 tasks for VCGEV\label{fig_2tvcg}}
\end{subfigure}
\begin{subfigure}[b]{0.23\textwidth}
  \centering
      \scalebox{0.37}{\input{3t_sValue.pgf}}
  \caption{3 tasks 9 players for SEVB\label{fig_3ts}}
\end{subfigure}
\caption{Experimental results for SEVB and VCGEV}
\end{figure}

\section{Conclusion}
We have studied a task allocation setting which merges information revelation challenge in mechanism design and the payoff distribution challenge in cooperative game theory. The two challenges cannot be solved by using existing techniques from just one of the two fields. We proposed the very fist attempt to solve the two challenges together. The solution guarantees that players will truthfully reveal their private information and the rewards they receive from the coalition are fairly distributed. The cost to achieve this is that we might need to distribute more rewards to the players than what they can achieve. However, the extra reward is bounded and the experiments showed that the extra reward is diminishing when more players are involved. We have not investigated the other solution concepts such as core in the current analysis~\cite{Aumann1992gt}. 

\appendix

\section{Proofs}
\begin{proof}[Proof of Lemma~\ref{lem_bound3}]
From Lemma 1, we know that when $n=m+1$, this lemma holds. Now if $n=m+2$, we need to add one more player for task $\tau_i$. Let us call the three players for $\tau_i$ $i, i^{d1}$ and $i^{d2}$ and $i$ is the player assigned task $\tau_i$ and $i^{d1}$ is a dummy player and $i^{d2}$ is the newly added player. Since $i^{d1}$ is a dummy player, its Shapley value is zero. If the Shapley value for $i^{d2}$ is $\epsilon$, then following the proof of Lemma 1, we get that the Shapley value for all players $j\not\in \{i, i^{d1}, i^{d2}\}$ is $\dfrac{p}{m}+\epsilon$ and the Shapley value for $i$ is $\dfrac{p}{m} - m\epsilon$, which are also their reward distributed under the SEVB mechanism. The reward for $i^{d2}$ under the SEVB mechanism is $\dfrac{p}{(m+1)m}$, but the reward for $i^{d1}$ is $\dfrac{p}{(m+1)m} - \delta$ (because $i^{d2}$ creates competition for $i^{d1}$, which decreases $i^{d1}$'s reward). Therefore, the total reward distributed under the SEVB mechanism is 
$$(\dfrac{p}{m} - m\epsilon) + (\dfrac{p}{m}+\epsilon)(m-1) + 2\dfrac{p}{(m+1)m} - \delta,$$
which is $(1+\dfrac{2}{(m+1)m})p - \epsilon - \delta$. This is maximised when $\epsilon = \delta = 0$. Following this, we can prove that for any $n \geq m+1$, the lemma holds.
\end{proof}

%\newpage
%\begin{lemma}
%\label{lem_bound4}
%Given $n\geq m+1$ players for $m$ tasks, assume that there are $k\geq 1$ identical players for task $\tau_j$, $n-k-(m-2)$ players for task $\tau_i \neq \tau_j$ (where all players except for one of them are dummy players), one player for each of the other $m-2$ tasks and the $n$ players together can finish the $m$ tasks before the deadline with a probability $p$. Then the total reward distributed under the SEVB mechanism is $(1+\sum_{i=1}^{k}\dfrac{i!(m-1)!}{(m+i)!}(n-m-k+1))p$. 
%\end{lemma}
\begin{proof}[Proof of Lemma~\ref{lem_bound4}]
Let $k$ players for $\tau_j$ be $\{j^1, \cdots, j^k\}$ and the non-dummy player for $\tau_i$ be $i$. Let us compute the Shapley value for the player $l$ for task $\tau_l \in T\setminus \{\tau_i, \tau_j\}$. 

If $k=1$, we have $m$ non-dummy players (one for each task), then marginal contribution for $l$ is $p$ only when $l$ is ordered in the last among the $m$ non-dummy players in a permutation (see Figure~\ref{fig_1} permutation type $(1)$, dummy players can be placed anywhere in the permutation). Hence, $l$'s Shapley value is $\dfrac{p}{m}$. 

If $k=2$, for the permutation type $(1)$, we can insert $j^2$ anywhere in the permutation without changing $l$'s marginal contribution. In addition, adding player $j^2$ will increase $l$'s marginal contribution for the permutation where $l$ is not ordered in the last (see Figure~\ref{fig_1} permutation type $(2)$, where $j^1$ is ordered in the last). Without $j^2$, $l$'s marginal contribution under permutation type $(2)$ is zero. Because of $j^2$, $l$'s marginal contribution becomes $p$. Therefore, when $k=2$, the new Shapley for $l$ is $\dfrac{p}{m}+\dfrac{p}{(m+1)m}$.

If $k=3$, for the permutation types $(1)$ and $(2)$, we insert $j^3$ anywhere without changing $l$'s marginal contribution, so the Shapley value under $k=2$ stays for $k=3$. In addition, adding $j^3$ increases $l$'s marginal contribution from zero to $p$ for the permutation where $l$ is ordered before $j^1$ and $j^2$ (see Figure~\ref{fig_1} permutation type $(3)$). If we ignore all the dummy players, among all the $(m+2)!$ permutations, there $(m-1)!\times 2$ type $(3)$ permutations. Hence, $l$'s Shapley value is increased by $\dfrac{p}{(m+2)(m+1)m}$ by adding $j^3$.

Following the above analysis, for any $k$, we get the Shapley value for $l$ is
\begin{equation}
\label{eq_lemb41}
\begin{aligned}
\dfrac{p}{m}+\dfrac{p}{(m+1)m}+\cdots+ \dfrac{(k-1)!p}{(m+k-1)(m+k-2)\cdots m} =\\ \sum_{i=1}^k \dfrac{(i-1)!(m-1)!}{(m+i-1)!}p
\end{aligned}
\end{equation}
This is also the Shapley value for player $i$. 

Let us compute the Shapley value for players $\{j^1, \cdots, j^k\}$. Following Figure~\ref{fig_2}, if $k=1$, then $j$'s marginal contribution is $p$ as long as $j^1$ is the last player among the $m$ non-dummy players (see Figure~\ref{fig_2} permutation type $(1)$). If we ignore the dummy players, then there are $m!$ total permutations, among them $(m-1)!$ permutations have $j^1$ in the last. Therefore, $j^1$'s Shapley value is $\dfrac{p}{m}$ when $k=1$. Now if $k=2$, among all the $(m-1)!$ original permutations, we have to insert another player $j^2$ inside each permutation. Among all the insertions, only if we insert $j^2$ after $j^1$, $j^1$'s marginal contribution will stay the same (see Figure~\ref{fig_2} permutation type $(2)$). Therefore, $j^1$'s Shapley value is reduced to $\dfrac{p}{(m+1)m}$. Keep doing this until we add the last player $j^k$ (see Figure~\ref{fig_2} permutation type $(k)$), we get the reduced Shapley value for each $j^1$ as (which is the same for all $j\in \{j^1, \cdots, j^k\}$)
\begin{equation}
\label{eq_lemb42}
\dfrac{(k-1)!(m-1)!}{(m+k-1)!}p
\end{equation}

For each dummy player $i^*$ for $\tau_i$, we compute the reward for $i^*$ under the SEVB mechanism as its Shapley value when $i^*$ is identical to $i$. Following Figure~\ref{fig_3}, if $k=1$, then $i^*$ has marginal contribution $p$ only if $i$ is the last player and $i^*$ is the second last player in a permutation (see Figure~\ref{fig_3} permutation type $(1)$). If we ignore the other dummy players, there are $(m+1)!$ total permutations and only $(m-1)!$ of them have $i$ in the last and $i^*$ in the second last position. Thus, $i^*$'s Shapley value is $\dfrac{p}{(m+1)m}$ when $k=1$. Now if we add another player $j^2$, firstly, it will not affect all the permutation of type $(1)$, in addition, for some permutations where $i^*$ has zero contribute before may have a positive contribution now. These are the permutations where the last three players are either $(i^*, j^1, i)$ or $(i^*, i, j^1)$ (see Figure~\ref{fig_3} permutation type $(2)$). Before insert $j^2$, we have $(m-2)!2$ such permutations, where $i^*$'s contribution is zero. If we add $j^2$ before $i^*$, $i^*$'s contribution becomes $p$, which in average will give $i^*$ an extra reward $\dfrac{2p}{(m+2)(m+1)m}$. Following this procedure to add all $k$ players, we get the total new Shapley value for $i^*$ as
\begin{equation}
\label{eq_lemb43}
\sum_{i=1}^k \dfrac{i!(m-1)!}{(m+i)!}p
\end{equation}

To sum up, the total reward distributed under the SEVB mechanism is the sum of $(m-1)$ times of \eqref{eq_lemb41}, $k$ times of \eqref{eq_lemb42}, and $(n-m-k+1)$ times of \eqref{eq_lemb43}. It is evident that the sum of $(m-1)$ times of \eqref{eq_lemb41} and $k$ times of \eqref{eq_lemb42} is $p$ because it is the total Shapley value under the ground coalition. Hence, the total reward is $(1+\sum_{i=1}^{k}\dfrac{i!(m-1)!}{(m+i)!}(n-m-k+1))p$.
\end{proof}

\begin{figure}[t]
  \centering
\begin{subfigure}[b]{0.5\textwidth}
\centering
      \includegraphics[width=0.8\textwidth]{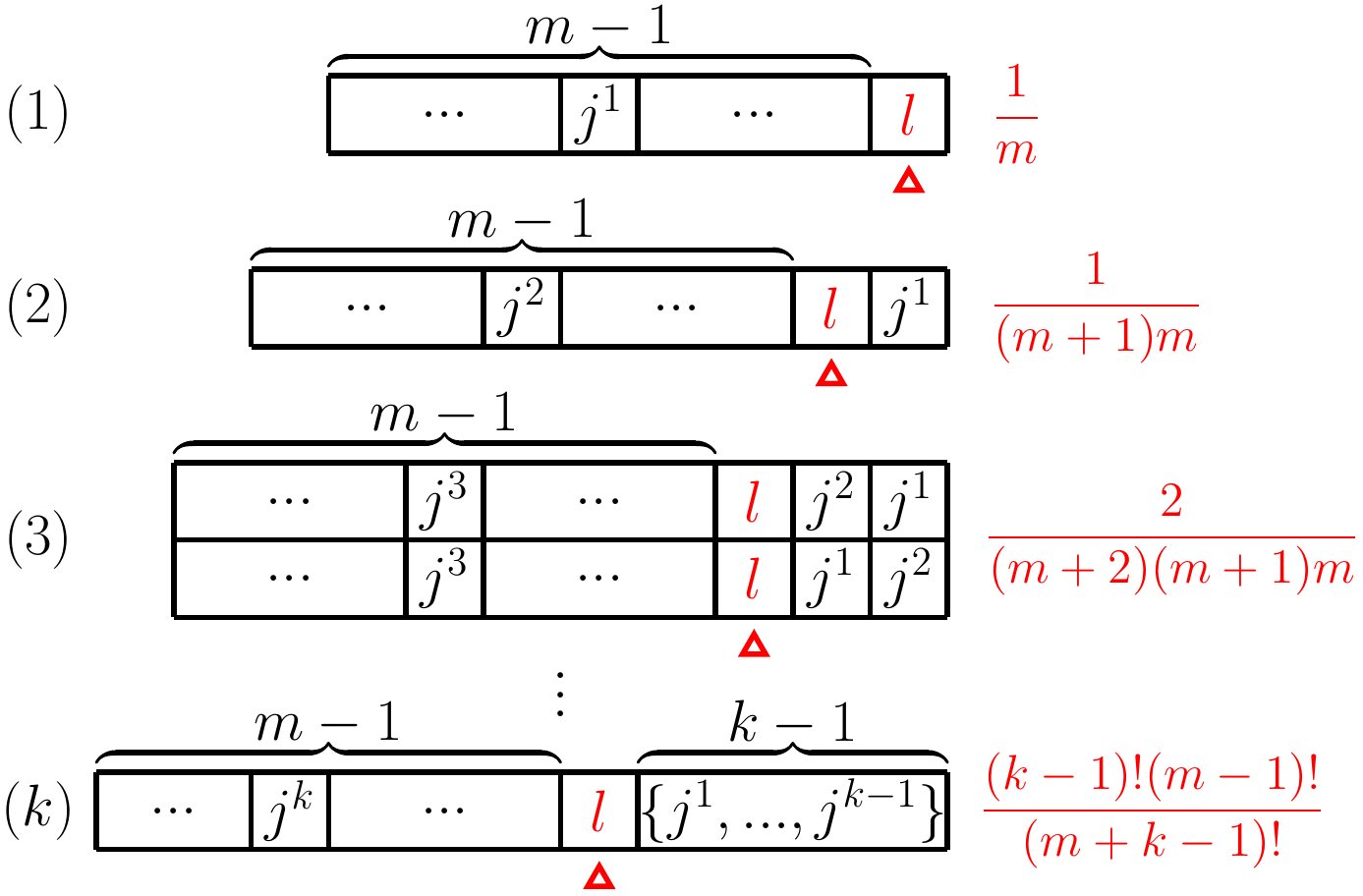}
  \caption{Reward computation for a normal player $l$.\label{fig_1}}
\end{subfigure}
\begin{subfigure}[b]{0.5\textwidth}
  \centering
      \includegraphics[width=0.8\textwidth]{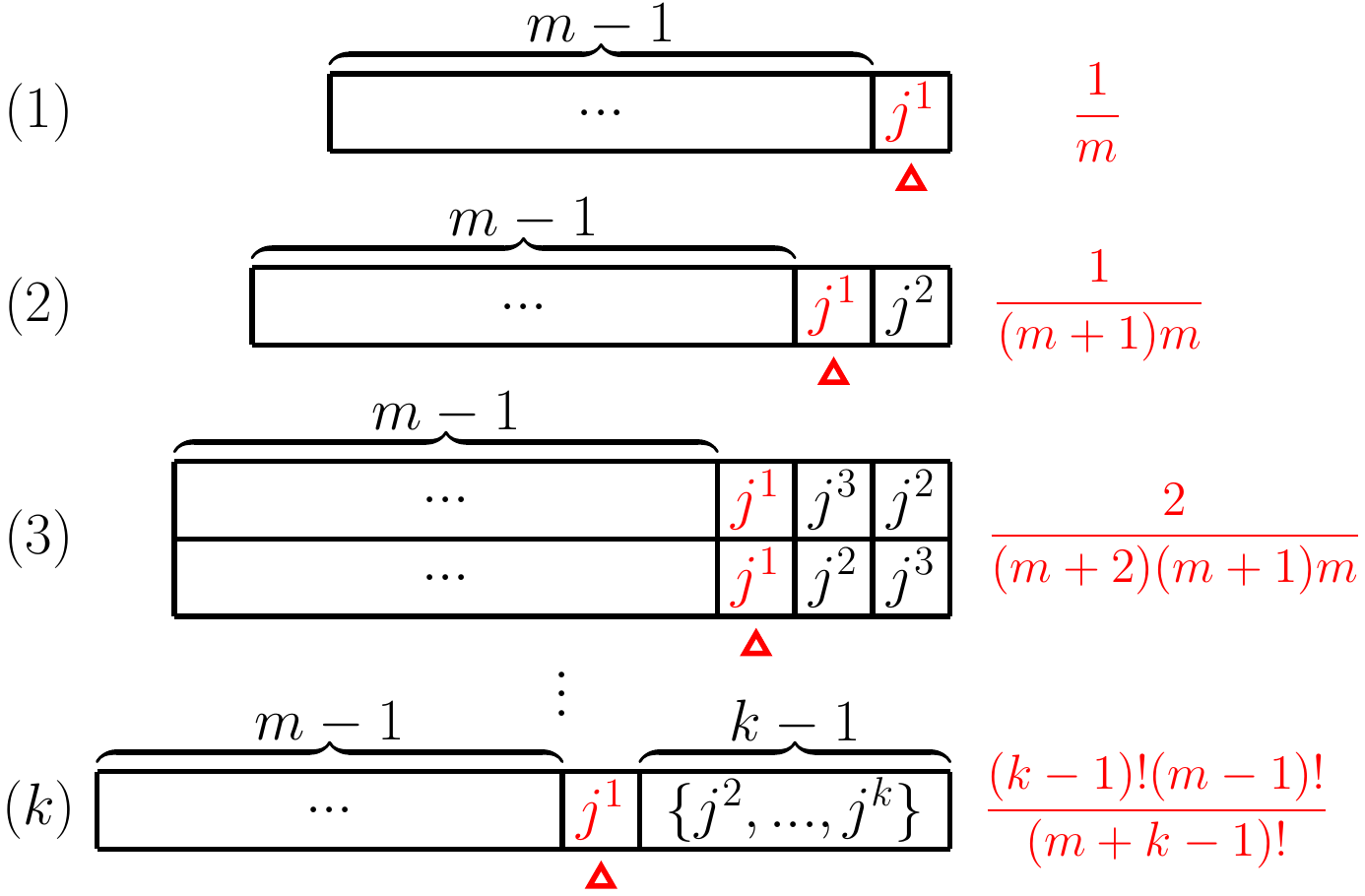}
  \caption{Reward computation for a player $j^i$ in $\tau_j$'s group.\label{fig_2}}
\end{subfigure}
\begin{subfigure}[b]{0.5\textwidth}
  \centering
      \includegraphics[width=0.8\textwidth]{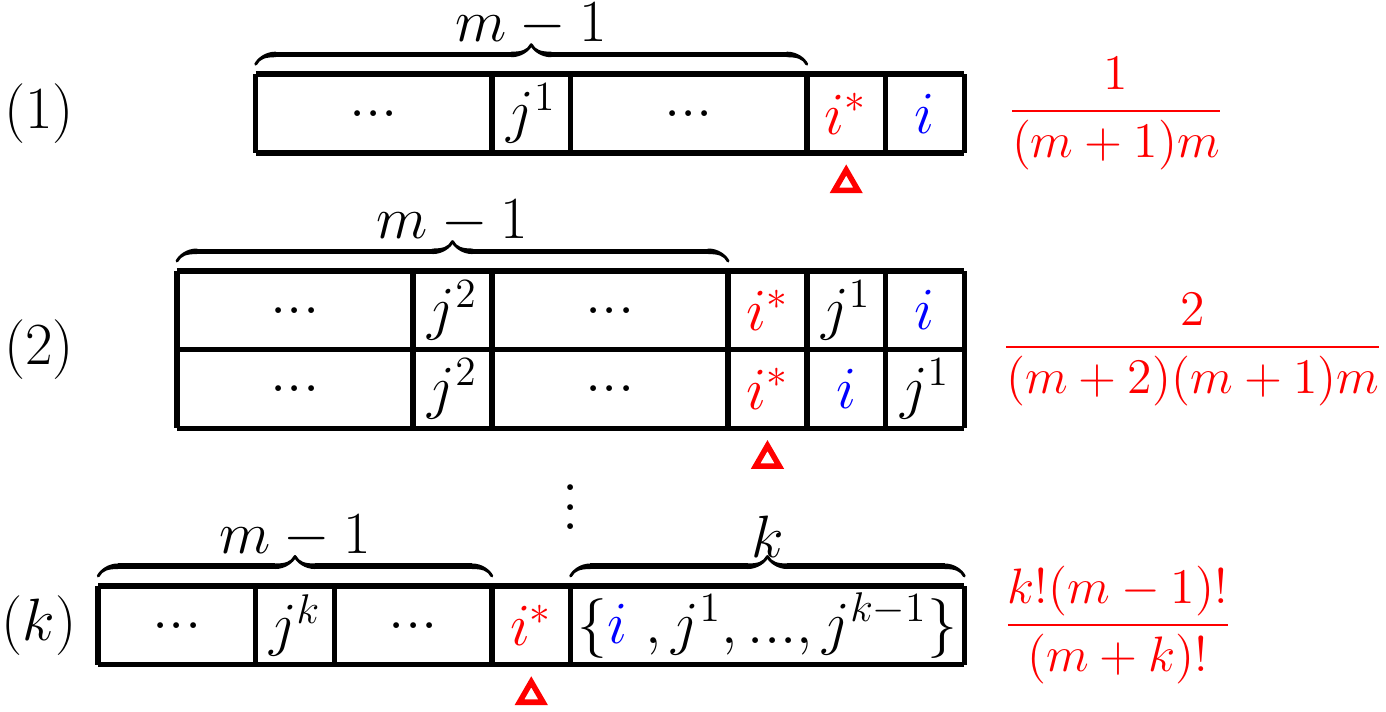}
  \caption{Reward computation for a dummy player $i^*$.\label{fig_3}}
  \end{subfigure}
  \caption{Reward computation under the SEVB mechanism.}
\end{figure}

\bibliographystyle{aaai}
\bibliography{bib}

\end{document}

%% file: 4t.pgf
%% Creator: Matplotlib, PGF backend
%%
%% To include the figure in your LaTeX document, write
%%   \input{<filename>.pgf}
%%
%% Make sure the required packages are loaded in your preamble
%%   \usepackage{pgf}
%%
%% Figures using additional raster images can only be included by \input if
%% they are in the same directory as the main LaTeX file. For loading figures
%% from other directories you can use the `import` package
%%   \usepackage{import}
%% and then include the figures with
%%   \import{<path to file>}{<filename>.pgf}
%%
%% Matplotlib used the following preamble
%%   \usepackage[utf8x]{inputenc}
%%   \usepackage[T1]{fontenc}
%%   \usepackage{fontspec}
%%
\begingroup%
\makeatletter%
\begin{pgfpicture}%
\pgfpathrectangle{\pgfpointorigin}{\pgfqpoint{4.296389in}{2.655314in}}%
\pgfusepath{use as bounding box, clip}%
\begin{pgfscope}%
\pgfsetbuttcap%
\pgfsetmiterjoin%
\definecolor{currentfill}{rgb}{1.000000,1.000000,1.000000}%
\pgfsetfillcolor{currentfill}%
\pgfsetlinewidth{0.000000pt}%
\definecolor{currentstroke}{rgb}{1.000000,1.000000,1.000000}%
\pgfsetstrokecolor{currentstroke}%
\pgfsetdash{}{0pt}%
\pgfpathmoveto{\pgfqpoint{0.000000in}{0.000000in}}%
\pgfpathlineto{\pgfqpoint{4.296389in}{0.000000in}}%
\pgfpathlineto{\pgfqpoint{4.296389in}{2.655314in}}%
\pgfpathlineto{\pgfqpoint{0.000000in}{2.655314in}}%
\pgfpathclose%
\pgfusepath{fill}%
\end{pgfscope}%
\begin{pgfscope}%
\pgfsetbuttcap%
\pgfsetmiterjoin%
\definecolor{currentfill}{rgb}{1.000000,1.000000,1.000000}%
\pgfsetfillcolor{currentfill}%
\pgfsetlinewidth{0.000000pt}%
\definecolor{currentstroke}{rgb}{0.000000,0.000000,0.000000}%
\pgfsetstrokecolor{currentstroke}%
\pgfsetstrokeopacity{0.000000}%
\pgfsetdash{}{0pt}%
\pgfpathmoveto{\pgfqpoint{0.307102in}{0.195889in}}%
\pgfpathlineto{\pgfqpoint{4.261389in}{0.195889in}}%
\pgfpathlineto{\pgfqpoint{4.261389in}{2.620314in}}%
\pgfpathlineto{\pgfqpoint{0.307102in}{2.620314in}}%
\pgfpathclose%
\pgfusepath{fill}%
\end{pgfscope}%
\begin{pgfscope}%
\pgfsetbuttcap%
\pgfsetroundjoin%
\definecolor{currentfill}{rgb}{0.000000,0.000000,0.000000}%
\pgfsetfillcolor{currentfill}%
\pgfsetlinewidth{0.803000pt}%
\definecolor{currentstroke}{rgb}{0.000000,0.000000,0.000000}%
\pgfsetstrokecolor{currentstroke}%
\pgfsetdash{}{0pt}%
\pgfsys@defobject{currentmarker}{\pgfqpoint{0.000000in}{-0.048611in}}{\pgfqpoint{0.000000in}{0.000000in}}{%
\pgfpathmoveto{\pgfqpoint{0.000000in}{0.000000in}}%
\pgfpathlineto{\pgfqpoint{0.000000in}{-0.048611in}}%
\pgfusepath{stroke,fill}%
}%
\begin{pgfscope}%
\pgfsys@transformshift{0.486842in}{0.195889in}%
\pgfsys@useobject{currentmarker}{}%
\end{pgfscope}%
\end{pgfscope}%
\begin{pgfscope}%
\pgftext[x=0.486842in,y=0.098667in,,top]{\rmfamily\fontsize{8.000000}{9.600000}\selectfont \(\displaystyle 1\)}%
\end{pgfscope}%
\begin{pgfscope}%
\pgfsetbuttcap%
\pgfsetroundjoin%
\definecolor{currentfill}{rgb}{0.000000,0.000000,0.000000}%
\pgfsetfillcolor{currentfill}%
\pgfsetlinewidth{0.803000pt}%
\definecolor{currentstroke}{rgb}{0.000000,0.000000,0.000000}%
\pgfsetstrokecolor{currentstroke}%
\pgfsetdash{}{0pt}%
\pgfsys@defobject{currentmarker}{\pgfqpoint{0.000000in}{-0.048611in}}{\pgfqpoint{0.000000in}{0.000000in}}{%
\pgfpathmoveto{\pgfqpoint{0.000000in}{0.000000in}}%
\pgfpathlineto{\pgfqpoint{0.000000in}{-0.048611in}}%
\pgfusepath{stroke,fill}%
}%
\begin{pgfscope}%
\pgfsys@transformshift{0.676042in}{0.195889in}%
\pgfsys@useobject{currentmarker}{}%
\end{pgfscope}%
\end{pgfscope}%
\begin{pgfscope}%
\pgftext[x=0.676042in,y=0.098667in,,top]{\rmfamily\fontsize{8.000000}{9.600000}\selectfont \(\displaystyle 2\)}%
\end{pgfscope}%
\begin{pgfscope}%
\pgfsetbuttcap%
\pgfsetroundjoin%
\definecolor{currentfill}{rgb}{0.000000,0.000000,0.000000}%
\pgfsetfillcolor{currentfill}%
\pgfsetlinewidth{0.803000pt}%
\definecolor{currentstroke}{rgb}{0.000000,0.000000,0.000000}%
\pgfsetstrokecolor{currentstroke}%
\pgfsetdash{}{0pt}%
\pgfsys@defobject{currentmarker}{\pgfqpoint{0.000000in}{-0.048611in}}{\pgfqpoint{0.000000in}{0.000000in}}{%
\pgfpathmoveto{\pgfqpoint{0.000000in}{0.000000in}}%
\pgfpathlineto{\pgfqpoint{0.000000in}{-0.048611in}}%
\pgfusepath{stroke,fill}%
}%
\begin{pgfscope}%
\pgfsys@transformshift{0.865243in}{0.195889in}%
\pgfsys@useobject{currentmarker}{}%
\end{pgfscope}%
\end{pgfscope}%
\begin{pgfscope}%
\pgftext[x=0.865243in,y=0.098667in,,top]{\rmfamily\fontsize{8.000000}{9.600000}\selectfont \(\displaystyle 3\)}%
\end{pgfscope}%
\begin{pgfscope}%
\pgfsetbuttcap%
\pgfsetroundjoin%
\definecolor{currentfill}{rgb}{0.000000,0.000000,0.000000}%
\pgfsetfillcolor{currentfill}%
\pgfsetlinewidth{0.803000pt}%
\definecolor{currentstroke}{rgb}{0.000000,0.000000,0.000000}%
\pgfsetstrokecolor{currentstroke}%
\pgfsetdash{}{0pt}%
\pgfsys@defobject{currentmarker}{\pgfqpoint{0.000000in}{-0.048611in}}{\pgfqpoint{0.000000in}{0.000000in}}{%
\pgfpathmoveto{\pgfqpoint{0.000000in}{0.000000in}}%
\pgfpathlineto{\pgfqpoint{0.000000in}{-0.048611in}}%
\pgfusepath{stroke,fill}%
}%
\begin{pgfscope}%
\pgfsys@transformshift{1.054443in}{0.195889in}%
\pgfsys@useobject{currentmarker}{}%
\end{pgfscope}%
\end{pgfscope}%
\begin{pgfscope}%
\pgftext[x=1.054443in,y=0.098667in,,top]{\rmfamily\fontsize{8.000000}{9.600000}\selectfont \(\displaystyle 4\)}%
\end{pgfscope}%
\begin{pgfscope}%
\pgfsetbuttcap%
\pgfsetroundjoin%
\definecolor{currentfill}{rgb}{0.000000,0.000000,0.000000}%
\pgfsetfillcolor{currentfill}%
\pgfsetlinewidth{0.803000pt}%
\definecolor{currentstroke}{rgb}{0.000000,0.000000,0.000000}%
\pgfsetstrokecolor{currentstroke}%
\pgfsetdash{}{0pt}%
\pgfsys@defobject{currentmarker}{\pgfqpoint{0.000000in}{-0.048611in}}{\pgfqpoint{0.000000in}{0.000000in}}{%
\pgfpathmoveto{\pgfqpoint{0.000000in}{0.000000in}}%
\pgfpathlineto{\pgfqpoint{0.000000in}{-0.048611in}}%
\pgfusepath{stroke,fill}%
}%
\begin{pgfscope}%
\pgfsys@transformshift{1.243643in}{0.195889in}%
\pgfsys@useobject{currentmarker}{}%
\end{pgfscope}%
\end{pgfscope}%
\begin{pgfscope}%
\pgftext[x=1.243643in,y=0.098667in,,top]{\rmfamily\fontsize{8.000000}{9.600000}\selectfont \(\displaystyle 5\)}%
\end{pgfscope}%
\begin{pgfscope}%
\pgfsetbuttcap%
\pgfsetroundjoin%
\definecolor{currentfill}{rgb}{0.000000,0.000000,0.000000}%
\pgfsetfillcolor{currentfill}%
\pgfsetlinewidth{0.803000pt}%
\definecolor{currentstroke}{rgb}{0.000000,0.000000,0.000000}%
\pgfsetstrokecolor{currentstroke}%
\pgfsetdash{}{0pt}%
\pgfsys@defobject{currentmarker}{\pgfqpoint{0.000000in}{-0.048611in}}{\pgfqpoint{0.000000in}{0.000000in}}{%
\pgfpathmoveto{\pgfqpoint{0.000000in}{0.000000in}}%
\pgfpathlineto{\pgfqpoint{0.000000in}{-0.048611in}}%
\pgfusepath{stroke,fill}%
}%
\begin{pgfscope}%
\pgfsys@transformshift{1.432844in}{0.195889in}%
\pgfsys@useobject{currentmarker}{}%
\end{pgfscope}%
\end{pgfscope}%
\begin{pgfscope}%
\pgftext[x=1.432844in,y=0.098667in,,top]{\rmfamily\fontsize{8.000000}{9.600000}\selectfont \(\displaystyle 6\)}%
\end{pgfscope}%
\begin{pgfscope}%
\pgfsetbuttcap%
\pgfsetroundjoin%
\definecolor{currentfill}{rgb}{0.000000,0.000000,0.000000}%
\pgfsetfillcolor{currentfill}%
\pgfsetlinewidth{0.803000pt}%
\definecolor{currentstroke}{rgb}{0.000000,0.000000,0.000000}%
\pgfsetstrokecolor{currentstroke}%
\pgfsetdash{}{0pt}%
\pgfsys@defobject{currentmarker}{\pgfqpoint{0.000000in}{-0.048611in}}{\pgfqpoint{0.000000in}{0.000000in}}{%
\pgfpathmoveto{\pgfqpoint{0.000000in}{0.000000in}}%
\pgfpathlineto{\pgfqpoint{0.000000in}{-0.048611in}}%
\pgfusepath{stroke,fill}%
}%
\begin{pgfscope}%
\pgfsys@transformshift{1.622044in}{0.195889in}%
\pgfsys@useobject{currentmarker}{}%
\end{pgfscope}%
\end{pgfscope}%
\begin{pgfscope}%
\pgftext[x=1.622044in,y=0.098667in,,top]{\rmfamily\fontsize{8.000000}{9.600000}\selectfont \(\displaystyle 7\)}%
\end{pgfscope}%
\begin{pgfscope}%
\pgfsetbuttcap%
\pgfsetroundjoin%
\definecolor{currentfill}{rgb}{0.000000,0.000000,0.000000}%
\pgfsetfillcolor{currentfill}%
\pgfsetlinewidth{0.803000pt}%
\definecolor{currentstroke}{rgb}{0.000000,0.000000,0.000000}%
\pgfsetstrokecolor{currentstroke}%
\pgfsetdash{}{0pt}%
\pgfsys@defobject{currentmarker}{\pgfqpoint{0.000000in}{-0.048611in}}{\pgfqpoint{0.000000in}{0.000000in}}{%
\pgfpathmoveto{\pgfqpoint{0.000000in}{0.000000in}}%
\pgfpathlineto{\pgfqpoint{0.000000in}{-0.048611in}}%
\pgfusepath{stroke,fill}%
}%
\begin{pgfscope}%
\pgfsys@transformshift{1.811244in}{0.195889in}%
\pgfsys@useobject{currentmarker}{}%
\end{pgfscope}%
\end{pgfscope}%
\begin{pgfscope}%
\pgftext[x=1.811244in,y=0.098667in,,top]{\rmfamily\fontsize{8.000000}{9.600000}\selectfont \(\displaystyle 8\)}%
\end{pgfscope}%
\begin{pgfscope}%
\pgfsetbuttcap%
\pgfsetroundjoin%
\definecolor{currentfill}{rgb}{0.000000,0.000000,0.000000}%
\pgfsetfillcolor{currentfill}%
\pgfsetlinewidth{0.803000pt}%
\definecolor{currentstroke}{rgb}{0.000000,0.000000,0.000000}%
\pgfsetstrokecolor{currentstroke}%
\pgfsetdash{}{0pt}%
\pgfsys@defobject{currentmarker}{\pgfqpoint{0.000000in}{-0.048611in}}{\pgfqpoint{0.000000in}{0.000000in}}{%
\pgfpathmoveto{\pgfqpoint{0.000000in}{0.000000in}}%
\pgfpathlineto{\pgfqpoint{0.000000in}{-0.048611in}}%
\pgfusepath{stroke,fill}%
}%
\begin{pgfscope}%
\pgfsys@transformshift{2.000445in}{0.195889in}%
\pgfsys@useobject{currentmarker}{}%
\end{pgfscope}%
\end{pgfscope}%
\begin{pgfscope}%
\pgftext[x=2.000445in,y=0.098667in,,top]{\rmfamily\fontsize{8.000000}{9.600000}\selectfont \(\displaystyle 9\)}%
\end{pgfscope}%
\begin{pgfscope}%
\pgfsetbuttcap%
\pgfsetroundjoin%
\definecolor{currentfill}{rgb}{0.000000,0.000000,0.000000}%
\pgfsetfillcolor{currentfill}%
\pgfsetlinewidth{0.803000pt}%
\definecolor{currentstroke}{rgb}{0.000000,0.000000,0.000000}%
\pgfsetstrokecolor{currentstroke}%
\pgfsetdash{}{0pt}%
\pgfsys@defobject{currentmarker}{\pgfqpoint{0.000000in}{-0.048611in}}{\pgfqpoint{0.000000in}{0.000000in}}{%
\pgfpathmoveto{\pgfqpoint{0.000000in}{0.000000in}}%
\pgfpathlineto{\pgfqpoint{0.000000in}{-0.048611in}}%
\pgfusepath{stroke,fill}%
}%
\begin{pgfscope}%
\pgfsys@transformshift{2.189645in}{0.195889in}%
\pgfsys@useobject{currentmarker}{}%
\end{pgfscope}%
\end{pgfscope}%
\begin{pgfscope}%
\pgftext[x=2.189645in,y=0.098667in,,top]{\rmfamily\fontsize{8.000000}{9.600000}\selectfont \(\displaystyle 10\)}%
\end{pgfscope}%
\begin{pgfscope}%
\pgfsetbuttcap%
\pgfsetroundjoin%
\definecolor{currentfill}{rgb}{0.000000,0.000000,0.000000}%
\pgfsetfillcolor{currentfill}%
\pgfsetlinewidth{0.803000pt}%
\definecolor{currentstroke}{rgb}{0.000000,0.000000,0.000000}%
\pgfsetstrokecolor{currentstroke}%
\pgfsetdash{}{0pt}%
\pgfsys@defobject{currentmarker}{\pgfqpoint{0.000000in}{-0.048611in}}{\pgfqpoint{0.000000in}{0.000000in}}{%
\pgfpathmoveto{\pgfqpoint{0.000000in}{0.000000in}}%
\pgfpathlineto{\pgfqpoint{0.000000in}{-0.048611in}}%
\pgfusepath{stroke,fill}%
}%
\begin{pgfscope}%
\pgfsys@transformshift{2.378845in}{0.195889in}%
\pgfsys@useobject{currentmarker}{}%
\end{pgfscope}%
\end{pgfscope}%
\begin{pgfscope}%
\pgftext[x=2.378845in,y=0.098667in,,top]{\rmfamily\fontsize{8.000000}{9.600000}\selectfont \(\displaystyle 11\)}%
\end{pgfscope}%
\begin{pgfscope}%
\pgfsetbuttcap%
\pgfsetroundjoin%
\definecolor{currentfill}{rgb}{0.000000,0.000000,0.000000}%
\pgfsetfillcolor{currentfill}%
\pgfsetlinewidth{0.803000pt}%
\definecolor{currentstroke}{rgb}{0.000000,0.000000,0.000000}%
\pgfsetstrokecolor{currentstroke}%
\pgfsetdash{}{0pt}%
\pgfsys@defobject{currentmarker}{\pgfqpoint{0.000000in}{-0.048611in}}{\pgfqpoint{0.000000in}{0.000000in}}{%
\pgfpathmoveto{\pgfqpoint{0.000000in}{0.000000in}}%
\pgfpathlineto{\pgfqpoint{0.000000in}{-0.048611in}}%
\pgfusepath{stroke,fill}%
}%
\begin{pgfscope}%
\pgfsys@transformshift{2.568046in}{0.195889in}%
\pgfsys@useobject{currentmarker}{}%
\end{pgfscope}%
\end{pgfscope}%
\begin{pgfscope}%
\pgftext[x=2.568046in,y=0.098667in,,top]{\rmfamily\fontsize{8.000000}{9.600000}\selectfont \(\displaystyle 12\)}%
\end{pgfscope}%
\begin{pgfscope}%
\pgfsetbuttcap%
\pgfsetroundjoin%
\definecolor{currentfill}{rgb}{0.000000,0.000000,0.000000}%
\pgfsetfillcolor{currentfill}%
\pgfsetlinewidth{0.803000pt}%
\definecolor{currentstroke}{rgb}{0.000000,0.000000,0.000000}%
\pgfsetstrokecolor{currentstroke}%
\pgfsetdash{}{0pt}%
\pgfsys@defobject{currentmarker}{\pgfqpoint{0.000000in}{-0.048611in}}{\pgfqpoint{0.000000in}{0.000000in}}{%
\pgfpathmoveto{\pgfqpoint{0.000000in}{0.000000in}}%
\pgfpathlineto{\pgfqpoint{0.000000in}{-0.048611in}}%
\pgfusepath{stroke,fill}%
}%
\begin{pgfscope}%
\pgfsys@transformshift{2.757246in}{0.195889in}%
\pgfsys@useobject{currentmarker}{}%
\end{pgfscope}%
\end{pgfscope}%
\begin{pgfscope}%
\pgftext[x=2.757246in,y=0.098667in,,top]{\rmfamily\fontsize{8.000000}{9.600000}\selectfont \(\displaystyle 13\)}%
\end{pgfscope}%
\begin{pgfscope}%
\pgfsetbuttcap%
\pgfsetroundjoin%
\definecolor{currentfill}{rgb}{0.000000,0.000000,0.000000}%
\pgfsetfillcolor{currentfill}%
\pgfsetlinewidth{0.803000pt}%
\definecolor{currentstroke}{rgb}{0.000000,0.000000,0.000000}%
\pgfsetstrokecolor{currentstroke}%
\pgfsetdash{}{0pt}%
\pgfsys@defobject{currentmarker}{\pgfqpoint{0.000000in}{-0.048611in}}{\pgfqpoint{0.000000in}{0.000000in}}{%
\pgfpathmoveto{\pgfqpoint{0.000000in}{0.000000in}}%
\pgfpathlineto{\pgfqpoint{0.000000in}{-0.048611in}}%
\pgfusepath{stroke,fill}%
}%
\begin{pgfscope}%
\pgfsys@transformshift{2.946446in}{0.195889in}%
\pgfsys@useobject{currentmarker}{}%
\end{pgfscope}%
\end{pgfscope}%
\begin{pgfscope}%
\pgftext[x=2.946446in,y=0.098667in,,top]{\rmfamily\fontsize{8.000000}{9.600000}\selectfont \(\displaystyle 14\)}%
\end{pgfscope}%
\begin{pgfscope}%
\pgfsetbuttcap%
\pgfsetroundjoin%
\definecolor{currentfill}{rgb}{0.000000,0.000000,0.000000}%
\pgfsetfillcolor{currentfill}%
\pgfsetlinewidth{0.803000pt}%
\definecolor{currentstroke}{rgb}{0.000000,0.000000,0.000000}%
\pgfsetstrokecolor{currentstroke}%
\pgfsetdash{}{0pt}%
\pgfsys@defobject{currentmarker}{\pgfqpoint{0.000000in}{-0.048611in}}{\pgfqpoint{0.000000in}{0.000000in}}{%
\pgfpathmoveto{\pgfqpoint{0.000000in}{0.000000in}}%
\pgfpathlineto{\pgfqpoint{0.000000in}{-0.048611in}}%
\pgfusepath{stroke,fill}%
}%
\begin{pgfscope}%
\pgfsys@transformshift{3.135647in}{0.195889in}%
\pgfsys@useobject{currentmarker}{}%
\end{pgfscope}%
\end{pgfscope}%
\begin{pgfscope}%
\pgftext[x=3.135647in,y=0.098667in,,top]{\rmfamily\fontsize{8.000000}{9.600000}\selectfont \(\displaystyle 15\)}%
\end{pgfscope}%
\begin{pgfscope}%
\pgfsetbuttcap%
\pgfsetroundjoin%
\definecolor{currentfill}{rgb}{0.000000,0.000000,0.000000}%
\pgfsetfillcolor{currentfill}%
\pgfsetlinewidth{0.803000pt}%
\definecolor{currentstroke}{rgb}{0.000000,0.000000,0.000000}%
\pgfsetstrokecolor{currentstroke}%
\pgfsetdash{}{0pt}%
\pgfsys@defobject{currentmarker}{\pgfqpoint{0.000000in}{-0.048611in}}{\pgfqpoint{0.000000in}{0.000000in}}{%
\pgfpathmoveto{\pgfqpoint{0.000000in}{0.000000in}}%
\pgfpathlineto{\pgfqpoint{0.000000in}{-0.048611in}}%
\pgfusepath{stroke,fill}%
}%
\begin{pgfscope}%
\pgfsys@transformshift{3.324847in}{0.195889in}%
\pgfsys@useobject{currentmarker}{}%
\end{pgfscope}%
\end{pgfscope}%
\begin{pgfscope}%
\pgftext[x=3.324847in,y=0.098667in,,top]{\rmfamily\fontsize{8.000000}{9.600000}\selectfont \(\displaystyle 16\)}%
\end{pgfscope}%
\begin{pgfscope}%
\pgfsetbuttcap%
\pgfsetroundjoin%
\definecolor{currentfill}{rgb}{0.000000,0.000000,0.000000}%
\pgfsetfillcolor{currentfill}%
\pgfsetlinewidth{0.803000pt}%
\definecolor{currentstroke}{rgb}{0.000000,0.000000,0.000000}%
\pgfsetstrokecolor{currentstroke}%
\pgfsetdash{}{0pt}%
\pgfsys@defobject{currentmarker}{\pgfqpoint{0.000000in}{-0.048611in}}{\pgfqpoint{0.000000in}{0.000000in}}{%
\pgfpathmoveto{\pgfqpoint{0.000000in}{0.000000in}}%
\pgfpathlineto{\pgfqpoint{0.000000in}{-0.048611in}}%
\pgfusepath{stroke,fill}%
}%
\begin{pgfscope}%
\pgfsys@transformshift{3.514047in}{0.195889in}%
\pgfsys@useobject{currentmarker}{}%
\end{pgfscope}%
\end{pgfscope}%
\begin{pgfscope}%
\pgftext[x=3.514047in,y=0.098667in,,top]{\rmfamily\fontsize{8.000000}{9.600000}\selectfont \(\displaystyle 17\)}%
\end{pgfscope}%
\begin{pgfscope}%
\pgfsetbuttcap%
\pgfsetroundjoin%
\definecolor{currentfill}{rgb}{0.000000,0.000000,0.000000}%
\pgfsetfillcolor{currentfill}%
\pgfsetlinewidth{0.803000pt}%
\definecolor{currentstroke}{rgb}{0.000000,0.000000,0.000000}%
\pgfsetstrokecolor{currentstroke}%
\pgfsetdash{}{0pt}%
\pgfsys@defobject{currentmarker}{\pgfqpoint{0.000000in}{-0.048611in}}{\pgfqpoint{0.000000in}{0.000000in}}{%
\pgfpathmoveto{\pgfqpoint{0.000000in}{0.000000in}}%
\pgfpathlineto{\pgfqpoint{0.000000in}{-0.048611in}}%
\pgfusepath{stroke,fill}%
}%
\begin{pgfscope}%
\pgfsys@transformshift{3.703248in}{0.195889in}%
\pgfsys@useobject{currentmarker}{}%
\end{pgfscope}%
\end{pgfscope}%
\begin{pgfscope}%
\pgftext[x=3.703248in,y=0.098667in,,top]{\rmfamily\fontsize{8.000000}{9.600000}\selectfont \(\displaystyle 18\)}%
\end{pgfscope}%
\begin{pgfscope}%
\pgfsetbuttcap%
\pgfsetroundjoin%
\definecolor{currentfill}{rgb}{0.000000,0.000000,0.000000}%
\pgfsetfillcolor{currentfill}%
\pgfsetlinewidth{0.803000pt}%
\definecolor{currentstroke}{rgb}{0.000000,0.000000,0.000000}%
\pgfsetstrokecolor{currentstroke}%
\pgfsetdash{}{0pt}%
\pgfsys@defobject{currentmarker}{\pgfqpoint{0.000000in}{-0.048611in}}{\pgfqpoint{0.000000in}{0.000000in}}{%
\pgfpathmoveto{\pgfqpoint{0.000000in}{0.000000in}}%
\pgfpathlineto{\pgfqpoint{0.000000in}{-0.048611in}}%
\pgfusepath{stroke,fill}%
}%
\begin{pgfscope}%
\pgfsys@transformshift{3.892448in}{0.195889in}%
\pgfsys@useobject{currentmarker}{}%
\end{pgfscope}%
\end{pgfscope}%
\begin{pgfscope}%
\pgftext[x=3.892448in,y=0.098667in,,top]{\rmfamily\fontsize{8.000000}{9.600000}\selectfont \(\displaystyle 19\)}%
\end{pgfscope}%
\begin{pgfscope}%
\pgfsetbuttcap%
\pgfsetroundjoin%
\definecolor{currentfill}{rgb}{0.000000,0.000000,0.000000}%
\pgfsetfillcolor{currentfill}%
\pgfsetlinewidth{0.803000pt}%
\definecolor{currentstroke}{rgb}{0.000000,0.000000,0.000000}%
\pgfsetstrokecolor{currentstroke}%
\pgfsetdash{}{0pt}%
\pgfsys@defobject{currentmarker}{\pgfqpoint{0.000000in}{-0.048611in}}{\pgfqpoint{0.000000in}{0.000000in}}{%
\pgfpathmoveto{\pgfqpoint{0.000000in}{0.000000in}}%
\pgfpathlineto{\pgfqpoint{0.000000in}{-0.048611in}}%
\pgfusepath{stroke,fill}%
}%
\begin{pgfscope}%
\pgfsys@transformshift{4.081648in}{0.195889in}%
\pgfsys@useobject{currentmarker}{}%
\end{pgfscope}%
\end{pgfscope}%
\begin{pgfscope}%
\pgftext[x=4.081648in,y=0.098667in,,top]{\rmfamily\fontsize{8.000000}{9.600000}\selectfont \(\displaystyle 20\)}%
\end{pgfscope}%
\begin{pgfscope}%
\pgfsetbuttcap%
\pgfsetroundjoin%
\definecolor{currentfill}{rgb}{0.000000,0.000000,0.000000}%
\pgfsetfillcolor{currentfill}%
\pgfsetlinewidth{0.803000pt}%
\definecolor{currentstroke}{rgb}{0.000000,0.000000,0.000000}%
\pgfsetstrokecolor{currentstroke}%
\pgfsetdash{}{0pt}%
\pgfsys@defobject{currentmarker}{\pgfqpoint{-0.048611in}{0.000000in}}{\pgfqpoint{0.000000in}{0.000000in}}{%
\pgfpathmoveto{\pgfqpoint{0.000000in}{0.000000in}}%
\pgfpathlineto{\pgfqpoint{-0.048611in}{0.000000in}}%
\pgfusepath{stroke,fill}%
}%
\begin{pgfscope}%
\pgfsys@transformshift{0.307102in}{0.235825in}%
\pgfsys@useobject{currentmarker}{}%
\end{pgfscope}%
\end{pgfscope}%
\begin{pgfscope}%
\pgftext[x=-0.000000in,y=0.197269in,left,base]{\rmfamily\fontsize{8.000000}{9.600000}\selectfont \(\displaystyle 0.00\)}%
\end{pgfscope}%
\begin{pgfscope}%
\pgfsetbuttcap%
\pgfsetroundjoin%
\definecolor{currentfill}{rgb}{0.000000,0.000000,0.000000}%
\pgfsetfillcolor{currentfill}%
\pgfsetlinewidth{0.803000pt}%
\definecolor{currentstroke}{rgb}{0.000000,0.000000,0.000000}%
\pgfsetstrokecolor{currentstroke}%
\pgfsetdash{}{0pt}%
\pgfsys@defobject{currentmarker}{\pgfqpoint{-0.048611in}{0.000000in}}{\pgfqpoint{0.000000in}{0.000000in}}{%
\pgfpathmoveto{\pgfqpoint{0.000000in}{0.000000in}}%
\pgfpathlineto{\pgfqpoint{-0.048611in}{0.000000in}}%
\pgfusepath{stroke,fill}%
}%
\begin{pgfscope}%
\pgfsys@transformshift{0.307102in}{0.515497in}%
\pgfsys@useobject{currentmarker}{}%
\end{pgfscope}%
\end{pgfscope}%
\begin{pgfscope}%
\pgftext[x=-0.000000in,y=0.476941in,left,base]{\rmfamily\fontsize{8.000000}{9.600000}\selectfont \(\displaystyle 0.25\)}%
\end{pgfscope}%
\begin{pgfscope}%
\pgfsetbuttcap%
\pgfsetroundjoin%
\definecolor{currentfill}{rgb}{0.000000,0.000000,0.000000}%
\pgfsetfillcolor{currentfill}%
\pgfsetlinewidth{0.803000pt}%
\definecolor{currentstroke}{rgb}{0.000000,0.000000,0.000000}%
\pgfsetstrokecolor{currentstroke}%
\pgfsetdash{}{0pt}%
\pgfsys@defobject{currentmarker}{\pgfqpoint{-0.048611in}{0.000000in}}{\pgfqpoint{0.000000in}{0.000000in}}{%
\pgfpathmoveto{\pgfqpoint{0.000000in}{0.000000in}}%
\pgfpathlineto{\pgfqpoint{-0.048611in}{0.000000in}}%
\pgfusepath{stroke,fill}%
}%
\begin{pgfscope}%
\pgfsys@transformshift{0.307102in}{0.795169in}%
\pgfsys@useobject{currentmarker}{}%
\end{pgfscope}%
\end{pgfscope}%
\begin{pgfscope}%
\pgftext[x=-0.000000in,y=0.756614in,left,base]{\rmfamily\fontsize{8.000000}{9.600000}\selectfont \(\displaystyle 0.50\)}%
\end{pgfscope}%
\begin{pgfscope}%
\pgfsetbuttcap%
\pgfsetroundjoin%
\definecolor{currentfill}{rgb}{0.000000,0.000000,0.000000}%
\pgfsetfillcolor{currentfill}%
\pgfsetlinewidth{0.803000pt}%
\definecolor{currentstroke}{rgb}{0.000000,0.000000,0.000000}%
\pgfsetstrokecolor{currentstroke}%
\pgfsetdash{}{0pt}%
\pgfsys@defobject{currentmarker}{\pgfqpoint{-0.048611in}{0.000000in}}{\pgfqpoint{0.000000in}{0.000000in}}{%
\pgfpathmoveto{\pgfqpoint{0.000000in}{0.000000in}}%
\pgfpathlineto{\pgfqpoint{-0.048611in}{0.000000in}}%
\pgfusepath{stroke,fill}%
}%
\begin{pgfscope}%
\pgfsys@transformshift{0.307102in}{1.074841in}%
\pgfsys@useobject{currentmarker}{}%
\end{pgfscope}%
\end{pgfscope}%
\begin{pgfscope}%
\pgftext[x=-0.000000in,y=1.036286in,left,base]{\rmfamily\fontsize{8.000000}{9.600000}\selectfont \(\displaystyle 0.75\)}%
\end{pgfscope}%
\begin{pgfscope}%
\pgfsetbuttcap%
\pgfsetroundjoin%
\definecolor{currentfill}{rgb}{0.000000,0.000000,0.000000}%
\pgfsetfillcolor{currentfill}%
\pgfsetlinewidth{0.803000pt}%
\definecolor{currentstroke}{rgb}{0.000000,0.000000,0.000000}%
\pgfsetstrokecolor{currentstroke}%
\pgfsetdash{}{0pt}%
\pgfsys@defobject{currentmarker}{\pgfqpoint{-0.048611in}{0.000000in}}{\pgfqpoint{0.000000in}{0.000000in}}{%
\pgfpathmoveto{\pgfqpoint{0.000000in}{0.000000in}}%
\pgfpathlineto{\pgfqpoint{-0.048611in}{0.000000in}}%
\pgfusepath{stroke,fill}%
}%
\begin{pgfscope}%
\pgfsys@transformshift{0.307102in}{1.354514in}%
\pgfsys@useobject{currentmarker}{}%
\end{pgfscope}%
\end{pgfscope}%
\begin{pgfscope}%
\pgftext[x=-0.000000in,y=1.315958in,left,base]{\rmfamily\fontsize{8.000000}{9.600000}\selectfont \(\displaystyle 1.00\)}%
\end{pgfscope}%
\begin{pgfscope}%
\pgfsetbuttcap%
\pgfsetroundjoin%
\definecolor{currentfill}{rgb}{0.000000,0.000000,0.000000}%
\pgfsetfillcolor{currentfill}%
\pgfsetlinewidth{0.803000pt}%
\definecolor{currentstroke}{rgb}{0.000000,0.000000,0.000000}%
\pgfsetstrokecolor{currentstroke}%
\pgfsetdash{}{0pt}%
\pgfsys@defobject{currentmarker}{\pgfqpoint{-0.048611in}{0.000000in}}{\pgfqpoint{0.000000in}{0.000000in}}{%
\pgfpathmoveto{\pgfqpoint{0.000000in}{0.000000in}}%
\pgfpathlineto{\pgfqpoint{-0.048611in}{0.000000in}}%
\pgfusepath{stroke,fill}%
}%
\begin{pgfscope}%
\pgfsys@transformshift{0.307102in}{1.634186in}%
\pgfsys@useobject{currentmarker}{}%
\end{pgfscope}%
\end{pgfscope}%
\begin{pgfscope}%
\pgftext[x=-0.000000in,y=1.595630in,left,base]{\rmfamily\fontsize{8.000000}{9.600000}\selectfont \(\displaystyle 1.25\)}%
\end{pgfscope}%
\begin{pgfscope}%
\pgfsetbuttcap%
\pgfsetroundjoin%
\definecolor{currentfill}{rgb}{0.000000,0.000000,0.000000}%
\pgfsetfillcolor{currentfill}%
\pgfsetlinewidth{0.803000pt}%
\definecolor{currentstroke}{rgb}{0.000000,0.000000,0.000000}%
\pgfsetstrokecolor{currentstroke}%
\pgfsetdash{}{0pt}%
\pgfsys@defobject{currentmarker}{\pgfqpoint{-0.048611in}{0.000000in}}{\pgfqpoint{0.000000in}{0.000000in}}{%
\pgfpathmoveto{\pgfqpoint{0.000000in}{0.000000in}}%
\pgfpathlineto{\pgfqpoint{-0.048611in}{0.000000in}}%
\pgfusepath{stroke,fill}%
}%
\begin{pgfscope}%
\pgfsys@transformshift{0.307102in}{1.913858in}%
\pgfsys@useobject{currentmarker}{}%
\end{pgfscope}%
\end{pgfscope}%
\begin{pgfscope}%
\pgftext[x=-0.000000in,y=1.875303in,left,base]{\rmfamily\fontsize{8.000000}{9.600000}\selectfont \(\displaystyle 1.50\)}%
\end{pgfscope}%
\begin{pgfscope}%
\pgfsetbuttcap%
\pgfsetroundjoin%
\definecolor{currentfill}{rgb}{0.000000,0.000000,0.000000}%
\pgfsetfillcolor{currentfill}%
\pgfsetlinewidth{0.803000pt}%
\definecolor{currentstroke}{rgb}{0.000000,0.000000,0.000000}%
\pgfsetstrokecolor{currentstroke}%
\pgfsetdash{}{0pt}%
\pgfsys@defobject{currentmarker}{\pgfqpoint{-0.048611in}{0.000000in}}{\pgfqpoint{0.000000in}{0.000000in}}{%
\pgfpathmoveto{\pgfqpoint{0.000000in}{0.000000in}}%
\pgfpathlineto{\pgfqpoint{-0.048611in}{0.000000in}}%
\pgfusepath{stroke,fill}%
}%
\begin{pgfscope}%
\pgfsys@transformshift{0.307102in}{2.193531in}%
\pgfsys@useobject{currentmarker}{}%
\end{pgfscope}%
\end{pgfscope}%
\begin{pgfscope}%
\pgftext[x=-0.000000in,y=2.154975in,left,base]{\rmfamily\fontsize{8.000000}{9.600000}\selectfont \(\displaystyle 1.75\)}%
\end{pgfscope}%
\begin{pgfscope}%
\pgfsetbuttcap%
\pgfsetroundjoin%
\definecolor{currentfill}{rgb}{0.000000,0.000000,0.000000}%
\pgfsetfillcolor{currentfill}%
\pgfsetlinewidth{0.803000pt}%
\definecolor{currentstroke}{rgb}{0.000000,0.000000,0.000000}%
\pgfsetstrokecolor{currentstroke}%
\pgfsetdash{}{0pt}%
\pgfsys@defobject{currentmarker}{\pgfqpoint{-0.048611in}{0.000000in}}{\pgfqpoint{0.000000in}{0.000000in}}{%
\pgfpathmoveto{\pgfqpoint{0.000000in}{0.000000in}}%
\pgfpathlineto{\pgfqpoint{-0.048611in}{0.000000in}}%
\pgfusepath{stroke,fill}%
}%
\begin{pgfscope}%
\pgfsys@transformshift{0.307102in}{2.473203in}%
\pgfsys@useobject{currentmarker}{}%
\end{pgfscope}%
\end{pgfscope}%
\begin{pgfscope}%
\pgftext[x=-0.000000in,y=2.434647in,left,base]{\rmfamily\fontsize{8.000000}{9.600000}\selectfont \(\displaystyle 2.00\)}%
\end{pgfscope}%
\begin{pgfscope}%
\pgfpathrectangle{\pgfqpoint{0.307102in}{0.195889in}}{\pgfqpoint{3.954287in}{2.424425in}}%
\pgfusepath{clip}%
\pgfsetrectcap%
\pgfsetroundjoin%
\pgfsetlinewidth{1.003750pt}%
\definecolor{currentstroke}{rgb}{0.000000,0.500000,0.000000}%
\pgfsetstrokecolor{currentstroke}%
\pgfsetdash{}{0pt}%
\pgfpathmoveto{\pgfqpoint{0.486842in}{1.354514in}}%
\pgfpathlineto{\pgfqpoint{0.676042in}{1.354514in}}%
\pgfpathlineto{\pgfqpoint{0.865243in}{1.354514in}}%
\pgfpathlineto{\pgfqpoint{1.054443in}{1.354514in}}%
\pgfpathlineto{\pgfqpoint{1.243643in}{1.354514in}}%
\pgfpathlineto{\pgfqpoint{1.432844in}{1.354514in}}%
\pgfpathlineto{\pgfqpoint{1.622044in}{1.354514in}}%
\pgfpathlineto{\pgfqpoint{1.811244in}{1.354514in}}%
\pgfpathlineto{\pgfqpoint{2.000445in}{1.354514in}}%
\pgfpathlineto{\pgfqpoint{2.189645in}{1.354514in}}%
\pgfpathlineto{\pgfqpoint{2.378845in}{1.354514in}}%
\pgfpathlineto{\pgfqpoint{2.568046in}{1.354514in}}%
\pgfpathlineto{\pgfqpoint{2.757246in}{1.354514in}}%
\pgfpathlineto{\pgfqpoint{2.946446in}{1.354514in}}%
\pgfpathlineto{\pgfqpoint{3.135647in}{1.354514in}}%
\pgfpathlineto{\pgfqpoint{3.324847in}{1.354514in}}%
\pgfpathlineto{\pgfqpoint{3.514047in}{1.354514in}}%
\pgfpathlineto{\pgfqpoint{3.703248in}{1.354514in}}%
\pgfpathlineto{\pgfqpoint{3.892448in}{1.354514in}}%
\pgfpathlineto{\pgfqpoint{4.081648in}{1.354514in}}%
\pgfusepath{stroke}%
\end{pgfscope}%
\begin{pgfscope}%
\pgfpathrectangle{\pgfqpoint{0.307102in}{0.195889in}}{\pgfqpoint{3.954287in}{2.424425in}}%
\pgfusepath{clip}%
\pgfsetrectcap%
\pgfsetroundjoin%
\pgfsetlinewidth{1.003750pt}%
\definecolor{currentstroke}{rgb}{0.121569,0.466667,0.705882}%
\pgfsetstrokecolor{currentstroke}%
\pgfsetdash{}{0pt}%
\pgfpathmoveto{\pgfqpoint{0.486842in}{1.407578in}}%
\pgfpathlineto{\pgfqpoint{0.676042in}{1.403276in}}%
\pgfpathlineto{\pgfqpoint{0.865243in}{1.388814in}}%
\pgfpathlineto{\pgfqpoint{1.054443in}{1.435342in}}%
\pgfpathlineto{\pgfqpoint{1.243643in}{1.426101in}}%
\pgfpathlineto{\pgfqpoint{1.432844in}{1.454709in}}%
\pgfpathlineto{\pgfqpoint{1.622044in}{1.368599in}}%
\pgfpathlineto{\pgfqpoint{1.811244in}{1.463544in}}%
\pgfpathlineto{\pgfqpoint{2.000445in}{1.373042in}}%
\pgfpathlineto{\pgfqpoint{2.189645in}{1.472700in}}%
\pgfpathlineto{\pgfqpoint{2.378845in}{1.400761in}}%
\pgfpathlineto{\pgfqpoint{2.568046in}{1.387732in}}%
\pgfpathlineto{\pgfqpoint{2.757246in}{1.404640in}}%
\pgfpathlineto{\pgfqpoint{2.946446in}{1.376147in}}%
\pgfpathlineto{\pgfqpoint{3.135647in}{1.358837in}}%
\pgfpathlineto{\pgfqpoint{3.324847in}{1.370393in}}%
\pgfpathlineto{\pgfqpoint{3.514047in}{1.399026in}}%
\pgfpathlineto{\pgfqpoint{3.703248in}{1.355267in}}%
\pgfpathlineto{\pgfqpoint{3.892448in}{1.394751in}}%
\pgfpathlineto{\pgfqpoint{4.081648in}{1.359037in}}%
\pgfusepath{stroke}%
\end{pgfscope}%
\begin{pgfscope}%
\pgfpathrectangle{\pgfqpoint{0.307102in}{0.195889in}}{\pgfqpoint{3.954287in}{2.424425in}}%
\pgfusepath{clip}%
\pgfsetbuttcap%
\pgfsetmiterjoin%
\definecolor{currentfill}{rgb}{0.121569,0.466667,0.705882}%
\pgfsetfillcolor{currentfill}%
\pgfsetlinewidth{1.003750pt}%
\definecolor{currentstroke}{rgb}{0.121569,0.466667,0.705882}%
\pgfsetstrokecolor{currentstroke}%
\pgfsetdash{}{0pt}%
\pgfsys@defobject{currentmarker}{\pgfqpoint{-0.027778in}{-0.027778in}}{\pgfqpoint{0.027778in}{0.027778in}}{%
\pgfpathmoveto{\pgfqpoint{0.000000in}{0.027778in}}%
\pgfpathlineto{\pgfqpoint{-0.027778in}{-0.027778in}}%
\pgfpathlineto{\pgfqpoint{0.027778in}{-0.027778in}}%
\pgfpathclose%
\pgfusepath{stroke,fill}%
}%
\begin{pgfscope}%
\pgfsys@transformshift{0.486842in}{1.407578in}%
\pgfsys@useobject{currentmarker}{}%
\end{pgfscope}%
\begin{pgfscope}%
\pgfsys@transformshift{0.676042in}{1.403276in}%
\pgfsys@useobject{currentmarker}{}%
\end{pgfscope}%
\begin{pgfscope}%
\pgfsys@transformshift{0.865243in}{1.388814in}%
\pgfsys@useobject{currentmarker}{}%
\end{pgfscope}%
\begin{pgfscope}%
\pgfsys@transformshift{1.054443in}{1.435342in}%
\pgfsys@useobject{currentmarker}{}%
\end{pgfscope}%
\begin{pgfscope}%
\pgfsys@transformshift{1.243643in}{1.426101in}%
\pgfsys@useobject{currentmarker}{}%
\end{pgfscope}%
\begin{pgfscope}%
\pgfsys@transformshift{1.432844in}{1.454709in}%
\pgfsys@useobject{currentmarker}{}%
\end{pgfscope}%
\begin{pgfscope}%
\pgfsys@transformshift{1.622044in}{1.368599in}%
\pgfsys@useobject{currentmarker}{}%
\end{pgfscope}%
\begin{pgfscope}%
\pgfsys@transformshift{1.811244in}{1.463544in}%
\pgfsys@useobject{currentmarker}{}%
\end{pgfscope}%
\begin{pgfscope}%
\pgfsys@transformshift{2.000445in}{1.373042in}%
\pgfsys@useobject{currentmarker}{}%
\end{pgfscope}%
\begin{pgfscope}%
\pgfsys@transformshift{2.189645in}{1.472700in}%
\pgfsys@useobject{currentmarker}{}%
\end{pgfscope}%
\begin{pgfscope}%
\pgfsys@transformshift{2.378845in}{1.400761in}%
\pgfsys@useobject{currentmarker}{}%
\end{pgfscope}%
\begin{pgfscope}%
\pgfsys@transformshift{2.568046in}{1.387732in}%
\pgfsys@useobject{currentmarker}{}%
\end{pgfscope}%
\begin{pgfscope}%
\pgfsys@transformshift{2.757246in}{1.404640in}%
\pgfsys@useobject{currentmarker}{}%
\end{pgfscope}%
\begin{pgfscope}%
\pgfsys@transformshift{2.946446in}{1.376147in}%
\pgfsys@useobject{currentmarker}{}%
\end{pgfscope}%
\begin{pgfscope}%
\pgfsys@transformshift{3.135647in}{1.358837in}%
\pgfsys@useobject{currentmarker}{}%
\end{pgfscope}%
\begin{pgfscope}%
\pgfsys@transformshift{3.324847in}{1.370393in}%
\pgfsys@useobject{currentmarker}{}%
\end{pgfscope}%
\begin{pgfscope}%
\pgfsys@transformshift{3.514047in}{1.399026in}%
\pgfsys@useobject{currentmarker}{}%
\end{pgfscope}%
\begin{pgfscope}%
\pgfsys@transformshift{3.703248in}{1.355267in}%
\pgfsys@useobject{currentmarker}{}%
\end{pgfscope}%
\begin{pgfscope}%
\pgfsys@transformshift{3.892448in}{1.394751in}%
\pgfsys@useobject{currentmarker}{}%
\end{pgfscope}%
\begin{pgfscope}%
\pgfsys@transformshift{4.081648in}{1.359037in}%
\pgfsys@useobject{currentmarker}{}%
\end{pgfscope}%
\end{pgfscope}%
\begin{pgfscope}%
\pgfpathrectangle{\pgfqpoint{0.307102in}{0.195889in}}{\pgfqpoint{3.954287in}{2.424425in}}%
\pgfusepath{clip}%
\pgfsetrectcap%
\pgfsetroundjoin%
\pgfsetlinewidth{1.003750pt}%
\definecolor{currentstroke}{rgb}{1.000000,0.498039,0.054902}%
\pgfsetstrokecolor{currentstroke}%
\pgfsetdash{}{0pt}%
\pgfpathmoveto{\pgfqpoint{0.486842in}{0.955323in}}%
\pgfpathlineto{\pgfqpoint{0.676042in}{0.465021in}}%
\pgfpathlineto{\pgfqpoint{0.865243in}{0.644744in}}%
\pgfpathlineto{\pgfqpoint{1.054443in}{0.306090in}}%
\pgfpathlineto{\pgfqpoint{1.243643in}{2.510113in}}%
\pgfpathlineto{\pgfqpoint{1.432844in}{1.044199in}}%
\pgfpathlineto{\pgfqpoint{1.622044in}{0.972612in}}%
\pgfpathlineto{\pgfqpoint{1.811244in}{1.402506in}}%
\pgfpathlineto{\pgfqpoint{2.000445in}{2.156783in}}%
\pgfpathlineto{\pgfqpoint{2.189645in}{1.835869in}}%
\pgfpathlineto{\pgfqpoint{2.378845in}{0.710812in}}%
\pgfpathlineto{\pgfqpoint{2.568046in}{0.773230in}}%
\pgfpathlineto{\pgfqpoint{2.757246in}{1.584522in}}%
\pgfpathlineto{\pgfqpoint{2.946446in}{2.312735in}}%
\pgfpathlineto{\pgfqpoint{3.135647in}{0.587015in}}%
\pgfpathlineto{\pgfqpoint{3.324847in}{0.460043in}}%
\pgfpathlineto{\pgfqpoint{3.514047in}{1.613365in}}%
\pgfpathlineto{\pgfqpoint{3.703248in}{1.355988in}}%
\pgfpathlineto{\pgfqpoint{3.892448in}{1.085046in}}%
\pgfpathlineto{\pgfqpoint{4.081648in}{0.694859in}}%
\pgfusepath{stroke}%
\end{pgfscope}%
\begin{pgfscope}%
\pgfpathrectangle{\pgfqpoint{0.307102in}{0.195889in}}{\pgfqpoint{3.954287in}{2.424425in}}%
\pgfusepath{clip}%
\pgfsetbuttcap%
\pgfsetbeveljoin%
\definecolor{currentfill}{rgb}{1.000000,0.498039,0.054902}%
\pgfsetfillcolor{currentfill}%
\pgfsetlinewidth{1.003750pt}%
\definecolor{currentstroke}{rgb}{1.000000,0.498039,0.054902}%
\pgfsetstrokecolor{currentstroke}%
\pgfsetdash{}{0pt}%
\pgfsys@defobject{currentmarker}{\pgfqpoint{-0.026418in}{-0.022473in}}{\pgfqpoint{0.026418in}{0.027778in}}{%
\pgfpathmoveto{\pgfqpoint{0.000000in}{0.027778in}}%
\pgfpathlineto{\pgfqpoint{-0.006236in}{0.008584in}}%
\pgfpathlineto{\pgfqpoint{-0.026418in}{0.008584in}}%
\pgfpathlineto{\pgfqpoint{-0.010091in}{-0.003279in}}%
\pgfpathlineto{\pgfqpoint{-0.016327in}{-0.022473in}}%
\pgfpathlineto{\pgfqpoint{-0.000000in}{-0.010610in}}%
\pgfpathlineto{\pgfqpoint{0.016327in}{-0.022473in}}%
\pgfpathlineto{\pgfqpoint{0.010091in}{-0.003279in}}%
\pgfpathlineto{\pgfqpoint{0.026418in}{0.008584in}}%
\pgfpathlineto{\pgfqpoint{0.006236in}{0.008584in}}%
\pgfpathclose%
\pgfusepath{stroke,fill}%
}%
\begin{pgfscope}%
\pgfsys@transformshift{0.486842in}{0.955323in}%
\pgfsys@useobject{currentmarker}{}%
\end{pgfscope}%
\begin{pgfscope}%
\pgfsys@transformshift{0.676042in}{0.465021in}%
\pgfsys@useobject{currentmarker}{}%
\end{pgfscope}%
\begin{pgfscope}%
\pgfsys@transformshift{0.865243in}{0.644744in}%
\pgfsys@useobject{currentmarker}{}%
\end{pgfscope}%
\begin{pgfscope}%
\pgfsys@transformshift{1.054443in}{0.306090in}%
\pgfsys@useobject{currentmarker}{}%
\end{pgfscope}%
\begin{pgfscope}%
\pgfsys@transformshift{1.243643in}{2.510113in}%
\pgfsys@useobject{currentmarker}{}%
\end{pgfscope}%
\begin{pgfscope}%
\pgfsys@transformshift{1.432844in}{1.044199in}%
\pgfsys@useobject{currentmarker}{}%
\end{pgfscope}%
\begin{pgfscope}%
\pgfsys@transformshift{1.622044in}{0.972612in}%
\pgfsys@useobject{currentmarker}{}%
\end{pgfscope}%
\begin{pgfscope}%
\pgfsys@transformshift{1.811244in}{1.402506in}%
\pgfsys@useobject{currentmarker}{}%
\end{pgfscope}%
\begin{pgfscope}%
\pgfsys@transformshift{2.000445in}{2.156783in}%
\pgfsys@useobject{currentmarker}{}%
\end{pgfscope}%
\begin{pgfscope}%
\pgfsys@transformshift{2.189645in}{1.835869in}%
\pgfsys@useobject{currentmarker}{}%
\end{pgfscope}%
\begin{pgfscope}%
\pgfsys@transformshift{2.378845in}{0.710812in}%
\pgfsys@useobject{currentmarker}{}%
\end{pgfscope}%
\begin{pgfscope}%
\pgfsys@transformshift{2.568046in}{0.773230in}%
\pgfsys@useobject{currentmarker}{}%
\end{pgfscope}%
\begin{pgfscope}%
\pgfsys@transformshift{2.757246in}{1.584522in}%
\pgfsys@useobject{currentmarker}{}%
\end{pgfscope}%
\begin{pgfscope}%
\pgfsys@transformshift{2.946446in}{2.312735in}%
\pgfsys@useobject{currentmarker}{}%
\end{pgfscope}%
\begin{pgfscope}%
\pgfsys@transformshift{3.135647in}{0.587015in}%
\pgfsys@useobject{currentmarker}{}%
\end{pgfscope}%
\begin{pgfscope}%
\pgfsys@transformshift{3.324847in}{0.460043in}%
\pgfsys@useobject{currentmarker}{}%
\end{pgfscope}%
\begin{pgfscope}%
\pgfsys@transformshift{3.514047in}{1.613365in}%
\pgfsys@useobject{currentmarker}{}%
\end{pgfscope}%
\begin{pgfscope}%
\pgfsys@transformshift{3.703248in}{1.355988in}%
\pgfsys@useobject{currentmarker}{}%
\end{pgfscope}%
\begin{pgfscope}%
\pgfsys@transformshift{3.892448in}{1.085046in}%
\pgfsys@useobject{currentmarker}{}%
\end{pgfscope}%
\begin{pgfscope}%
\pgfsys@transformshift{4.081648in}{0.694859in}%
\pgfsys@useobject{currentmarker}{}%
\end{pgfscope}%
\end{pgfscope}%
\begin{pgfscope}%
\pgfsetrectcap%
\pgfsetmiterjoin%
\pgfsetlinewidth{0.803000pt}%
\definecolor{currentstroke}{rgb}{0.000000,0.000000,0.000000}%
\pgfsetstrokecolor{currentstroke}%
\pgfsetdash{}{0pt}%
\pgfpathmoveto{\pgfqpoint{0.307102in}{0.195889in}}%
\pgfpathlineto{\pgfqpoint{0.307102in}{2.620314in}}%
\pgfusepath{stroke}%
\end{pgfscope}%
\begin{pgfscope}%
\pgfsetrectcap%
\pgfsetmiterjoin%
\pgfsetlinewidth{0.803000pt}%
\definecolor{currentstroke}{rgb}{0.000000,0.000000,0.000000}%
\pgfsetstrokecolor{currentstroke}%
\pgfsetdash{}{0pt}%
\pgfpathmoveto{\pgfqpoint{4.261389in}{0.195889in}}%
\pgfpathlineto{\pgfqpoint{4.261389in}{2.620314in}}%
\pgfusepath{stroke}%
\end{pgfscope}%
\begin{pgfscope}%
\pgfsetrectcap%
\pgfsetmiterjoin%
\pgfsetlinewidth{0.803000pt}%
\definecolor{currentstroke}{rgb}{0.000000,0.000000,0.000000}%
\pgfsetstrokecolor{currentstroke}%
\pgfsetdash{}{0pt}%
\pgfpathmoveto{\pgfqpoint{0.307102in}{0.195889in}}%
\pgfpathlineto{\pgfqpoint{4.261389in}{0.195889in}}%
\pgfusepath{stroke}%
\end{pgfscope}%
\begin{pgfscope}%
\pgfsetrectcap%
\pgfsetmiterjoin%
\pgfsetlinewidth{0.803000pt}%
\definecolor{currentstroke}{rgb}{0.000000,0.000000,0.000000}%
\pgfsetstrokecolor{currentstroke}%
\pgfsetdash{}{0pt}%
\pgfpathmoveto{\pgfqpoint{0.307102in}{2.620314in}}%
\pgfpathlineto{\pgfqpoint{4.261389in}{2.620314in}}%
\pgfusepath{stroke}%
\end{pgfscope}%
\begin{pgfscope}%
\pgfsetbuttcap%
\pgfsetmiterjoin%
\definecolor{currentfill}{rgb}{1.000000,1.000000,1.000000}%
\pgfsetfillcolor{currentfill}%
\pgfsetfillopacity{0.800000}%
\pgfsetlinewidth{1.003750pt}%
\definecolor{currentstroke}{rgb}{0.800000,0.800000,0.800000}%
\pgfsetstrokecolor{currentstroke}%
\pgfsetstrokeopacity{0.800000}%
\pgfsetdash{}{0pt}%
\pgfpathmoveto{\pgfqpoint{3.396277in}{2.066759in}}%
\pgfpathlineto{\pgfqpoint{4.183611in}{2.066759in}}%
\pgfpathquadraticcurveto{\pgfqpoint{4.205833in}{2.066759in}}{\pgfqpoint{4.205833in}{2.088981in}}%
\pgfpathlineto{\pgfqpoint{4.205833in}{2.542536in}}%
\pgfpathquadraticcurveto{\pgfqpoint{4.205833in}{2.564759in}}{\pgfqpoint{4.183611in}{2.564759in}}%
\pgfpathlineto{\pgfqpoint{3.396277in}{2.564759in}}%
\pgfpathquadraticcurveto{\pgfqpoint{3.374055in}{2.564759in}}{\pgfqpoint{3.374055in}{2.542536in}}%
\pgfpathlineto{\pgfqpoint{3.374055in}{2.088981in}}%
\pgfpathquadraticcurveto{\pgfqpoint{3.374055in}{2.066759in}}{\pgfqpoint{3.396277in}{2.066759in}}%
\pgfpathclose%
\pgfusepath{stroke,fill}%
\end{pgfscope}%
\begin{pgfscope}%
\pgfsetrectcap%
\pgfsetroundjoin%
\pgfsetlinewidth{1.003750pt}%
\definecolor{currentstroke}{rgb}{0.000000,0.500000,0.000000}%
\pgfsetstrokecolor{currentstroke}%
\pgfsetdash{}{0pt}%
\pgfpathmoveto{\pgfqpoint{3.418500in}{2.481425in}}%
\pgfpathlineto{\pgfqpoint{3.640722in}{2.481425in}}%
\pgfusepath{stroke}%
\end{pgfscope}%
\begin{pgfscope}%
\pgftext[x=3.729611in,y=2.442536in,left,base]{\rmfamily\fontsize{8.000000}{9.600000}\selectfont baseline}%
\end{pgfscope}%
\begin{pgfscope}%
\pgfsetrectcap%
\pgfsetroundjoin%
\pgfsetlinewidth{1.003750pt}%
\definecolor{currentstroke}{rgb}{0.121569,0.466667,0.705882}%
\pgfsetstrokecolor{currentstroke}%
\pgfsetdash{}{0pt}%
\pgfpathmoveto{\pgfqpoint{3.418500in}{2.326537in}}%
\pgfpathlineto{\pgfqpoint{3.640722in}{2.326537in}}%
\pgfusepath{stroke}%
\end{pgfscope}%
\begin{pgfscope}%
\pgfsetbuttcap%
\pgfsetmiterjoin%
\definecolor{currentfill}{rgb}{0.121569,0.466667,0.705882}%
\pgfsetfillcolor{currentfill}%
\pgfsetlinewidth{1.003750pt}%
\definecolor{currentstroke}{rgb}{0.121569,0.466667,0.705882}%
\pgfsetstrokecolor{currentstroke}%
\pgfsetdash{}{0pt}%
\pgfsys@defobject{currentmarker}{\pgfqpoint{-0.027778in}{-0.027778in}}{\pgfqpoint{0.027778in}{0.027778in}}{%
\pgfpathmoveto{\pgfqpoint{0.000000in}{0.027778in}}%
\pgfpathlineto{\pgfqpoint{-0.027778in}{-0.027778in}}%
\pgfpathlineto{\pgfqpoint{0.027778in}{-0.027778in}}%
\pgfpathclose%
\pgfusepath{stroke,fill}%
}%
\begin{pgfscope}%
\pgfsys@transformshift{3.529611in}{2.326537in}%
\pgfsys@useobject{currentmarker}{}%
\end{pgfscope}%
\end{pgfscope}%
\begin{pgfscope}%
\pgftext[x=3.729611in,y=2.287648in,left,base]{\rmfamily\fontsize{8.000000}{9.600000}\selectfont SEVB}%
\end{pgfscope}%
\begin{pgfscope}%
\pgfsetrectcap%
\pgfsetroundjoin%
\pgfsetlinewidth{1.003750pt}%
\definecolor{currentstroke}{rgb}{1.000000,0.498039,0.054902}%
\pgfsetstrokecolor{currentstroke}%
\pgfsetdash{}{0pt}%
\pgfpathmoveto{\pgfqpoint{3.418500in}{2.171648in}}%
\pgfpathlineto{\pgfqpoint{3.640722in}{2.171648in}}%
\pgfusepath{stroke}%
\end{pgfscope}%
\begin{pgfscope}%
\pgfsetbuttcap%
\pgfsetbeveljoin%
\definecolor{currentfill}{rgb}{1.000000,0.498039,0.054902}%
\pgfsetfillcolor{currentfill}%
\pgfsetlinewidth{1.003750pt}%
\definecolor{currentstroke}{rgb}{1.000000,0.498039,0.054902}%
\pgfsetstrokecolor{currentstroke}%
\pgfsetdash{}{0pt}%
\pgfsys@defobject{currentmarker}{\pgfqpoint{-0.026418in}{-0.022473in}}{\pgfqpoint{0.026418in}{0.027778in}}{%
\pgfpathmoveto{\pgfqpoint{0.000000in}{0.027778in}}%
\pgfpathlineto{\pgfqpoint{-0.006236in}{0.008584in}}%
\pgfpathlineto{\pgfqpoint{-0.026418in}{0.008584in}}%
\pgfpathlineto{\pgfqpoint{-0.010091in}{-0.003279in}}%
\pgfpathlineto{\pgfqpoint{-0.016327in}{-0.022473in}}%
\pgfpathlineto{\pgfqpoint{-0.000000in}{-0.010610in}}%
\pgfpathlineto{\pgfqpoint{0.016327in}{-0.022473in}}%
\pgfpathlineto{\pgfqpoint{0.010091in}{-0.003279in}}%
\pgfpathlineto{\pgfqpoint{0.026418in}{0.008584in}}%
\pgfpathlineto{\pgfqpoint{0.006236in}{0.008584in}}%
\pgfpathclose%
\pgfusepath{stroke,fill}%
}%
\begin{pgfscope}%
\pgfsys@transformshift{3.529611in}{2.171648in}%
\pgfsys@useobject{currentmarker}{}%
\end{pgfscope}%
\end{pgfscope}%
\begin{pgfscope}%
\pgftext[x=3.729611in,y=2.132759in,left,base]{\rmfamily\fontsize{8.000000}{9.600000}\selectfont VCGEV}%
\end{pgfscope}%
\end{pgfpicture}%
\makeatother%
\endgroup%

%% file: 3t_sValue.pgf
%% Creator: Matplotlib, PGF backend
%%
%% To include the figure in your LaTeX document, write
%%   \input{<filename>.pgf}
%%
%% Make sure the required packages are loaded in your preamble
%%   \usepackage{pgf}
%%
%% Figures using additional raster images can only be included by \input if
%% they are in the same directory as the main LaTeX file. For loading figures
%% from other directories you can use the `import` package
%%   \usepackage{import}
%% and then include the figures with
%%   \import{<path to file>}{<filename>.pgf}
%%
%% Matplotlib used the following preamble
%%   \usepackage[utf8x]{inputenc}
%%   \usepackage[T1]{fontenc}
%%   \usepackage{fontspec}
%%
\begingroup%
\makeatletter%
\begin{pgfpicture}%
\pgfpathrectangle{\pgfpointorigin}{\pgfqpoint{4.296389in}{2.655314in}}%
\pgfusepath{use as bounding box, clip}%
\begin{pgfscope}%
\pgfsetbuttcap%
\pgfsetmiterjoin%
\definecolor{currentfill}{rgb}{1.000000,1.000000,1.000000}%
\pgfsetfillcolor{currentfill}%
\pgfsetlinewidth{0.000000pt}%
\definecolor{currentstroke}{rgb}{1.000000,1.000000,1.000000}%
\pgfsetstrokecolor{currentstroke}%
\pgfsetdash{}{0pt}%
\pgfpathmoveto{\pgfqpoint{0.000000in}{0.000000in}}%
\pgfpathlineto{\pgfqpoint{4.296389in}{0.000000in}}%
\pgfpathlineto{\pgfqpoint{4.296389in}{2.655314in}}%
\pgfpathlineto{\pgfqpoint{0.000000in}{2.655314in}}%
\pgfpathclose%
\pgfusepath{fill}%
\end{pgfscope}%
\begin{pgfscope}%
\pgfsetbuttcap%
\pgfsetmiterjoin%
\definecolor{currentfill}{rgb}{1.000000,1.000000,1.000000}%
\pgfsetfillcolor{currentfill}%
\pgfsetlinewidth{0.000000pt}%
\definecolor{currentstroke}{rgb}{0.000000,0.000000,0.000000}%
\pgfsetstrokecolor{currentstroke}%
\pgfsetstrokeopacity{0.000000}%
\pgfsetdash{}{0pt}%
\pgfpathmoveto{\pgfqpoint{0.307102in}{0.195889in}}%
\pgfpathlineto{\pgfqpoint{4.261389in}{0.195889in}}%
\pgfpathlineto{\pgfqpoint{4.261389in}{2.620314in}}%
\pgfpathlineto{\pgfqpoint{0.307102in}{2.620314in}}%
\pgfpathclose%
\pgfusepath{fill}%
\end{pgfscope}%
\begin{pgfscope}%
\pgfsetbuttcap%
\pgfsetroundjoin%
\definecolor{currentfill}{rgb}{0.000000,0.000000,0.000000}%
\pgfsetfillcolor{currentfill}%
\pgfsetlinewidth{0.803000pt}%
\definecolor{currentstroke}{rgb}{0.000000,0.000000,0.000000}%
\pgfsetstrokecolor{currentstroke}%
\pgfsetdash{}{0pt}%
\pgfsys@defobject{currentmarker}{\pgfqpoint{0.000000in}{-0.048611in}}{\pgfqpoint{0.000000in}{0.000000in}}{%
\pgfpathmoveto{\pgfqpoint{0.000000in}{0.000000in}}%
\pgfpathlineto{\pgfqpoint{0.000000in}{-0.048611in}}%
\pgfusepath{stroke,fill}%
}%
\begin{pgfscope}%
\pgfsys@transformshift{0.413479in}{0.195889in}%
\pgfsys@useobject{currentmarker}{}%
\end{pgfscope}%
\end{pgfscope}%
\begin{pgfscope}%
\pgftext[x=0.413479in,y=0.098667in,,top]{\rmfamily\fontsize{8.000000}{9.600000}\selectfont \(\displaystyle 0\)}%
\end{pgfscope}%
\begin{pgfscope}%
\pgfsetbuttcap%
\pgfsetroundjoin%
\definecolor{currentfill}{rgb}{0.000000,0.000000,0.000000}%
\pgfsetfillcolor{currentfill}%
\pgfsetlinewidth{0.803000pt}%
\definecolor{currentstroke}{rgb}{0.000000,0.000000,0.000000}%
\pgfsetstrokecolor{currentstroke}%
\pgfsetdash{}{0pt}%
\pgfsys@defobject{currentmarker}{\pgfqpoint{0.000000in}{-0.048611in}}{\pgfqpoint{0.000000in}{0.000000in}}{%
\pgfpathmoveto{\pgfqpoint{0.000000in}{0.000000in}}%
\pgfpathlineto{\pgfqpoint{0.000000in}{-0.048611in}}%
\pgfusepath{stroke,fill}%
}%
\begin{pgfscope}%
\pgfsys@transformshift{1.147113in}{0.195889in}%
\pgfsys@useobject{currentmarker}{}%
\end{pgfscope}%
\end{pgfscope}%
\begin{pgfscope}%
\pgftext[x=1.147113in,y=0.098667in,,top]{\rmfamily\fontsize{8.000000}{9.600000}\selectfont \(\displaystyle 10\)}%
\end{pgfscope}%
\begin{pgfscope}%
\pgfsetbuttcap%
\pgfsetroundjoin%
\definecolor{currentfill}{rgb}{0.000000,0.000000,0.000000}%
\pgfsetfillcolor{currentfill}%
\pgfsetlinewidth{0.803000pt}%
\definecolor{currentstroke}{rgb}{0.000000,0.000000,0.000000}%
\pgfsetstrokecolor{currentstroke}%
\pgfsetdash{}{0pt}%
\pgfsys@defobject{currentmarker}{\pgfqpoint{0.000000in}{-0.048611in}}{\pgfqpoint{0.000000in}{0.000000in}}{%
\pgfpathmoveto{\pgfqpoint{0.000000in}{0.000000in}}%
\pgfpathlineto{\pgfqpoint{0.000000in}{-0.048611in}}%
\pgfusepath{stroke,fill}%
}%
\begin{pgfscope}%
\pgfsys@transformshift{1.880747in}{0.195889in}%
\pgfsys@useobject{currentmarker}{}%
\end{pgfscope}%
\end{pgfscope}%
\begin{pgfscope}%
\pgftext[x=1.880747in,y=0.098667in,,top]{\rmfamily\fontsize{8.000000}{9.600000}\selectfont \(\displaystyle 20\)}%
\end{pgfscope}%
\begin{pgfscope}%
\pgfsetbuttcap%
\pgfsetroundjoin%
\definecolor{currentfill}{rgb}{0.000000,0.000000,0.000000}%
\pgfsetfillcolor{currentfill}%
\pgfsetlinewidth{0.803000pt}%
\definecolor{currentstroke}{rgb}{0.000000,0.000000,0.000000}%
\pgfsetstrokecolor{currentstroke}%
\pgfsetdash{}{0pt}%
\pgfsys@defobject{currentmarker}{\pgfqpoint{0.000000in}{-0.048611in}}{\pgfqpoint{0.000000in}{0.000000in}}{%
\pgfpathmoveto{\pgfqpoint{0.000000in}{0.000000in}}%
\pgfpathlineto{\pgfqpoint{0.000000in}{-0.048611in}}%
\pgfusepath{stroke,fill}%
}%
\begin{pgfscope}%
\pgfsys@transformshift{2.614380in}{0.195889in}%
\pgfsys@useobject{currentmarker}{}%
\end{pgfscope}%
\end{pgfscope}%
\begin{pgfscope}%
\pgftext[x=2.614380in,y=0.098667in,,top]{\rmfamily\fontsize{8.000000}{9.600000}\selectfont \(\displaystyle 30\)}%
\end{pgfscope}%
\begin{pgfscope}%
\pgfsetbuttcap%
\pgfsetroundjoin%
\definecolor{currentfill}{rgb}{0.000000,0.000000,0.000000}%
\pgfsetfillcolor{currentfill}%
\pgfsetlinewidth{0.803000pt}%
\definecolor{currentstroke}{rgb}{0.000000,0.000000,0.000000}%
\pgfsetstrokecolor{currentstroke}%
\pgfsetdash{}{0pt}%
\pgfsys@defobject{currentmarker}{\pgfqpoint{0.000000in}{-0.048611in}}{\pgfqpoint{0.000000in}{0.000000in}}{%
\pgfpathmoveto{\pgfqpoint{0.000000in}{0.000000in}}%
\pgfpathlineto{\pgfqpoint{0.000000in}{-0.048611in}}%
\pgfusepath{stroke,fill}%
}%
\begin{pgfscope}%
\pgfsys@transformshift{3.348014in}{0.195889in}%
\pgfsys@useobject{currentmarker}{}%
\end{pgfscope}%
\end{pgfscope}%
\begin{pgfscope}%
\pgftext[x=3.348014in,y=0.098667in,,top]{\rmfamily\fontsize{8.000000}{9.600000}\selectfont \(\displaystyle 40\)}%
\end{pgfscope}%
\begin{pgfscope}%
\pgfsetbuttcap%
\pgfsetroundjoin%
\definecolor{currentfill}{rgb}{0.000000,0.000000,0.000000}%
\pgfsetfillcolor{currentfill}%
\pgfsetlinewidth{0.803000pt}%
\definecolor{currentstroke}{rgb}{0.000000,0.000000,0.000000}%
\pgfsetstrokecolor{currentstroke}%
\pgfsetdash{}{0pt}%
\pgfsys@defobject{currentmarker}{\pgfqpoint{0.000000in}{-0.048611in}}{\pgfqpoint{0.000000in}{0.000000in}}{%
\pgfpathmoveto{\pgfqpoint{0.000000in}{0.000000in}}%
\pgfpathlineto{\pgfqpoint{0.000000in}{-0.048611in}}%
\pgfusepath{stroke,fill}%
}%
\begin{pgfscope}%
\pgfsys@transformshift{4.081648in}{0.195889in}%
\pgfsys@useobject{currentmarker}{}%
\end{pgfscope}%
\end{pgfscope}%
\begin{pgfscope}%
\pgftext[x=4.081648in,y=0.098667in,,top]{\rmfamily\fontsize{8.000000}{9.600000}\selectfont \(\displaystyle 50\)}%
\end{pgfscope}%
\begin{pgfscope}%
\pgfsetbuttcap%
\pgfsetroundjoin%
\definecolor{currentfill}{rgb}{0.000000,0.000000,0.000000}%
\pgfsetfillcolor{currentfill}%
\pgfsetlinewidth{0.803000pt}%
\definecolor{currentstroke}{rgb}{0.000000,0.000000,0.000000}%
\pgfsetstrokecolor{currentstroke}%
\pgfsetdash{}{0pt}%
\pgfsys@defobject{currentmarker}{\pgfqpoint{-0.048611in}{0.000000in}}{\pgfqpoint{0.000000in}{0.000000in}}{%
\pgfpathmoveto{\pgfqpoint{0.000000in}{0.000000in}}%
\pgfpathlineto{\pgfqpoint{-0.048611in}{0.000000in}}%
\pgfusepath{stroke,fill}%
}%
\begin{pgfscope}%
\pgfsys@transformshift{0.307102in}{0.233089in}%
\pgfsys@useobject{currentmarker}{}%
\end{pgfscope}%
\end{pgfscope}%
\begin{pgfscope}%
\pgftext[x=-0.000000in,y=0.194534in,left,base]{\rmfamily\fontsize{8.000000}{9.600000}\selectfont \(\displaystyle 0.00\)}%
\end{pgfscope}%
\begin{pgfscope}%
\pgfsetbuttcap%
\pgfsetroundjoin%
\definecolor{currentfill}{rgb}{0.000000,0.000000,0.000000}%
\pgfsetfillcolor{currentfill}%
\pgfsetlinewidth{0.803000pt}%
\definecolor{currentstroke}{rgb}{0.000000,0.000000,0.000000}%
\pgfsetstrokecolor{currentstroke}%
\pgfsetdash{}{0pt}%
\pgfsys@defobject{currentmarker}{\pgfqpoint{-0.048611in}{0.000000in}}{\pgfqpoint{0.000000in}{0.000000in}}{%
\pgfpathmoveto{\pgfqpoint{0.000000in}{0.000000in}}%
\pgfpathlineto{\pgfqpoint{-0.048611in}{0.000000in}}%
\pgfusepath{stroke,fill}%
}%
\begin{pgfscope}%
\pgfsys@transformshift{0.307102in}{0.692728in}%
\pgfsys@useobject{currentmarker}{}%
\end{pgfscope}%
\end{pgfscope}%
\begin{pgfscope}%
\pgftext[x=-0.000000in,y=0.654173in,left,base]{\rmfamily\fontsize{8.000000}{9.600000}\selectfont \(\displaystyle 0.05\)}%
\end{pgfscope}%
\begin{pgfscope}%
\pgfsetbuttcap%
\pgfsetroundjoin%
\definecolor{currentfill}{rgb}{0.000000,0.000000,0.000000}%
\pgfsetfillcolor{currentfill}%
\pgfsetlinewidth{0.803000pt}%
\definecolor{currentstroke}{rgb}{0.000000,0.000000,0.000000}%
\pgfsetstrokecolor{currentstroke}%
\pgfsetdash{}{0pt}%
\pgfsys@defobject{currentmarker}{\pgfqpoint{-0.048611in}{0.000000in}}{\pgfqpoint{0.000000in}{0.000000in}}{%
\pgfpathmoveto{\pgfqpoint{0.000000in}{0.000000in}}%
\pgfpathlineto{\pgfqpoint{-0.048611in}{0.000000in}}%
\pgfusepath{stroke,fill}%
}%
\begin{pgfscope}%
\pgfsys@transformshift{0.307102in}{1.152368in}%
\pgfsys@useobject{currentmarker}{}%
\end{pgfscope}%
\end{pgfscope}%
\begin{pgfscope}%
\pgftext[x=-0.000000in,y=1.113812in,left,base]{\rmfamily\fontsize{8.000000}{9.600000}\selectfont \(\displaystyle 0.10\)}%
\end{pgfscope}%
\begin{pgfscope}%
\pgfsetbuttcap%
\pgfsetroundjoin%
\definecolor{currentfill}{rgb}{0.000000,0.000000,0.000000}%
\pgfsetfillcolor{currentfill}%
\pgfsetlinewidth{0.803000pt}%
\definecolor{currentstroke}{rgb}{0.000000,0.000000,0.000000}%
\pgfsetstrokecolor{currentstroke}%
\pgfsetdash{}{0pt}%
\pgfsys@defobject{currentmarker}{\pgfqpoint{-0.048611in}{0.000000in}}{\pgfqpoint{0.000000in}{0.000000in}}{%
\pgfpathmoveto{\pgfqpoint{0.000000in}{0.000000in}}%
\pgfpathlineto{\pgfqpoint{-0.048611in}{0.000000in}}%
\pgfusepath{stroke,fill}%
}%
\begin{pgfscope}%
\pgfsys@transformshift{0.307102in}{1.612007in}%
\pgfsys@useobject{currentmarker}{}%
\end{pgfscope}%
\end{pgfscope}%
\begin{pgfscope}%
\pgftext[x=-0.000000in,y=1.573451in,left,base]{\rmfamily\fontsize{8.000000}{9.600000}\selectfont \(\displaystyle 0.15\)}%
\end{pgfscope}%
\begin{pgfscope}%
\pgfsetbuttcap%
\pgfsetroundjoin%
\definecolor{currentfill}{rgb}{0.000000,0.000000,0.000000}%
\pgfsetfillcolor{currentfill}%
\pgfsetlinewidth{0.803000pt}%
\definecolor{currentstroke}{rgb}{0.000000,0.000000,0.000000}%
\pgfsetstrokecolor{currentstroke}%
\pgfsetdash{}{0pt}%
\pgfsys@defobject{currentmarker}{\pgfqpoint{-0.048611in}{0.000000in}}{\pgfqpoint{0.000000in}{0.000000in}}{%
\pgfpathmoveto{\pgfqpoint{0.000000in}{0.000000in}}%
\pgfpathlineto{\pgfqpoint{-0.048611in}{0.000000in}}%
\pgfusepath{stroke,fill}%
}%
\begin{pgfscope}%
\pgfsys@transformshift{0.307102in}{2.071646in}%
\pgfsys@useobject{currentmarker}{}%
\end{pgfscope}%
\end{pgfscope}%
\begin{pgfscope}%
\pgftext[x=-0.000000in,y=2.033091in,left,base]{\rmfamily\fontsize{8.000000}{9.600000}\selectfont \(\displaystyle 0.20\)}%
\end{pgfscope}%
\begin{pgfscope}%
\pgfsetbuttcap%
\pgfsetroundjoin%
\definecolor{currentfill}{rgb}{0.000000,0.000000,0.000000}%
\pgfsetfillcolor{currentfill}%
\pgfsetlinewidth{0.803000pt}%
\definecolor{currentstroke}{rgb}{0.000000,0.000000,0.000000}%
\pgfsetstrokecolor{currentstroke}%
\pgfsetdash{}{0pt}%
\pgfsys@defobject{currentmarker}{\pgfqpoint{-0.048611in}{0.000000in}}{\pgfqpoint{0.000000in}{0.000000in}}{%
\pgfpathmoveto{\pgfqpoint{0.000000in}{0.000000in}}%
\pgfpathlineto{\pgfqpoint{-0.048611in}{0.000000in}}%
\pgfusepath{stroke,fill}%
}%
\begin{pgfscope}%
\pgfsys@transformshift{0.307102in}{2.531285in}%
\pgfsys@useobject{currentmarker}{}%
\end{pgfscope}%
\end{pgfscope}%
\begin{pgfscope}%
\pgftext[x=-0.000000in,y=2.492730in,left,base]{\rmfamily\fontsize{8.000000}{9.600000}\selectfont \(\displaystyle 0.25\)}%
\end{pgfscope}%
\begin{pgfscope}%
\pgfpathrectangle{\pgfqpoint{0.307102in}{0.195889in}}{\pgfqpoint{3.954287in}{2.424425in}}%
\pgfusepath{clip}%
\pgfsetrectcap%
\pgfsetroundjoin%
\pgfsetlinewidth{1.003750pt}%
\definecolor{currentstroke}{rgb}{0.121569,0.466667,0.705882}%
\pgfsetstrokecolor{currentstroke}%
\pgfsetdash{}{0pt}%
\pgfpathmoveto{\pgfqpoint{0.486842in}{1.680889in}}%
\pgfpathlineto{\pgfqpoint{0.560206in}{1.850787in}}%
\pgfpathlineto{\pgfqpoint{0.633569in}{2.253144in}}%
\pgfpathlineto{\pgfqpoint{0.706932in}{1.236128in}}%
\pgfpathlineto{\pgfqpoint{0.780296in}{1.798493in}}%
\pgfpathlineto{\pgfqpoint{0.853659in}{1.110904in}}%
\pgfpathlineto{\pgfqpoint{0.927022in}{1.115246in}}%
\pgfpathlineto{\pgfqpoint{1.000386in}{1.570838in}}%
\pgfpathlineto{\pgfqpoint{1.073749in}{1.789422in}}%
\pgfpathlineto{\pgfqpoint{1.147113in}{1.408813in}}%
\pgfpathlineto{\pgfqpoint{1.220476in}{1.529541in}}%
\pgfpathlineto{\pgfqpoint{1.293839in}{1.243448in}}%
\pgfpathlineto{\pgfqpoint{1.367203in}{2.002766in}}%
\pgfpathlineto{\pgfqpoint{1.440566in}{1.122478in}}%
\pgfpathlineto{\pgfqpoint{1.513930in}{0.525102in}}%
\pgfpathlineto{\pgfqpoint{1.587293in}{1.877204in}}%
\pgfpathlineto{\pgfqpoint{1.660656in}{1.728126in}}%
\pgfpathlineto{\pgfqpoint{1.734020in}{1.595080in}}%
\pgfpathlineto{\pgfqpoint{1.807383in}{1.769766in}}%
\pgfpathlineto{\pgfqpoint{1.880747in}{1.369018in}}%
\pgfpathlineto{\pgfqpoint{1.954110in}{0.570771in}}%
\pgfpathlineto{\pgfqpoint{2.027473in}{0.608975in}}%
\pgfpathlineto{\pgfqpoint{2.100837in}{0.601124in}}%
\pgfpathlineto{\pgfqpoint{2.174200in}{2.132901in}}%
\pgfpathlineto{\pgfqpoint{2.247563in}{1.064616in}}%
\pgfpathlineto{\pgfqpoint{2.320927in}{1.126674in}}%
\pgfpathlineto{\pgfqpoint{2.394290in}{1.536118in}}%
\pgfpathlineto{\pgfqpoint{2.467654in}{1.597557in}}%
\pgfpathlineto{\pgfqpoint{2.541017in}{1.902902in}}%
\pgfpathlineto{\pgfqpoint{2.614380in}{0.685411in}}%
\pgfpathlineto{\pgfqpoint{2.687744in}{0.781369in}}%
\pgfpathlineto{\pgfqpoint{2.761107in}{1.323259in}}%
\pgfpathlineto{\pgfqpoint{2.834471in}{2.102729in}}%
\pgfpathlineto{\pgfqpoint{2.907834in}{1.223716in}}%
\pgfpathlineto{\pgfqpoint{2.981197in}{0.688400in}}%
\pgfpathlineto{\pgfqpoint{3.054561in}{1.773803in}}%
\pgfpathlineto{\pgfqpoint{3.127924in}{0.424700in}}%
\pgfpathlineto{\pgfqpoint{3.201288in}{2.134687in}}%
\pgfpathlineto{\pgfqpoint{3.274651in}{1.507089in}}%
\pgfpathlineto{\pgfqpoint{3.348014in}{2.155077in}}%
\pgfpathlineto{\pgfqpoint{3.421378in}{1.195295in}}%
\pgfpathlineto{\pgfqpoint{3.494741in}{0.779782in}}%
\pgfpathlineto{\pgfqpoint{3.568105in}{2.094196in}}%
\pgfpathlineto{\pgfqpoint{3.641468in}{1.464484in}}%
\pgfpathlineto{\pgfqpoint{3.714831in}{1.345741in}}%
\pgfpathlineto{\pgfqpoint{3.788195in}{0.306090in}}%
\pgfpathlineto{\pgfqpoint{3.861558in}{1.754282in}}%
\pgfpathlineto{\pgfqpoint{3.934921in}{1.329989in}}%
\pgfpathlineto{\pgfqpoint{4.008285in}{0.754355in}}%
\pgfpathlineto{\pgfqpoint{4.081648in}{1.936041in}}%
\pgfusepath{stroke}%
\end{pgfscope}%
\begin{pgfscope}%
\pgfpathrectangle{\pgfqpoint{0.307102in}{0.195889in}}{\pgfqpoint{3.954287in}{2.424425in}}%
\pgfusepath{clip}%
\pgfsetbuttcap%
\pgfsetroundjoin%
\definecolor{currentfill}{rgb}{0.121569,0.466667,0.705882}%
\pgfsetfillcolor{currentfill}%
\pgfsetlinewidth{1.003750pt}%
\definecolor{currentstroke}{rgb}{0.121569,0.466667,0.705882}%
\pgfsetstrokecolor{currentstroke}%
\pgfsetdash{}{0pt}%
\pgfsys@defobject{currentmarker}{\pgfqpoint{-0.013889in}{-0.013889in}}{\pgfqpoint{0.013889in}{0.013889in}}{%
\pgfpathmoveto{\pgfqpoint{0.000000in}{-0.013889in}}%
\pgfpathcurveto{\pgfqpoint{0.003683in}{-0.013889in}}{\pgfqpoint{0.007216in}{-0.012425in}}{\pgfqpoint{0.009821in}{-0.009821in}}%
\pgfpathcurveto{\pgfqpoint{0.012425in}{-0.007216in}}{\pgfqpoint{0.013889in}{-0.003683in}}{\pgfqpoint{0.013889in}{0.000000in}}%
\pgfpathcurveto{\pgfqpoint{0.013889in}{0.003683in}}{\pgfqpoint{0.012425in}{0.007216in}}{\pgfqpoint{0.009821in}{0.009821in}}%
\pgfpathcurveto{\pgfqpoint{0.007216in}{0.012425in}}{\pgfqpoint{0.003683in}{0.013889in}}{\pgfqpoint{0.000000in}{0.013889in}}%
\pgfpathcurveto{\pgfqpoint{-0.003683in}{0.013889in}}{\pgfqpoint{-0.007216in}{0.012425in}}{\pgfqpoint{-0.009821in}{0.009821in}}%
\pgfpathcurveto{\pgfqpoint{-0.012425in}{0.007216in}}{\pgfqpoint{-0.013889in}{0.003683in}}{\pgfqpoint{-0.013889in}{0.000000in}}%
\pgfpathcurveto{\pgfqpoint{-0.013889in}{-0.003683in}}{\pgfqpoint{-0.012425in}{-0.007216in}}{\pgfqpoint{-0.009821in}{-0.009821in}}%
\pgfpathcurveto{\pgfqpoint{-0.007216in}{-0.012425in}}{\pgfqpoint{-0.003683in}{-0.013889in}}{\pgfqpoint{0.000000in}{-0.013889in}}%
\pgfpathclose%
\pgfusepath{stroke,fill}%
}%
\begin{pgfscope}%
\pgfsys@transformshift{0.486842in}{1.680889in}%
\pgfsys@useobject{currentmarker}{}%
\end{pgfscope}%
\begin{pgfscope}%
\pgfsys@transformshift{0.560206in}{1.850787in}%
\pgfsys@useobject{currentmarker}{}%
\end{pgfscope}%
\begin{pgfscope}%
\pgfsys@transformshift{0.633569in}{2.253144in}%
\pgfsys@useobject{currentmarker}{}%
\end{pgfscope}%
\begin{pgfscope}%
\pgfsys@transformshift{0.706932in}{1.236128in}%
\pgfsys@useobject{currentmarker}{}%
\end{pgfscope}%
\begin{pgfscope}%
\pgfsys@transformshift{0.780296in}{1.798493in}%
\pgfsys@useobject{currentmarker}{}%
\end{pgfscope}%
\begin{pgfscope}%
\pgfsys@transformshift{0.853659in}{1.110904in}%
\pgfsys@useobject{currentmarker}{}%
\end{pgfscope}%
\begin{pgfscope}%
\pgfsys@transformshift{0.927022in}{1.115246in}%
\pgfsys@useobject{currentmarker}{}%
\end{pgfscope}%
\begin{pgfscope}%
\pgfsys@transformshift{1.000386in}{1.570838in}%
\pgfsys@useobject{currentmarker}{}%
\end{pgfscope}%
\begin{pgfscope}%
\pgfsys@transformshift{1.073749in}{1.789422in}%
\pgfsys@useobject{currentmarker}{}%
\end{pgfscope}%
\begin{pgfscope}%
\pgfsys@transformshift{1.147113in}{1.408813in}%
\pgfsys@useobject{currentmarker}{}%
\end{pgfscope}%
\begin{pgfscope}%
\pgfsys@transformshift{1.220476in}{1.529541in}%
\pgfsys@useobject{currentmarker}{}%
\end{pgfscope}%
\begin{pgfscope}%
\pgfsys@transformshift{1.293839in}{1.243448in}%
\pgfsys@useobject{currentmarker}{}%
\end{pgfscope}%
\begin{pgfscope}%
\pgfsys@transformshift{1.367203in}{2.002766in}%
\pgfsys@useobject{currentmarker}{}%
\end{pgfscope}%
\begin{pgfscope}%
\pgfsys@transformshift{1.440566in}{1.122478in}%
\pgfsys@useobject{currentmarker}{}%
\end{pgfscope}%
\begin{pgfscope}%
\pgfsys@transformshift{1.513930in}{0.525102in}%
\pgfsys@useobject{currentmarker}{}%
\end{pgfscope}%
\begin{pgfscope}%
\pgfsys@transformshift{1.587293in}{1.877204in}%
\pgfsys@useobject{currentmarker}{}%
\end{pgfscope}%
\begin{pgfscope}%
\pgfsys@transformshift{1.660656in}{1.728126in}%
\pgfsys@useobject{currentmarker}{}%
\end{pgfscope}%
\begin{pgfscope}%
\pgfsys@transformshift{1.734020in}{1.595080in}%
\pgfsys@useobject{currentmarker}{}%
\end{pgfscope}%
\begin{pgfscope}%
\pgfsys@transformshift{1.807383in}{1.769766in}%
\pgfsys@useobject{currentmarker}{}%
\end{pgfscope}%
\begin{pgfscope}%
\pgfsys@transformshift{1.880747in}{1.369018in}%
\pgfsys@useobject{currentmarker}{}%
\end{pgfscope}%
\begin{pgfscope}%
\pgfsys@transformshift{1.954110in}{0.570771in}%
\pgfsys@useobject{currentmarker}{}%
\end{pgfscope}%
\begin{pgfscope}%
\pgfsys@transformshift{2.027473in}{0.608975in}%
\pgfsys@useobject{currentmarker}{}%
\end{pgfscope}%
\begin{pgfscope}%
\pgfsys@transformshift{2.100837in}{0.601124in}%
\pgfsys@useobject{currentmarker}{}%
\end{pgfscope}%
\begin{pgfscope}%
\pgfsys@transformshift{2.174200in}{2.132901in}%
\pgfsys@useobject{currentmarker}{}%
\end{pgfscope}%
\begin{pgfscope}%
\pgfsys@transformshift{2.247563in}{1.064616in}%
\pgfsys@useobject{currentmarker}{}%
\end{pgfscope}%
\begin{pgfscope}%
\pgfsys@transformshift{2.320927in}{1.126674in}%
\pgfsys@useobject{currentmarker}{}%
\end{pgfscope}%
\begin{pgfscope}%
\pgfsys@transformshift{2.394290in}{1.536118in}%
\pgfsys@useobject{currentmarker}{}%
\end{pgfscope}%
\begin{pgfscope}%
\pgfsys@transformshift{2.467654in}{1.597557in}%
\pgfsys@useobject{currentmarker}{}%
\end{pgfscope}%
\begin{pgfscope}%
\pgfsys@transformshift{2.541017in}{1.902902in}%
\pgfsys@useobject{currentmarker}{}%
\end{pgfscope}%
\begin{pgfscope}%
\pgfsys@transformshift{2.614380in}{0.685411in}%
\pgfsys@useobject{currentmarker}{}%
\end{pgfscope}%
\begin{pgfscope}%
\pgfsys@transformshift{2.687744in}{0.781369in}%
\pgfsys@useobject{currentmarker}{}%
\end{pgfscope}%
\begin{pgfscope}%
\pgfsys@transformshift{2.761107in}{1.323259in}%
\pgfsys@useobject{currentmarker}{}%
\end{pgfscope}%
\begin{pgfscope}%
\pgfsys@transformshift{2.834471in}{2.102729in}%
\pgfsys@useobject{currentmarker}{}%
\end{pgfscope}%
\begin{pgfscope}%
\pgfsys@transformshift{2.907834in}{1.223716in}%
\pgfsys@useobject{currentmarker}{}%
\end{pgfscope}%
\begin{pgfscope}%
\pgfsys@transformshift{2.981197in}{0.688400in}%
\pgfsys@useobject{currentmarker}{}%
\end{pgfscope}%
\begin{pgfscope}%
\pgfsys@transformshift{3.054561in}{1.773803in}%
\pgfsys@useobject{currentmarker}{}%
\end{pgfscope}%
\begin{pgfscope}%
\pgfsys@transformshift{3.127924in}{0.424700in}%
\pgfsys@useobject{currentmarker}{}%
\end{pgfscope}%
\begin{pgfscope}%
\pgfsys@transformshift{3.201288in}{2.134687in}%
\pgfsys@useobject{currentmarker}{}%
\end{pgfscope}%
\begin{pgfscope}%
\pgfsys@transformshift{3.274651in}{1.507089in}%
\pgfsys@useobject{currentmarker}{}%
\end{pgfscope}%
\begin{pgfscope}%
\pgfsys@transformshift{3.348014in}{2.155077in}%
\pgfsys@useobject{currentmarker}{}%
\end{pgfscope}%
\begin{pgfscope}%
\pgfsys@transformshift{3.421378in}{1.195295in}%
\pgfsys@useobject{currentmarker}{}%
\end{pgfscope}%
\begin{pgfscope}%
\pgfsys@transformshift{3.494741in}{0.779782in}%
\pgfsys@useobject{currentmarker}{}%
\end{pgfscope}%
\begin{pgfscope}%
\pgfsys@transformshift{3.568105in}{2.094196in}%
\pgfsys@useobject{currentmarker}{}%
\end{pgfscope}%
\begin{pgfscope}%
\pgfsys@transformshift{3.641468in}{1.464484in}%
\pgfsys@useobject{currentmarker}{}%
\end{pgfscope}%
\begin{pgfscope}%
\pgfsys@transformshift{3.714831in}{1.345741in}%
\pgfsys@useobject{currentmarker}{}%
\end{pgfscope}%
\begin{pgfscope}%
\pgfsys@transformshift{3.788195in}{0.306090in}%
\pgfsys@useobject{currentmarker}{}%
\end{pgfscope}%
\begin{pgfscope}%
\pgfsys@transformshift{3.861558in}{1.754282in}%
\pgfsys@useobject{currentmarker}{}%
\end{pgfscope}%
\begin{pgfscope}%
\pgfsys@transformshift{3.934921in}{1.329989in}%
\pgfsys@useobject{currentmarker}{}%
\end{pgfscope}%
\begin{pgfscope}%
\pgfsys@transformshift{4.008285in}{0.754355in}%
\pgfsys@useobject{currentmarker}{}%
\end{pgfscope}%
\begin{pgfscope}%
\pgfsys@transformshift{4.081648in}{1.936041in}%
\pgfsys@useobject{currentmarker}{}%
\end{pgfscope}%
\end{pgfscope}%
\begin{pgfscope}%
\pgfpathrectangle{\pgfqpoint{0.307102in}{0.195889in}}{\pgfqpoint{3.954287in}{2.424425in}}%
\pgfusepath{clip}%
\pgfsetrectcap%
\pgfsetroundjoin%
\pgfsetlinewidth{1.003750pt}%
\definecolor{currentstroke}{rgb}{1.000000,0.498039,0.054902}%
\pgfsetstrokecolor{currentstroke}%
\pgfsetdash{}{0pt}%
\pgfpathmoveto{\pgfqpoint{0.486842in}{1.806202in}}%
\pgfpathlineto{\pgfqpoint{0.560206in}{1.164664in}}%
\pgfpathlineto{\pgfqpoint{0.633569in}{1.713449in}}%
\pgfpathlineto{\pgfqpoint{0.706932in}{1.077418in}}%
\pgfpathlineto{\pgfqpoint{0.780296in}{2.091421in}}%
\pgfpathlineto{\pgfqpoint{0.853659in}{0.879860in}}%
\pgfpathlineto{\pgfqpoint{0.927022in}{0.583394in}}%
\pgfpathlineto{\pgfqpoint{1.000386in}{0.595187in}}%
\pgfpathlineto{\pgfqpoint{1.073749in}{1.561659in}}%
\pgfpathlineto{\pgfqpoint{1.147113in}{0.521887in}}%
\pgfpathlineto{\pgfqpoint{1.220476in}{1.404488in}}%
\pgfpathlineto{\pgfqpoint{1.293839in}{2.510113in}}%
\pgfpathlineto{\pgfqpoint{1.367203in}{1.129321in}}%
\pgfpathlineto{\pgfqpoint{1.440566in}{1.064224in}}%
\pgfpathlineto{\pgfqpoint{1.513930in}{1.074718in}}%
\pgfpathlineto{\pgfqpoint{1.587293in}{1.730839in}}%
\pgfpathlineto{\pgfqpoint{1.660656in}{0.719409in}}%
\pgfpathlineto{\pgfqpoint{1.734020in}{1.831307in}}%
\pgfpathlineto{\pgfqpoint{1.807383in}{1.440934in}}%
\pgfpathlineto{\pgfqpoint{1.880747in}{0.523311in}}%
\pgfpathlineto{\pgfqpoint{1.954110in}{1.297988in}}%
\pgfpathlineto{\pgfqpoint{2.027473in}{2.004090in}}%
\pgfpathlineto{\pgfqpoint{2.100837in}{1.709621in}}%
\pgfpathlineto{\pgfqpoint{2.174200in}{1.688148in}}%
\pgfpathlineto{\pgfqpoint{2.247563in}{1.882006in}}%
\pgfpathlineto{\pgfqpoint{2.320927in}{1.137560in}}%
\pgfpathlineto{\pgfqpoint{2.394290in}{1.784670in}}%
\pgfpathlineto{\pgfqpoint{2.467654in}{1.931490in}}%
\pgfpathlineto{\pgfqpoint{2.541017in}{1.415734in}}%
\pgfpathlineto{\pgfqpoint{2.614380in}{0.582884in}}%
\pgfpathlineto{\pgfqpoint{2.687744in}{2.145481in}}%
\pgfpathlineto{\pgfqpoint{2.761107in}{2.359688in}}%
\pgfpathlineto{\pgfqpoint{2.834471in}{1.693205in}}%
\pgfpathlineto{\pgfqpoint{2.907834in}{0.816279in}}%
\pgfpathlineto{\pgfqpoint{2.981197in}{1.637330in}}%
\pgfpathlineto{\pgfqpoint{3.054561in}{1.799726in}}%
\pgfpathlineto{\pgfqpoint{3.127924in}{1.327342in}}%
\pgfpathlineto{\pgfqpoint{3.201288in}{0.912926in}}%
\pgfpathlineto{\pgfqpoint{3.274651in}{1.839258in}}%
\pgfpathlineto{\pgfqpoint{3.348014in}{1.099107in}}%
\pgfpathlineto{\pgfqpoint{3.421378in}{0.722335in}}%
\pgfpathlineto{\pgfqpoint{3.494741in}{1.587998in}}%
\pgfpathlineto{\pgfqpoint{3.568105in}{1.482430in}}%
\pgfpathlineto{\pgfqpoint{3.641468in}{1.645613in}}%
\pgfpathlineto{\pgfqpoint{3.714831in}{1.779836in}}%
\pgfpathlineto{\pgfqpoint{3.788195in}{1.846530in}}%
\pgfpathlineto{\pgfqpoint{3.861558in}{1.861641in}}%
\pgfpathlineto{\pgfqpoint{3.934921in}{1.668408in}}%
\pgfpathlineto{\pgfqpoint{4.008285in}{1.322964in}}%
\pgfpathlineto{\pgfqpoint{4.081648in}{1.346751in}}%
\pgfusepath{stroke}%
\end{pgfscope}%
\begin{pgfscope}%
\pgfpathrectangle{\pgfqpoint{0.307102in}{0.195889in}}{\pgfqpoint{3.954287in}{2.424425in}}%
\pgfusepath{clip}%
\pgfsetbuttcap%
\pgfsetbeveljoin%
\definecolor{currentfill}{rgb}{1.000000,0.498039,0.054902}%
\pgfsetfillcolor{currentfill}%
\pgfsetlinewidth{1.003750pt}%
\definecolor{currentstroke}{rgb}{1.000000,0.498039,0.054902}%
\pgfsetstrokecolor{currentstroke}%
\pgfsetdash{}{0pt}%
\pgfsys@defobject{currentmarker}{\pgfqpoint{-0.026418in}{-0.022473in}}{\pgfqpoint{0.026418in}{0.027778in}}{%
\pgfpathmoveto{\pgfqpoint{0.000000in}{0.027778in}}%
\pgfpathlineto{\pgfqpoint{-0.006236in}{0.008584in}}%
\pgfpathlineto{\pgfqpoint{-0.026418in}{0.008584in}}%
\pgfpathlineto{\pgfqpoint{-0.010091in}{-0.003279in}}%
\pgfpathlineto{\pgfqpoint{-0.016327in}{-0.022473in}}%
\pgfpathlineto{\pgfqpoint{-0.000000in}{-0.010610in}}%
\pgfpathlineto{\pgfqpoint{0.016327in}{-0.022473in}}%
\pgfpathlineto{\pgfqpoint{0.010091in}{-0.003279in}}%
\pgfpathlineto{\pgfqpoint{0.026418in}{0.008584in}}%
\pgfpathlineto{\pgfqpoint{0.006236in}{0.008584in}}%
\pgfpathclose%
\pgfusepath{stroke,fill}%
}%
\begin{pgfscope}%
\pgfsys@transformshift{0.486842in}{1.806202in}%
\pgfsys@useobject{currentmarker}{}%
\end{pgfscope}%
\begin{pgfscope}%
\pgfsys@transformshift{0.560206in}{1.164664in}%
\pgfsys@useobject{currentmarker}{}%
\end{pgfscope}%
\begin{pgfscope}%
\pgfsys@transformshift{0.633569in}{1.713449in}%
\pgfsys@useobject{currentmarker}{}%
\end{pgfscope}%
\begin{pgfscope}%
\pgfsys@transformshift{0.706932in}{1.077418in}%
\pgfsys@useobject{currentmarker}{}%
\end{pgfscope}%
\begin{pgfscope}%
\pgfsys@transformshift{0.780296in}{2.091421in}%
\pgfsys@useobject{currentmarker}{}%
\end{pgfscope}%
\begin{pgfscope}%
\pgfsys@transformshift{0.853659in}{0.879860in}%
\pgfsys@useobject{currentmarker}{}%
\end{pgfscope}%
\begin{pgfscope}%
\pgfsys@transformshift{0.927022in}{0.583394in}%
\pgfsys@useobject{currentmarker}{}%
\end{pgfscope}%
\begin{pgfscope}%
\pgfsys@transformshift{1.000386in}{0.595187in}%
\pgfsys@useobject{currentmarker}{}%
\end{pgfscope}%
\begin{pgfscope}%
\pgfsys@transformshift{1.073749in}{1.561659in}%
\pgfsys@useobject{currentmarker}{}%
\end{pgfscope}%
\begin{pgfscope}%
\pgfsys@transformshift{1.147113in}{0.521887in}%
\pgfsys@useobject{currentmarker}{}%
\end{pgfscope}%
\begin{pgfscope}%
\pgfsys@transformshift{1.220476in}{1.404488in}%
\pgfsys@useobject{currentmarker}{}%
\end{pgfscope}%
\begin{pgfscope}%
\pgfsys@transformshift{1.293839in}{2.510113in}%
\pgfsys@useobject{currentmarker}{}%
\end{pgfscope}%
\begin{pgfscope}%
\pgfsys@transformshift{1.367203in}{1.129321in}%
\pgfsys@useobject{currentmarker}{}%
\end{pgfscope}%
\begin{pgfscope}%
\pgfsys@transformshift{1.440566in}{1.064224in}%
\pgfsys@useobject{currentmarker}{}%
\end{pgfscope}%
\begin{pgfscope}%
\pgfsys@transformshift{1.513930in}{1.074718in}%
\pgfsys@useobject{currentmarker}{}%
\end{pgfscope}%
\begin{pgfscope}%
\pgfsys@transformshift{1.587293in}{1.730839in}%
\pgfsys@useobject{currentmarker}{}%
\end{pgfscope}%
\begin{pgfscope}%
\pgfsys@transformshift{1.660656in}{0.719409in}%
\pgfsys@useobject{currentmarker}{}%
\end{pgfscope}%
\begin{pgfscope}%
\pgfsys@transformshift{1.734020in}{1.831307in}%
\pgfsys@useobject{currentmarker}{}%
\end{pgfscope}%
\begin{pgfscope}%
\pgfsys@transformshift{1.807383in}{1.440934in}%
\pgfsys@useobject{currentmarker}{}%
\end{pgfscope}%
\begin{pgfscope}%
\pgfsys@transformshift{1.880747in}{0.523311in}%
\pgfsys@useobject{currentmarker}{}%
\end{pgfscope}%
\begin{pgfscope}%
\pgfsys@transformshift{1.954110in}{1.297988in}%
\pgfsys@useobject{currentmarker}{}%
\end{pgfscope}%
\begin{pgfscope}%
\pgfsys@transformshift{2.027473in}{2.004090in}%
\pgfsys@useobject{currentmarker}{}%
\end{pgfscope}%
\begin{pgfscope}%
\pgfsys@transformshift{2.100837in}{1.709621in}%
\pgfsys@useobject{currentmarker}{}%
\end{pgfscope}%
\begin{pgfscope}%
\pgfsys@transformshift{2.174200in}{1.688148in}%
\pgfsys@useobject{currentmarker}{}%
\end{pgfscope}%
\begin{pgfscope}%
\pgfsys@transformshift{2.247563in}{1.882006in}%
\pgfsys@useobject{currentmarker}{}%
\end{pgfscope}%
\begin{pgfscope}%
\pgfsys@transformshift{2.320927in}{1.137560in}%
\pgfsys@useobject{currentmarker}{}%
\end{pgfscope}%
\begin{pgfscope}%
\pgfsys@transformshift{2.394290in}{1.784670in}%
\pgfsys@useobject{currentmarker}{}%
\end{pgfscope}%
\begin{pgfscope}%
\pgfsys@transformshift{2.467654in}{1.931490in}%
\pgfsys@useobject{currentmarker}{}%
\end{pgfscope}%
\begin{pgfscope}%
\pgfsys@transformshift{2.541017in}{1.415734in}%
\pgfsys@useobject{currentmarker}{}%
\end{pgfscope}%
\begin{pgfscope}%
\pgfsys@transformshift{2.614380in}{0.582884in}%
\pgfsys@useobject{currentmarker}{}%
\end{pgfscope}%
\begin{pgfscope}%
\pgfsys@transformshift{2.687744in}{2.145481in}%
\pgfsys@useobject{currentmarker}{}%
\end{pgfscope}%
\begin{pgfscope}%
\pgfsys@transformshift{2.761107in}{2.359688in}%
\pgfsys@useobject{currentmarker}{}%
\end{pgfscope}%
\begin{pgfscope}%
\pgfsys@transformshift{2.834471in}{1.693205in}%
\pgfsys@useobject{currentmarker}{}%
\end{pgfscope}%
\begin{pgfscope}%
\pgfsys@transformshift{2.907834in}{0.816279in}%
\pgfsys@useobject{currentmarker}{}%
\end{pgfscope}%
\begin{pgfscope}%
\pgfsys@transformshift{2.981197in}{1.637330in}%
\pgfsys@useobject{currentmarker}{}%
\end{pgfscope}%
\begin{pgfscope}%
\pgfsys@transformshift{3.054561in}{1.799726in}%
\pgfsys@useobject{currentmarker}{}%
\end{pgfscope}%
\begin{pgfscope}%
\pgfsys@transformshift{3.127924in}{1.327342in}%
\pgfsys@useobject{currentmarker}{}%
\end{pgfscope}%
\begin{pgfscope}%
\pgfsys@transformshift{3.201288in}{0.912926in}%
\pgfsys@useobject{currentmarker}{}%
\end{pgfscope}%
\begin{pgfscope}%
\pgfsys@transformshift{3.274651in}{1.839258in}%
\pgfsys@useobject{currentmarker}{}%
\end{pgfscope}%
\begin{pgfscope}%
\pgfsys@transformshift{3.348014in}{1.099107in}%
\pgfsys@useobject{currentmarker}{}%
\end{pgfscope}%
\begin{pgfscope}%
\pgfsys@transformshift{3.421378in}{0.722335in}%
\pgfsys@useobject{currentmarker}{}%
\end{pgfscope}%
\begin{pgfscope}%
\pgfsys@transformshift{3.494741in}{1.587998in}%
\pgfsys@useobject{currentmarker}{}%
\end{pgfscope}%
\begin{pgfscope}%
\pgfsys@transformshift{3.568105in}{1.482430in}%
\pgfsys@useobject{currentmarker}{}%
\end{pgfscope}%
\begin{pgfscope}%
\pgfsys@transformshift{3.641468in}{1.645613in}%
\pgfsys@useobject{currentmarker}{}%
\end{pgfscope}%
\begin{pgfscope}%
\pgfsys@transformshift{3.714831in}{1.779836in}%
\pgfsys@useobject{currentmarker}{}%
\end{pgfscope}%
\begin{pgfscope}%
\pgfsys@transformshift{3.788195in}{1.846530in}%
\pgfsys@useobject{currentmarker}{}%
\end{pgfscope}%
\begin{pgfscope}%
\pgfsys@transformshift{3.861558in}{1.861641in}%
\pgfsys@useobject{currentmarker}{}%
\end{pgfscope}%
\begin{pgfscope}%
\pgfsys@transformshift{3.934921in}{1.668408in}%
\pgfsys@useobject{currentmarker}{}%
\end{pgfscope}%
\begin{pgfscope}%
\pgfsys@transformshift{4.008285in}{1.322964in}%
\pgfsys@useobject{currentmarker}{}%
\end{pgfscope}%
\begin{pgfscope}%
\pgfsys@transformshift{4.081648in}{1.346751in}%
\pgfsys@useobject{currentmarker}{}%
\end{pgfscope}%
\end{pgfscope}%
\begin{pgfscope}%
\pgfpathrectangle{\pgfqpoint{0.307102in}{0.195889in}}{\pgfqpoint{3.954287in}{2.424425in}}%
\pgfusepath{clip}%
\pgfsetrectcap%
\pgfsetroundjoin%
\pgfsetlinewidth{1.003750pt}%
\definecolor{currentstroke}{rgb}{0.172549,0.627451,0.172549}%
\pgfsetstrokecolor{currentstroke}%
\pgfsetdash{}{0pt}%
\pgfpathmoveto{\pgfqpoint{0.486842in}{0.827177in}}%
\pgfpathlineto{\pgfqpoint{0.560206in}{1.104113in}}%
\pgfpathlineto{\pgfqpoint{0.633569in}{0.629856in}}%
\pgfpathlineto{\pgfqpoint{0.706932in}{1.028256in}}%
\pgfpathlineto{\pgfqpoint{0.780296in}{0.568901in}}%
\pgfpathlineto{\pgfqpoint{0.853659in}{1.196143in}}%
\pgfpathlineto{\pgfqpoint{0.927022in}{0.645672in}}%
\pgfpathlineto{\pgfqpoint{1.000386in}{0.334125in}}%
\pgfpathlineto{\pgfqpoint{1.073749in}{1.009066in}}%
\pgfpathlineto{\pgfqpoint{1.147113in}{0.364273in}}%
\pgfpathlineto{\pgfqpoint{1.220476in}{1.196143in}}%
\pgfpathlineto{\pgfqpoint{1.293839in}{0.852124in}}%
\pgfpathlineto{\pgfqpoint{1.367203in}{0.340552in}}%
\pgfpathlineto{\pgfqpoint{1.440566in}{0.438390in}}%
\pgfpathlineto{\pgfqpoint{1.513930in}{0.579564in}}%
\pgfpathlineto{\pgfqpoint{1.587293in}{0.388238in}}%
\pgfpathlineto{\pgfqpoint{1.660656in}{0.647151in}}%
\pgfpathlineto{\pgfqpoint{1.734020in}{0.542366in}}%
\pgfpathlineto{\pgfqpoint{1.807383in}{0.355051in}}%
\pgfpathlineto{\pgfqpoint{1.880747in}{0.743463in}}%
\pgfpathlineto{\pgfqpoint{1.954110in}{0.520286in}}%
\pgfpathlineto{\pgfqpoint{2.027473in}{0.964618in}}%
\pgfpathlineto{\pgfqpoint{2.100837in}{1.083540in}}%
\pgfpathlineto{\pgfqpoint{2.174200in}{0.557887in}}%
\pgfpathlineto{\pgfqpoint{2.247563in}{0.751655in}}%
\pgfpathlineto{\pgfqpoint{2.320927in}{1.025445in}}%
\pgfpathlineto{\pgfqpoint{2.394290in}{0.754219in}}%
\pgfpathlineto{\pgfqpoint{2.467654in}{1.229279in}}%
\pgfpathlineto{\pgfqpoint{2.541017in}{0.403555in}}%
\pgfpathlineto{\pgfqpoint{2.614380in}{1.196143in}}%
\pgfpathlineto{\pgfqpoint{2.687744in}{0.497258in}}%
\pgfpathlineto{\pgfqpoint{2.761107in}{0.996269in}}%
\pgfpathlineto{\pgfqpoint{2.834471in}{1.196143in}}%
\pgfpathlineto{\pgfqpoint{2.907834in}{0.849613in}}%
\pgfpathlineto{\pgfqpoint{2.981197in}{0.437872in}}%
\pgfpathlineto{\pgfqpoint{3.054561in}{0.937068in}}%
\pgfpathlineto{\pgfqpoint{3.127924in}{0.341036in}}%
\pgfpathlineto{\pgfqpoint{3.201288in}{0.310796in}}%
\pgfpathlineto{\pgfqpoint{3.274651in}{0.667747in}}%
\pgfpathlineto{\pgfqpoint{3.348014in}{0.552721in}}%
\pgfpathlineto{\pgfqpoint{3.421378in}{0.345899in}}%
\pgfpathlineto{\pgfqpoint{3.494741in}{0.920449in}}%
\pgfpathlineto{\pgfqpoint{3.568105in}{1.196143in}}%
\pgfpathlineto{\pgfqpoint{3.641468in}{0.722522in}}%
\pgfpathlineto{\pgfqpoint{3.714831in}{0.909456in}}%
\pgfpathlineto{\pgfqpoint{3.788195in}{0.389675in}}%
\pgfpathlineto{\pgfqpoint{3.861558in}{1.196143in}}%
\pgfpathlineto{\pgfqpoint{3.934921in}{0.819679in}}%
\pgfpathlineto{\pgfqpoint{4.008285in}{0.813074in}}%
\pgfpathlineto{\pgfqpoint{4.081648in}{1.124307in}}%
\pgfusepath{stroke}%
\end{pgfscope}%
\begin{pgfscope}%
\pgfpathrectangle{\pgfqpoint{0.307102in}{0.195889in}}{\pgfqpoint{3.954287in}{2.424425in}}%
\pgfusepath{clip}%
\pgfsetbuttcap%
\pgfsetmiterjoin%
\definecolor{currentfill}{rgb}{0.172549,0.627451,0.172549}%
\pgfsetfillcolor{currentfill}%
\pgfsetlinewidth{1.003750pt}%
\definecolor{currentstroke}{rgb}{0.172549,0.627451,0.172549}%
\pgfsetstrokecolor{currentstroke}%
\pgfsetdash{}{0pt}%
\pgfsys@defobject{currentmarker}{\pgfqpoint{-0.027778in}{-0.027778in}}{\pgfqpoint{0.027778in}{0.027778in}}{%
\pgfpathmoveto{\pgfqpoint{0.000000in}{0.027778in}}%
\pgfpathlineto{\pgfqpoint{-0.027778in}{-0.027778in}}%
\pgfpathlineto{\pgfqpoint{0.027778in}{-0.027778in}}%
\pgfpathclose%
\pgfusepath{stroke,fill}%
}%
\begin{pgfscope}%
\pgfsys@transformshift{0.486842in}{0.827177in}%
\pgfsys@useobject{currentmarker}{}%
\end{pgfscope}%
\begin{pgfscope}%
\pgfsys@transformshift{0.560206in}{1.104113in}%
\pgfsys@useobject{currentmarker}{}%
\end{pgfscope}%
\begin{pgfscope}%
\pgfsys@transformshift{0.633569in}{0.629856in}%
\pgfsys@useobject{currentmarker}{}%
\end{pgfscope}%
\begin{pgfscope}%
\pgfsys@transformshift{0.706932in}{1.028256in}%
\pgfsys@useobject{currentmarker}{}%
\end{pgfscope}%
\begin{pgfscope}%
\pgfsys@transformshift{0.780296in}{0.568901in}%
\pgfsys@useobject{currentmarker}{}%
\end{pgfscope}%
\begin{pgfscope}%
\pgfsys@transformshift{0.853659in}{1.196143in}%
\pgfsys@useobject{currentmarker}{}%
\end{pgfscope}%
\begin{pgfscope}%
\pgfsys@transformshift{0.927022in}{0.645672in}%
\pgfsys@useobject{currentmarker}{}%
\end{pgfscope}%
\begin{pgfscope}%
\pgfsys@transformshift{1.000386in}{0.334125in}%
\pgfsys@useobject{currentmarker}{}%
\end{pgfscope}%
\begin{pgfscope}%
\pgfsys@transformshift{1.073749in}{1.009066in}%
\pgfsys@useobject{currentmarker}{}%
\end{pgfscope}%
\begin{pgfscope}%
\pgfsys@transformshift{1.147113in}{0.364273in}%
\pgfsys@useobject{currentmarker}{}%
\end{pgfscope}%
\begin{pgfscope}%
\pgfsys@transformshift{1.220476in}{1.196143in}%
\pgfsys@useobject{currentmarker}{}%
\end{pgfscope}%
\begin{pgfscope}%
\pgfsys@transformshift{1.293839in}{0.852124in}%
\pgfsys@useobject{currentmarker}{}%
\end{pgfscope}%
\begin{pgfscope}%
\pgfsys@transformshift{1.367203in}{0.340552in}%
\pgfsys@useobject{currentmarker}{}%
\end{pgfscope}%
\begin{pgfscope}%
\pgfsys@transformshift{1.440566in}{0.438390in}%
\pgfsys@useobject{currentmarker}{}%
\end{pgfscope}%
\begin{pgfscope}%
\pgfsys@transformshift{1.513930in}{0.579564in}%
\pgfsys@useobject{currentmarker}{}%
\end{pgfscope}%
\begin{pgfscope}%
\pgfsys@transformshift{1.587293in}{0.388238in}%
\pgfsys@useobject{currentmarker}{}%
\end{pgfscope}%
\begin{pgfscope}%
\pgfsys@transformshift{1.660656in}{0.647151in}%
\pgfsys@useobject{currentmarker}{}%
\end{pgfscope}%
\begin{pgfscope}%
\pgfsys@transformshift{1.734020in}{0.542366in}%
\pgfsys@useobject{currentmarker}{}%
\end{pgfscope}%
\begin{pgfscope}%
\pgfsys@transformshift{1.807383in}{0.355051in}%
\pgfsys@useobject{currentmarker}{}%
\end{pgfscope}%
\begin{pgfscope}%
\pgfsys@transformshift{1.880747in}{0.743463in}%
\pgfsys@useobject{currentmarker}{}%
\end{pgfscope}%
\begin{pgfscope}%
\pgfsys@transformshift{1.954110in}{0.520286in}%
\pgfsys@useobject{currentmarker}{}%
\end{pgfscope}%
\begin{pgfscope}%
\pgfsys@transformshift{2.027473in}{0.964618in}%
\pgfsys@useobject{currentmarker}{}%
\end{pgfscope}%
\begin{pgfscope}%
\pgfsys@transformshift{2.100837in}{1.083540in}%
\pgfsys@useobject{currentmarker}{}%
\end{pgfscope}%
\begin{pgfscope}%
\pgfsys@transformshift{2.174200in}{0.557887in}%
\pgfsys@useobject{currentmarker}{}%
\end{pgfscope}%
\begin{pgfscope}%
\pgfsys@transformshift{2.247563in}{0.751655in}%
\pgfsys@useobject{currentmarker}{}%
\end{pgfscope}%
\begin{pgfscope}%
\pgfsys@transformshift{2.320927in}{1.025445in}%
\pgfsys@useobject{currentmarker}{}%
\end{pgfscope}%
\begin{pgfscope}%
\pgfsys@transformshift{2.394290in}{0.754219in}%
\pgfsys@useobject{currentmarker}{}%
\end{pgfscope}%
\begin{pgfscope}%
\pgfsys@transformshift{2.467654in}{1.229279in}%
\pgfsys@useobject{currentmarker}{}%
\end{pgfscope}%
\begin{pgfscope}%
\pgfsys@transformshift{2.541017in}{0.403555in}%
\pgfsys@useobject{currentmarker}{}%
\end{pgfscope}%
\begin{pgfscope}%
\pgfsys@transformshift{2.614380in}{1.196143in}%
\pgfsys@useobject{currentmarker}{}%
\end{pgfscope}%
\begin{pgfscope}%
\pgfsys@transformshift{2.687744in}{0.497258in}%
\pgfsys@useobject{currentmarker}{}%
\end{pgfscope}%
\begin{pgfscope}%
\pgfsys@transformshift{2.761107in}{0.996269in}%
\pgfsys@useobject{currentmarker}{}%
\end{pgfscope}%
\begin{pgfscope}%
\pgfsys@transformshift{2.834471in}{1.196143in}%
\pgfsys@useobject{currentmarker}{}%
\end{pgfscope}%
\begin{pgfscope}%
\pgfsys@transformshift{2.907834in}{0.849613in}%
\pgfsys@useobject{currentmarker}{}%
\end{pgfscope}%
\begin{pgfscope}%
\pgfsys@transformshift{2.981197in}{0.437872in}%
\pgfsys@useobject{currentmarker}{}%
\end{pgfscope}%
\begin{pgfscope}%
\pgfsys@transformshift{3.054561in}{0.937068in}%
\pgfsys@useobject{currentmarker}{}%
\end{pgfscope}%
\begin{pgfscope}%
\pgfsys@transformshift{3.127924in}{0.341036in}%
\pgfsys@useobject{currentmarker}{}%
\end{pgfscope}%
\begin{pgfscope}%
\pgfsys@transformshift{3.201288in}{0.310796in}%
\pgfsys@useobject{currentmarker}{}%
\end{pgfscope}%
\begin{pgfscope}%
\pgfsys@transformshift{3.274651in}{0.667747in}%
\pgfsys@useobject{currentmarker}{}%
\end{pgfscope}%
\begin{pgfscope}%
\pgfsys@transformshift{3.348014in}{0.552721in}%
\pgfsys@useobject{currentmarker}{}%
\end{pgfscope}%
\begin{pgfscope}%
\pgfsys@transformshift{3.421378in}{0.345899in}%
\pgfsys@useobject{currentmarker}{}%
\end{pgfscope}%
\begin{pgfscope}%
\pgfsys@transformshift{3.494741in}{0.920449in}%
\pgfsys@useobject{currentmarker}{}%
\end{pgfscope}%
\begin{pgfscope}%
\pgfsys@transformshift{3.568105in}{1.196143in}%
\pgfsys@useobject{currentmarker}{}%
\end{pgfscope}%
\begin{pgfscope}%
\pgfsys@transformshift{3.641468in}{0.722522in}%
\pgfsys@useobject{currentmarker}{}%
\end{pgfscope}%
\begin{pgfscope}%
\pgfsys@transformshift{3.714831in}{0.909456in}%
\pgfsys@useobject{currentmarker}{}%
\end{pgfscope}%
\begin{pgfscope}%
\pgfsys@transformshift{3.788195in}{0.389675in}%
\pgfsys@useobject{currentmarker}{}%
\end{pgfscope}%
\begin{pgfscope}%
\pgfsys@transformshift{3.861558in}{1.196143in}%
\pgfsys@useobject{currentmarker}{}%
\end{pgfscope}%
\begin{pgfscope}%
\pgfsys@transformshift{3.934921in}{0.819679in}%
\pgfsys@useobject{currentmarker}{}%
\end{pgfscope}%
\begin{pgfscope}%
\pgfsys@transformshift{4.008285in}{0.813074in}%
\pgfsys@useobject{currentmarker}{}%
\end{pgfscope}%
\begin{pgfscope}%
\pgfsys@transformshift{4.081648in}{1.124307in}%
\pgfsys@useobject{currentmarker}{}%
\end{pgfscope}%
\end{pgfscope}%
\begin{pgfscope}%
\pgfsetrectcap%
\pgfsetmiterjoin%
\pgfsetlinewidth{0.803000pt}%
\definecolor{currentstroke}{rgb}{0.000000,0.000000,0.000000}%
\pgfsetstrokecolor{currentstroke}%
\pgfsetdash{}{0pt}%
\pgfpathmoveto{\pgfqpoint{0.307102in}{0.195889in}}%
\pgfpathlineto{\pgfqpoint{0.307102in}{2.620314in}}%
\pgfusepath{stroke}%
\end{pgfscope}%
\begin{pgfscope}%
\pgfsetrectcap%
\pgfsetmiterjoin%
\pgfsetlinewidth{0.803000pt}%
\definecolor{currentstroke}{rgb}{0.000000,0.000000,0.000000}%
\pgfsetstrokecolor{currentstroke}%
\pgfsetdash{}{0pt}%
\pgfpathmoveto{\pgfqpoint{4.261389in}{0.195889in}}%
\pgfpathlineto{\pgfqpoint{4.261389in}{2.620314in}}%
\pgfusepath{stroke}%
\end{pgfscope}%
\begin{pgfscope}%
\pgfsetrectcap%
\pgfsetmiterjoin%
\pgfsetlinewidth{0.803000pt}%
\definecolor{currentstroke}{rgb}{0.000000,0.000000,0.000000}%
\pgfsetstrokecolor{currentstroke}%
\pgfsetdash{}{0pt}%
\pgfpathmoveto{\pgfqpoint{0.307102in}{0.195889in}}%
\pgfpathlineto{\pgfqpoint{4.261389in}{0.195889in}}%
\pgfusepath{stroke}%
\end{pgfscope}%
\begin{pgfscope}%
\pgfsetrectcap%
\pgfsetmiterjoin%
\pgfsetlinewidth{0.803000pt}%
\definecolor{currentstroke}{rgb}{0.000000,0.000000,0.000000}%
\pgfsetstrokecolor{currentstroke}%
\pgfsetdash{}{0pt}%
\pgfpathmoveto{\pgfqpoint{0.307102in}{2.620314in}}%
\pgfpathlineto{\pgfqpoint{4.261389in}{2.620314in}}%
\pgfusepath{stroke}%
\end{pgfscope}%
\begin{pgfscope}%
\pgfsetbuttcap%
\pgfsetmiterjoin%
\definecolor{currentfill}{rgb}{1.000000,1.000000,1.000000}%
\pgfsetfillcolor{currentfill}%
\pgfsetfillopacity{0.800000}%
\pgfsetlinewidth{1.003750pt}%
\definecolor{currentstroke}{rgb}{0.800000,0.800000,0.800000}%
\pgfsetstrokecolor{currentstroke}%
\pgfsetstrokeopacity{0.800000}%
\pgfsetdash{}{0pt}%
\pgfpathmoveto{\pgfqpoint{3.441277in}{2.031426in}}%
\pgfpathlineto{\pgfqpoint{4.183611in}{2.031426in}}%
\pgfpathquadraticcurveto{\pgfqpoint{4.205833in}{2.031426in}}{\pgfqpoint{4.205833in}{2.053648in}}%
\pgfpathlineto{\pgfqpoint{4.205833in}{2.542536in}}%
\pgfpathquadraticcurveto{\pgfqpoint{4.205833in}{2.564759in}}{\pgfqpoint{4.183611in}{2.564759in}}%
\pgfpathlineto{\pgfqpoint{3.441277in}{2.564759in}}%
\pgfpathquadraticcurveto{\pgfqpoint{3.419055in}{2.564759in}}{\pgfqpoint{3.419055in}{2.542536in}}%
\pgfpathlineto{\pgfqpoint{3.419055in}{2.053648in}}%
\pgfpathquadraticcurveto{\pgfqpoint{3.419055in}{2.031426in}}{\pgfqpoint{3.441277in}{2.031426in}}%
\pgfpathclose%
\pgfusepath{stroke,fill}%
\end{pgfscope}%
\begin{pgfscope}%
\pgfsetrectcap%
\pgfsetroundjoin%
\pgfsetlinewidth{1.003750pt}%
\definecolor{currentstroke}{rgb}{0.121569,0.466667,0.705882}%
\pgfsetstrokecolor{currentstroke}%
\pgfsetdash{}{0pt}%
\pgfpathmoveto{\pgfqpoint{3.463500in}{2.475870in}}%
\pgfpathlineto{\pgfqpoint{3.685722in}{2.475870in}}%
\pgfusepath{stroke}%
\end{pgfscope}%
\begin{pgfscope}%
\pgfsetbuttcap%
\pgfsetroundjoin%
\definecolor{currentfill}{rgb}{0.121569,0.466667,0.705882}%
\pgfsetfillcolor{currentfill}%
\pgfsetlinewidth{1.003750pt}%
\definecolor{currentstroke}{rgb}{0.121569,0.466667,0.705882}%
\pgfsetstrokecolor{currentstroke}%
\pgfsetdash{}{0pt}%
\pgfsys@defobject{currentmarker}{\pgfqpoint{-0.013889in}{-0.013889in}}{\pgfqpoint{0.013889in}{0.013889in}}{%
\pgfpathmoveto{\pgfqpoint{0.000000in}{-0.013889in}}%
\pgfpathcurveto{\pgfqpoint{0.003683in}{-0.013889in}}{\pgfqpoint{0.007216in}{-0.012425in}}{\pgfqpoint{0.009821in}{-0.009821in}}%
\pgfpathcurveto{\pgfqpoint{0.012425in}{-0.007216in}}{\pgfqpoint{0.013889in}{-0.003683in}}{\pgfqpoint{0.013889in}{0.000000in}}%
\pgfpathcurveto{\pgfqpoint{0.013889in}{0.003683in}}{\pgfqpoint{0.012425in}{0.007216in}}{\pgfqpoint{0.009821in}{0.009821in}}%
\pgfpathcurveto{\pgfqpoint{0.007216in}{0.012425in}}{\pgfqpoint{0.003683in}{0.013889in}}{\pgfqpoint{0.000000in}{0.013889in}}%
\pgfpathcurveto{\pgfqpoint{-0.003683in}{0.013889in}}{\pgfqpoint{-0.007216in}{0.012425in}}{\pgfqpoint{-0.009821in}{0.009821in}}%
\pgfpathcurveto{\pgfqpoint{-0.012425in}{0.007216in}}{\pgfqpoint{-0.013889in}{0.003683in}}{\pgfqpoint{-0.013889in}{0.000000in}}%
\pgfpathcurveto{\pgfqpoint{-0.013889in}{-0.003683in}}{\pgfqpoint{-0.012425in}{-0.007216in}}{\pgfqpoint{-0.009821in}{-0.009821in}}%
\pgfpathcurveto{\pgfqpoint{-0.007216in}{-0.012425in}}{\pgfqpoint{-0.003683in}{-0.013889in}}{\pgfqpoint{0.000000in}{-0.013889in}}%
\pgfpathclose%
\pgfusepath{stroke,fill}%
}%
\begin{pgfscope}%
\pgfsys@transformshift{3.574611in}{2.475870in}%
\pgfsys@useobject{currentmarker}{}%
\end{pgfscope}%
\end{pgfscope}%
\begin{pgfscope}%
\pgftext[x=3.774611in,y=2.436981in,left,base]{\rmfamily\fontsize{8.000000}{9.600000}\selectfont [3, 3, 3]}%
\end{pgfscope}%
\begin{pgfscope}%
\pgfsetrectcap%
\pgfsetroundjoin%
\pgfsetlinewidth{1.003750pt}%
\definecolor{currentstroke}{rgb}{1.000000,0.498039,0.054902}%
\pgfsetstrokecolor{currentstroke}%
\pgfsetdash{}{0pt}%
\pgfpathmoveto{\pgfqpoint{3.463500in}{2.309203in}}%
\pgfpathlineto{\pgfqpoint{3.685722in}{2.309203in}}%
\pgfusepath{stroke}%
\end{pgfscope}%
\begin{pgfscope}%
\pgfsetbuttcap%
\pgfsetbeveljoin%
\definecolor{currentfill}{rgb}{1.000000,0.498039,0.054902}%
\pgfsetfillcolor{currentfill}%
\pgfsetlinewidth{1.003750pt}%
\definecolor{currentstroke}{rgb}{1.000000,0.498039,0.054902}%
\pgfsetstrokecolor{currentstroke}%
\pgfsetdash{}{0pt}%
\pgfsys@defobject{currentmarker}{\pgfqpoint{-0.026418in}{-0.022473in}}{\pgfqpoint{0.026418in}{0.027778in}}{%
\pgfpathmoveto{\pgfqpoint{0.000000in}{0.027778in}}%
\pgfpathlineto{\pgfqpoint{-0.006236in}{0.008584in}}%
\pgfpathlineto{\pgfqpoint{-0.026418in}{0.008584in}}%
\pgfpathlineto{\pgfqpoint{-0.010091in}{-0.003279in}}%
\pgfpathlineto{\pgfqpoint{-0.016327in}{-0.022473in}}%
\pgfpathlineto{\pgfqpoint{-0.000000in}{-0.010610in}}%
\pgfpathlineto{\pgfqpoint{0.016327in}{-0.022473in}}%
\pgfpathlineto{\pgfqpoint{0.010091in}{-0.003279in}}%
\pgfpathlineto{\pgfqpoint{0.026418in}{0.008584in}}%
\pgfpathlineto{\pgfqpoint{0.006236in}{0.008584in}}%
\pgfpathclose%
\pgfusepath{stroke,fill}%
}%
\begin{pgfscope}%
\pgfsys@transformshift{3.574611in}{2.309203in}%
\pgfsys@useobject{currentmarker}{}%
\end{pgfscope}%
\end{pgfscope}%
\begin{pgfscope}%
\pgftext[x=3.774611in,y=2.270314in,left,base]{\rmfamily\fontsize{8.000000}{9.600000}\selectfont [4, 3, 2]}%
\end{pgfscope}%
\begin{pgfscope}%
\pgfsetrectcap%
\pgfsetroundjoin%
\pgfsetlinewidth{1.003750pt}%
\definecolor{currentstroke}{rgb}{0.172549,0.627451,0.172549}%
\pgfsetstrokecolor{currentstroke}%
\pgfsetdash{}{0pt}%
\pgfpathmoveto{\pgfqpoint{3.463500in}{2.142537in}}%
\pgfpathlineto{\pgfqpoint{3.685722in}{2.142537in}}%
\pgfusepath{stroke}%
\end{pgfscope}%
\begin{pgfscope}%
\pgfsetbuttcap%
\pgfsetmiterjoin%
\definecolor{currentfill}{rgb}{0.172549,0.627451,0.172549}%
\pgfsetfillcolor{currentfill}%
\pgfsetlinewidth{1.003750pt}%
\definecolor{currentstroke}{rgb}{0.172549,0.627451,0.172549}%
\pgfsetstrokecolor{currentstroke}%
\pgfsetdash{}{0pt}%
\pgfsys@defobject{currentmarker}{\pgfqpoint{-0.027778in}{-0.027778in}}{\pgfqpoint{0.027778in}{0.027778in}}{%
\pgfpathmoveto{\pgfqpoint{0.000000in}{0.027778in}}%
\pgfpathlineto{\pgfqpoint{-0.027778in}{-0.027778in}}%
\pgfpathlineto{\pgfqpoint{0.027778in}{-0.027778in}}%
\pgfpathclose%
\pgfusepath{stroke,fill}%
}%
\begin{pgfscope}%
\pgfsys@transformshift{3.574611in}{2.142537in}%
\pgfsys@useobject{currentmarker}{}%
\end{pgfscope}%
\end{pgfscope}%
\begin{pgfscope}%
\pgftext[x=3.774611in,y=2.103648in,left,base]{\rmfamily\fontsize{8.000000}{9.600000}\selectfont [5, 3, 1]}%
\end{pgfscope}%
\end{pgfpicture}%
\makeatother%
\endgroup%

%% file: shapley-aaai20-zhao.bbl
\begin{thebibliography}{}

\bibitem[\protect\citeauthoryear{Aumann and Hart}{1992}]{Aumann1992gt}
Aumann, R., and Hart, S., eds.
\newblock 1992.
\newblock {\em Handbook of Game Theory with Economic Applications}, volume~1.
\newblock Elsevier, 1 edition.

\bibitem[\protect\citeauthoryear{Bachrach \bgroup et al\mbox.\egroup
  }{2012}]{Bachrach2012}
Bachrach, Y.; Meir, R.; Feldman, M.; and Tennenholtz, M.
\newblock 2012.
\newblock Solving cooperative reliability games.
\newblock {\em CoRR} abs/1202.3700.

\bibitem[\protect\citeauthoryear{Chalkiadakis and
  Boutilier}{2004}]{Chalkiadakis2004}
Chalkiadakis, G., and Boutilier, C.
\newblock 2004.
\newblock Bayesian reinforcement learning for coalition formation under
  uncertainty.
\newblock In {\em Proceedings of the Third International Joint Conference on
  Autonomous Agents and Multiagent Systems - Volume 3}, AAMAS '04,  1090--1097.
\newblock Washington, DC, USA: IEEE Computer Society.

\bibitem[\protect\citeauthoryear{Chalkiadakis and
  Boutilier}{2012}]{Chalkiadakis2012}
Chalkiadakis, G., and Boutilier, C.
\newblock 2012.
\newblock Sequentially optimal repeated coalition formation under uncertainty.
\newblock {\em Autonomous Agents and Multi-Agent Systems} 24(3):441--484.

\bibitem[\protect\citeauthoryear{Chalkiadakis, Markakis, and
  Boutilier}{2007}]{Chalkiadakis2007}
Chalkiadakis, G.; Markakis, E.; and Boutilier, C.
\newblock 2007.
\newblock Coalition formation under uncertainty: Bargaining equilibria and the
  bayesian core stability concept.
\newblock In {\em Proceedings of the 6th International Joint Conference on
  Autonomous Agents and Multiagent Systems}, AAMAS '07,  64:1--64:8.
\newblock New York, NY, USA: ACM.

\bibitem[\protect\citeauthoryear{Charnes and Granot}{1976}]{Charnes1976}
Charnes, A., and Granot, D.
\newblock 1976.
\newblock Coalitional and chance-constrained solutions to n-person games. i:
  The prior satisficing nucleolus.
\newblock {\em SIAM Journal on Applied Mathematics} 31(2):358--367.

\bibitem[\protect\citeauthoryear{Charnes and Granot}{1977}]{Charnes1977}
Charnes, A., and Granot, D.
\newblock 1977.
\newblock Coalitional and chance-constrained solutions to n-person games, ii:
  Two-stage solutions.
\newblock {\em Operations Research} 25:1013--1019.

\bibitem[\protect\citeauthoryear{Clarke}{1971}]{clarke1971multipart}
Clarke, E.~H.
\newblock 1971.
\newblock Multipart pricing of public goods.
\newblock {\em Public choice} 11(1):17--33.

\bibitem[\protect\citeauthoryear{Conitzer and
  Vidali}{2014}]{conitzer2014mechanism}
Conitzer, V., and Vidali, A.
\newblock 2014.
\newblock Mechanism design for scheduling with uncertain execution time.
\newblock In {\em {AAAI} Conference on Artificial Intelligence}.

\bibitem[\protect\citeauthoryear{Davis and Maschler}{1965}]{Davis1995kernel}
Davis, M., and Maschler, M.
\newblock 1965.
\newblock {The kernel of a cooperative game}.
\newblock {\em Naval Research Logistics Quarterly} 12(3):223--259.

\bibitem[\protect\citeauthoryear{Frank}{1969}]{frank1969shortest}
Frank, H.
\newblock 1969.
\newblock Shortest paths in probabilistic graphs.
\newblock {\em Operations Research} 17(4):583--599.

\bibitem[\protect\citeauthoryear{Groves}{1973}]{groves1973incentives}
Groves, T.
\newblock 1973.
\newblock Incentives in teams.
\newblock {\em Econometrica: Journal of the Econometric Society}  617--631.

\bibitem[\protect\citeauthoryear{Ieong and Shoham}{2008}]{Ieong2008}
Ieong, S., and Shoham, Y.
\newblock 2008.
\newblock Bayesian coalitional games.
\newblock In {\em Proceedings of the 23rd National Conference on Artificial
  Intelligence - Volume 1}, AAAI'08,  95--100.
\newblock AAAI Press.

\bibitem[\protect\citeauthoryear{Li and Conitzer}{2015}]{Li2015}
Li, Y., and Conitzer, V.
\newblock 2015.
\newblock Cooperative game solution concepts that maximize stability under
  noise.
\newblock In {\em Proceedings of the Twenty-Ninth AAAI Conference on Artificial
  Intelligence}, AAAI'15,  979--985.
\newblock AAAI Press.

\bibitem[\protect\citeauthoryear{Myerson}{2007}]{MYERSON2007}
Myerson, R.~B.
\newblock 2007.
\newblock Virtual utility and the core for games with incomplete information.
\newblock {\em Journal of Economic Theory} 136(1):260 -- 285.

\bibitem[\protect\citeauthoryear{Nisan and Ronen}{2001}]{Nisan2001}
Nisan, N., and Ronen, A.
\newblock 2001.
\newblock Algorithmic mechanism design.
\newblock {\em Games and Economic Behavior} 35(1-2):166--196.

\bibitem[\protect\citeauthoryear{Nisan \bgroup et al\mbox.\egroup
  }{2007}]{nisan2007algorithmic}
Nisan, N.; Roughgarden, T.; Tardos, E.; and Vazirani, V.~V.
\newblock 2007.
\newblock {\em Algorithmic game theory}, volume~1.
\newblock Cambridge University Press Cambridge.

\bibitem[\protect\citeauthoryear{Porter \bgroup et al\mbox.\egroup
  }{2008}]{porter2008fault}
Porter, R.; Ronen, A.; Shoham, Y.; and Tennenholtz, M.
\newblock 2008.
\newblock Fault tolerant mechanism design.
\newblock {\em Artificial Intelligence} 172(15):1783--1799.

\bibitem[\protect\citeauthoryear{Ramchurn \bgroup et al\mbox.\egroup
  }{2009}]{Ramchurn2009Trust}
Ramchurn, S.~D.; Mezzetti, C.; Giovannucci, A.; Rodriguez-Aguilar, J.~A.; Dash,
  R.~K.; and Jennings, N.~R.
\newblock 2009.
\newblock Trust-based mechanisms for robust and efficient task allocation in
  the presence of execution uncertainty.
\newblock {\em J. Artif. Int. Res.} 35(1):119--159.

\bibitem[\protect\citeauthoryear{Shapley}{1953}]{shapley1953value}
Shapley, L.~S.
\newblock 1953.
\newblock A value for n-person games.
\newblock {\em Contributions to the Theory of Games} 2(28):307--317.

\bibitem[\protect\citeauthoryear{Stein \bgroup et al\mbox.\egroup
  }{2011}]{stein2011algorithms}
Stein, S.; Gerding, E.~H.; Rogers, A.; Larson, K.; and Jennings, N.~R.
\newblock 2011.
\newblock Algorithms and mechanisms for procuring services with uncertain
  durations using redundancy.
\newblock {\em Artificial Intelligence} 175(14):2021--2060.

\bibitem[\protect\citeauthoryear{Suijs and Borm}{1999}]{Suijs1999GEB}
Suijs, J., and Borm, P.
\newblock 1999.
\newblock Stochastic cooperative games: Superadditivity, convexity, and
  certainty equivalents.
\newblock {\em Games and Economic Behavior} 27(2):331--345.

\bibitem[\protect\citeauthoryear{Suijs \bgroup et al\mbox.\egroup
  }{1999}]{Suijs1999}
Suijs, J.; Borm, P.; Waegenaere, A.~D.; and Tijs, S.
\newblock 1999.
\newblock Cooperative games with stochastic payoffs.
\newblock {\em European Journal of Operational Research} 113(1):193 -- 205.

\bibitem[\protect\citeauthoryear{Vickrey}{1961}]{vickrey1961counterspeculation}
Vickrey, W.
\newblock 1961.
\newblock Counterspeculation, auctions, and competitive sealed tenders.
\newblock {\em The Journal of Finance} 16(1):8--37.

\bibitem[\protect\citeauthoryear{Zhao, Ramchurn, and
  Jennings}{2016}]{zhao2016fault}
Zhao, D.; Ramchurn, S.~D.; and Jennings, N.~R.
\newblock 2016.
\newblock Fault tolerant mechanism design for general task allocation.
\newblock In {\em Proceedings of the 2016 International Conference on
  Autonomous Agents \& Multiagent Systems},  323--331.
\newblock International Foundation for Autonomous Agents and Multiagent
  Systems.

\end{thebibliography}
